\documentclass[11pt]{scrartcl}

\usepackage[utf8]{inputenc}
\usepackage{amssymb}
\usepackage{amsmath}
\usepackage{fullpage}
\usepackage{enumitem}
\usepackage{hyperref}
\usepackage{amsfonts}
\usepackage{bbm}
\usepackage{amsthm}
\usepackage{newtxmath}
\usepackage{graphicx}
\usepackage{float}
\usepackage{caption}
\usepackage{subcaption}
\usepackage{listings}
\usepackage{booktabs}
\usepackage{Sweave}
\usepackage[round, semicolon, sort]{natbib}
\usepackage{bm}
\usepackage{multirow}
\usepackage[font={small, it}]{caption}
\usepackage{booktabs}

\usepackage[linesnumbered,ruled,vlined]{algorithm2e}
\usepackage{algpseudocode}  
\usepackage{hyperref}
\hypersetup{
    colorlinks=true,
    linkcolor=blue,
    filecolor=blue,      
    urlcolor=blue,
    citecolor=blue,
}

\newtheorem{theorem}{Theorem}

\newtheorem{lemma}{Lemma}

\newcommand{\cA}{\mathcal{A}}
\newcommand{\cB}{\mathcal{B}}

\newcommand{\cK}{\mathcal{K}}
\newcommand{\cM}{\mathcal{M}}
\newcommand{\cN}{\mathcal{N}}

\newcommand{\cT}{\mathcal{T}}

\newcommand{\ba}{\bm{a}}
\newcommand{\bb}{\bm{b}}
\newcommand{\be}{\bm{e}}
\newcommand{\bh}{\bm{h}}
\newcommand{\bt}{\bm{t}}

\newcommand{\bv}{\bm{v}}

\newcommand{\bx}{\bm{x}}

\newcommand{\bC}{\bm{C}}

\newcommand{\bK}{\bm{K}}
\newcommand{\bU}{\bm{U}}
\newcommand{\bV}{\bm{V}}
\newcommand{\bX}{\bm{X}}
\newcommand{\bY}{\bm{Y}}
\newcommand{\bZ}{\bm{Z}}

\newcommand{\bbt}{\bm{\beta}}

\newcommand{\cky}{\check{y}}

\title{A Non-Parametric Box-Cox Approach to Robustifying High-Dimensional Linear Hypothesis Testing}
\author{He Zhou and Hui Zou\footnote{ Corresponding author at: 313 Ford Hall, School of Statistics, 224 Church Street SE, Minneapolis, MN 55455, USA.
E-mail address: \href{zouxx019@umn.edu}{zouxx019@umn.edu}. 
}
\\
\large\textit{School of Statistics, University of Minnesota, United States}}
\date{}

\begin{document}

\maketitle
\begin{abstract}
The mainstream theory of hypothesis testing in high-dimensional regression typically assumes the underlying true model is a low-dimensional linear regression model, yet the Box-Cox transformation is a regression technique commonly used to mitigate anomalies like non-additivity and heteroscedasticity. This paper introduces a more flexible framework, the non-parametric Box-Cox model with unspecified transformation, to address model mis-specification in high-dimensional linear hypothesis testing while preserving the interpretation of regression coefficients. Model estimation and computation in high dimensions poses challenges beyond traditional sparse penalization methods. We propose the constrained partial penalized composite probit regression method for sparse estimation and investigate its statistical properties. Additionally, we present a computationally efficient algorithm using augmented Lagrangian and coordinate majorization descent for solving regularization problems with folded concave penalization and linear constraints. For testing linear hypotheses, we propose the partial penalized composite likelihood ratio test, score test and Wald test, and show that their limiting distributions under null and local alternatives follow generalized chi-squared distributions with the same degrees of freedom and noncentral parameter. Extensive simulation studies are conducted to examine the finite sample performance of the proposed tests. Our analysis of supermarket data illustrates potential discrepancies between our testing procedures and standard high-dimensional methods, highlighting the importance of our robustified approach.
\end{abstract}

\vspace{0.1cm}
\noindent
{\textbf{Keywords:}
High-dimensional testing, Linear hypothesis, Box-Cox model,  Non-parametric transformation, Composite estimation, Composite likelihood ratio test, Score test, Wald test} 
\vspace{0.1cm}

\section{Introduction}

In a high-dimensional regression problem, the number of covariates might diverge with the sample size or even exceed the sample size.  
In the past few decades, a large number of papers have been devoted to developing sparse estimation methods and variable selection techniques for high-dimensional regression problems. For example, the LASSO \citep{tibshirani1996regression} and the concave penalization \citep{fan2001variable} provide two mainstream sparse penalization methods. For a comprehensive understanding of this field, readers can be referred to \cite{FanLiZhangZou2020} for a comprehensive treatment of this topic. Recent advancements in statistical inference have focused on hypothesis testing concerning sparse estimates on high-dimensional linear regression models. Works by \cite{lockhart2014significance},  \cite{taylor2014post}, \cite{lee2016exact}, etc., start from proposing inference tools for the LASSO \citep{tibshirani1996regression} estimate. However, owing to the non-negligible estimation bias inherent in the LASSO estimates, \cite{zhang2014confidence}, \cite{javanmard2014confidence}, and \cite{van2014asymptotically} have proposed to debias the LASSO-type estimators before making further inferences. The development of those debiased inference methods motivated an increase interest in testing under generalized linear models (GLMs, \citealt{mccullagh1989generalized}). For instance, \cite{guo2021inference} and \cite{cai2021statistical} proposed the bias-corrected inferences for high-dimensional binary GLMs; \cite{ren2015asymptotic} studied the Gaussian graphical models; \cite{fang2020test} considered tests specific to longitudinal data analysis; etc. A comprehensive overview of recent advancements in statistical inference for high-dimensional regression models is provided in the survey by \cite{cai2023statistical}. 

Despite the extensive efforts on hypothesis testing in high-dimensional (generalized) linear regression models, it remains unknown whether those methods and theoretical results can be extended to more general regression models, such as the Box-Cox transformed regression model \citep{box1964analysis}. On the other hand, specifying a certain (generalized) linear regression model as the true underlying model might lead to the issues of model mis-specification in many practical applications. Let $Y$ be the response variable and $(x_1,\ldots,x_p)^{T}$ be the vector of covariates. Recent research efforts in high-dimensional statistical inferences often start with a normal linear assumption, ${Y} = \sum_{j \in \cA} x_j \beta_j +\varepsilon $ (where $\cA$ denotes a small subset of important variables), this assumption, while convenient for theoretical analysis, may not accurately capture the complexities of the true underlying model. Consequently, conducting further inferences and tests based on the estimations obtained from mis-specified models will lead to unstable and untrustworthy conclusions. 

In light of the challenges in conducting hypothesis testings under more complicated models in high-dimensional statistical inferences, it requires extra efforts to enhance the robustness of statistical inference procedures without compromising theoretical integrity. Our paper aims to partially address this issue by focusing on hypothesis testing within the framework of a non-parametric transformed linear model:
\begin{equation}
	g(Y) = \bx^{T}\bbt + \varepsilon,\quad\varepsilon\sim N(0,1),
	\label{eq:non-parametric_Box-Cox_model}
\end{equation}
where $g(\cdot)$ is an \emph{unspecified} monotone increasing function, and the error variance is assumed to be one as a scaling parameter can be absorbed into the monotone function $g(\cdot)$. This model \eqref{eq:non-parametric_Box-Cox_model}, named as the \emph{non-parametric Box-Cox model} \citep{zhou2023nonparametric}, is a non-parametric generalization of the well-known Box-Cox regression model \citep{box1964analysis} with the similar goal of achieving additivity, normality and homoscedasticity via data transformation. \cite{zhou2023nonparametric} proposed a two-step methodology for the sparse estimation and prediction of this model under ultra-high dimensional settings (i.e., $p$ grows exponentially with $n$). This model demonstrates \emph{robustness against model mis-specification} from two perspectives: Firstly, a transformation applied to the response will fix the anomalies such as non-additivity and heteroscedasticity that violate the normal linear assumption; Secondly, unlike the prevalent Box-Cox power transformation, our approach employs a non-parametric transformation as any pre-chosen parametric form may suffer from mis-specification for a given application. To our knowledge, in ultra-high dimensional settings, obtaining optimal estimates of regression coefficients poses a significant challenge, even for the classic Box-Cox power transformation model, let alone for the non-parametric transformed model. However, by leveraging the \emph{Composite Probit Regression} method introduced by \cite{zhou2023nonparametric}, we are able to obtain the sparse estimators with strong oracle properties \citep{fan2014strong}. By developing testing procedures based on such optimal estimator, we anticipate achieving robust power performance even in the presence of local alternatives. It's noteworthy that this model only performs transformation on the response $Y$, leaving the covariates $\bx$ untouched. We focus on this generalized model for one critical reason: extending the model assumptions should not change the interpretation of the regression coefficients $\bbt$. Preserving this interpretability is essential when conducting hypotheses for $\bbt$ or linear contrasts of $\bbt$. Motivated by these considerations, we propose to develop testing procedures that are robust against any unknown transformation on the response. These procedures are designed to maintain the interpretation of the coefficients while theoretically ensuring robust power under local alternatives.

To complete the story, we consider testing linear hypothesis 
\begin{equation}
	H_0:~\bC\bbt_{0,\cM}=\bt
        \label{eq:linear_hypothesis}
\end{equation}
for the non-parametric Box-Cox model \eqref{eq:non-parametric_Box-Cox_model}, where $\bbt_{0,\cM}$ is a subvector of the true regression coefficients $\bbt_{0}$. In high-dimensional settings, the number of covariates $p$ can diverge with the sample size or exceed the sample size, while the cardinality of subset $\cM$ is assumed to be much smaller to the sample size. By setting matrix $\bC$ to have only one row, it considers the special class of hypotheses for individual regression coefficients or related one-dimensional functionals. Large amount of work has been devoted to this class of hypotheses under high-dimensional (generalized) linear models \citep{van2014asymptotically, javanmard2014confidence, lockhart2014significance, ren2015asymptotic, guo2021inference, cai2021optimal}. With a general $\bC$, the focus is on dealing with simultaneous inference for multiple regression coefficients. When $\bC$ is set to be the identity matrix and $\bt=\bm{0}$, it reduces to a special class of hypotheses $\bbt_{0,\cM}=\bm{0}$. \cite{wang2013partial} proposed a penalized likelihood ratio statistic for this class requiring $p=o(n^{1/5})$; \cite{ning2017general} considered a decorrelated score test for generic penalized M-estimators; \cite{fang2017testing} developed decorrelated Wald, score and partial likelihood ratio tests for proportional hazards models; \cite{zhang2017simultaneous} proposed a bootstrap-assisted procedure for high-dimensional linear models. With general $\bC$ and $\bt$, \cite{shi2019linear} proposed Wald, score and likelihood ratio tests for GLMs under ultra-high dimensional settings (i.e., $p$ grows exponentially with $n$). Building upon this work, we generalize the constrained partial penalized testing framework of \cite{shi2019linear} and develop the Wald, score and likelihood ratio type of tests for the non-parametric Box-Cox model under ultra-high dimensional settings. 

It is challenging to generalize the Wald, score and likelihood ratio tests for GLMs \citep{shi2019linear} to the non-parametric Box-Cox model under ultra-high dimensionality. 
Firstly, the choice of sparse estimation for the regression coefficients $\bbt$ is critical. The LASSO-type estimation has the bias issue, while the nonconcave penalized estimation \citep{fan2001variable} does not. 
Secondly, even under lower-dimensional settings, estimating the regression coefficients on top of the estimated transformation will introduces considerate  variability and will weaken the performance of the estimation. To solve the first two challenges, we follow the folded-concave penalized composite probit regression technique invented by \cite{zhou2023nonparametric} for sparse estimation with optimal theoretical guarantee. 
Thirdly, when conducting tests associated with the penalized estimator $\hat{\bbt}$ (e.g., the Wald test), hypotheses with local alternatives might not have the desired power because of the minimal signal conditions for regression coefficients under nonconcave penalization \citep{fan2004nonconcave, fan2011nonconcave}. For example, consider the null hypothesis $H_0:~\bbt_{0,\cM}=\bm{0}$ and its local alternative $H_a:~\bbt_{0,\cM}=\bm{h}_n$ for some sequence of $\bm{h}_n$ such that $\Vert\bm{h}_n\Vert_2 \ll \lambda_n$ (where $\lambda_n$ is the tuning parameter for the nonconcave penalization). The estimation and variable selection on ${\cM}$ will not have theoretical guarantee due to the violation of the minimal signal condition. 
Fourthly, the estimation method using composite probit regression falls into the context of composite likelihood inference. Consequently, there exists a discrepancy between the sensitivity matrix and the variability matrix, and therefore the Fisher information matrix needs to be substituted by the Godamebe information matrix \citep{godambe1960optimum}, also referred to as the \emph{sandwich information matrix}. Extensive efforts have been dedicated to the composite likelihood inferences, and interested readers can find a comprehensive overview in the survey paper by \cite{varin2011overview}. Following the typical context of composite likelihood inferences, one can construct composite likelihood versions of the Wald and score statistics based on the Godambe information \citep{molenberghs2005models}, which have the usual asymptotic chi-squared distribution. However, they are not asymptotically equivalent to the composite likelihood likelihood ratio statistic, which follows a non-standard asymptotic chi-squared distribution \citep{kent1982robust}. Numerous efforts have been made to adjust the composite likelihood ratio statistic to make it has the standard asymptotic chi-squared distribution \citep{geys1999pseudolikelihood, rotnitzky1990hypothesis, zou2008one, pace2011adjusting}. In this paper, we choose to construct a novel version of the Wald and score test statistics that are asymptotically equivalent to the composite likelihood ratio statistic.
Last but not least, when developing the score and likelihood ratio type of tests, we encounter the challenge of estimating regression parameters under the null hypothesis, which invokes non-convex optimization with linear constraints. Based on our knowledge, there is a lack of computationally efficient algorithms that can produce sparse solutions for nonconcave penalized optimization problems with constraints. While the ADMM algorithm proposed by \cite{shi2019linear} offers a potential solution, it suffers from computational inefficiency and accuracy issues, as it requires to use a Newton-type algorithm for solving a convex sub-problem with high-dimensional variables.

We briefly summarize our {contributions} as follows. 
Firstly, we extend the scope of high-dimensional hypothesis testing by considering a more general form of the underlying true model. In contrast to the typical linear regression model, we apply an unspecified transformation on the response to remedy the common anomalies such as non-additivity and heteroscedasticity that frequently occur in practical scenarios. 
Secondly, we combine the composite probit regression (CPR) method with the folded-concave partial penalization to develop our partial penalized composite probit regression estimators. Extending the work by \cite{zhou2023nonparametric}, we provide a comprehensive theoretical analysis, deriving convergence rates and limiting distributions for both constrained and unconstrained estimators in ultra-high dimensional settings.
Thirdly, we use these partial penalized composite probit regression estimators to formulate our partial penalized Wald, score and composite-likelihood-ratio test statistics and our Wald and score test statistics are novel contributions to the literature. . We establish the asymptotic equivalence of our partial penalized CPR tests and derive the asymptotic distributions of our test statistics under both the null hypothesis and local alternatives. One interesting observation is that the limiting distributions of our test statistics are generalized $\chi^2$ distributions, which can be degenerated to the usual $\chi^2$ distributions \citep{shi2019linear} when degenerating the composite probit regression to probit regression. 
Last but not least, we develop an algorithm for computing the folded-concave partial penalized CPR estimator with equality constraints. To deal with the equality constraints, we use the augmented lagrangian method \citep{boyd2011distributed}. Inspired by the efficacy of coordinate descent algorithm for generalized linear models \citep{friedman2010regularization}, we design a coordinate-majorization-descent algorithm to efficiently solve the lagrangian-augmented unconstrained optimization problem.

The rest of the paper is organized as follows. In section \ref{sec:NBC}, we provides more details about the non-parametric Box-Cox model and the composite probit regression method for estimation of regression coefficients. In section \ref{sec:partial_penalized_CPR}, we study the statistical properties of the partial penalized composite probit regression estimator with folded concave penalization. In section \ref{sec:partial_penalized_tests}, we give formal definition for our partial penalized Wald, score and composite-likelihood-ratio test statistics, and establish their limiting distributions, and show their equivalence. In section \ref{sec:algorithm}, we provide the detailed implementations of our proposed algorithms. Simulation studies are presented in section \ref{sec:simulations}. The proofs and addition numerical results are presented in the Appendix.

We introduce the following notations for the rest of the paper. For a vector $\bm{v}=(v_1,\dots,v_m)\in\mathbb{R}^{m}$ and a subset $\mathcal{J}\subseteq [1,\dots,m]$, denote $\bm{v}_{\mathcal{J}}$ as the sub-vector of $\bm{v}$ with indices in $\mathcal{J}$. Denote $\vert\bm{v}\vert=(\vert v_1\vert,\dots,\vert v_m\vert)$ and $\text{diag}(\bm{v})$ be the diagonal matrix with the $j$-th diagonal element being $v_j$. For a matrix $\bm{U}=(u_{ij})_{m\times n}$ and a subset of indices $\mathcal{T}\subseteq [1,\dots,n]$, denote $\bm{U}_{i,:}$ as the $i$-th row of $\bm{U}$ and $\bm{U}_{\mathcal{T}}$ as the sub-matrix of $\bm{U}$ consisting the columns with indices in $\mathcal{T}$. Let $\lambda_{\min}(\bm{U})$ and $\lambda_{\max}(\bm{U})$ be the smallest and largest eigenvalues of $\bm{U}$. We also introduce matrix norms induced by vector norms: Given $1\leq p\leq\infty$, $\Vert\bm{U}\Vert_{p}=\sup\{\Vert \bm{U}\bm{x}\Vert_{p}: \Vert\bm{x}\Vert_{p}=1\}$. In special cases, $\Vert\bm{U}\Vert_{1}=\max_{j}\sum_{i}\vert u_{ij}\vert$, $\Vert\bm{U}\Vert_{2}=\lambda_{\max}^{1/2}(\bm{U}^{T}\bm{U})$, and $\Vert\bm{U}\Vert_{\infty}=\max_{i}\sum_{j}\vert u_{ij}\vert$. Given $1\leq\alpha,\gamma\leq\infty$, $\Vert\bm{U}\Vert_{\alpha,\gamma}=\sup\{\Vert \bm{U}\bm{x}\Vert_{\gamma}: \Vert\bm{x}\Vert_{\alpha}=1\}$. Thus, $\Vert \bm{U}\bm{x}\Vert_{\gamma} \leq \Vert\bm{U}\Vert_{\alpha,\gamma}\Vert\bm{x}\Vert_{\alpha}$. In special cases, $\Vert\bm{U}\Vert_{2,\infty}=\max_{i}\Vert\bm{U}_{i,:}\Vert_2$; $\Vert\bm{U}\Vert_{1,2}=\max_{j}\Vert\bm{U}_{:,j}\Vert_2$

\section{Non-Parametric Box-Cox Model}\label{sec:NBC}

In this paper, we consider the following non-parametric Box-Cox model as the underlying true model:
\begin{equation*}
	g(Y) = \bx^{T}\bbt + \varepsilon,\quad\varepsilon\sim N(0,1),
\end{equation*}
where $g(\cdot)$ is an \emph{unspecified} monotone increasing function. Compared with the classical linear model assumption, our generalization is robust against model mis-specification while preserving the interpretation of the regression coefficients $\bbt$. Estimating the regression coefficients on top of an estimated transformation can significantly increase the variance of the estimator, thereby reducing the power performance when further conducting hypothesis testings. Therefore, we borrow the \emph{composite probit regression (CPR)} method proposed by \cite{zhou2023nonparametric} which forms the first step of their two-stage methodology and allows us to obtain the optimal estimation without knowing the transformation.

\subsection{Estimation: Composite Probit Regression}\label{subsec:CPR}
Consider estimating the parameter $\bbt$ in our non-parametric Box-Cox regression model based on $n$ i.i.d. observations $\{(\bx_i,y_i)\}_{i=1}^{n}$. 
Given a user-chosen threshold $y_0$, we have
\begin{equation}
	\mathbb{P}(Y\geq y_0\vert\bx) = \mathbb{P}(g(Y)\geq g(y_0)\vert\bx) = \mathbb{P}(\varepsilon\geq g(y_0) - \bx^{T}\bbt\vert\bx) = \Phi(-g(y_0) + \bx^{T}\bbt),
	\label{eq:model_derivation}
\end{equation}
where $\Phi(\cdot)$ is the cumulative density function (CDF) of standard normal distribution. 
Thus, by creating new response variables $\tilde{y}_i = I_{\{y_i\geq y_0\}}$, we have $\{(\tilde{y}_i,\bx_i)\}_{i=1}^{n}$
follow a probit regression model with intercept $-g(y_0)$ and regression coefficient $\bbt$. The choice of threshold values only affects the intercept term in the probit model but not the regression coefficient.
To borrow strength from multiple probit regression models induced from different thresholds, \cite{zhou2023nonparametric} proposed the composite probit regression method. Let $\{y_0^{(k)}, k=1, \cdots, K\}$ be the sequence of threshold values, $\tilde{y}_{ki}=I_{\{y_i\geq y_0^{(k)}\}}$, and $\bm{b}={(b_{1},\cdots,b_{K})}^{T}$. The \emph{composite probit likelihood} function for the estimation of parameter $\bbt$ is defined as:
\begin{equation}
    M_n(\bbt, \bb)
    :=
    \sum_{k=1}^{K}w_k\left\{\frac{1}{n}\sum_{i=1}^{n}
    \left[
        \tilde{y}_{ki}h(\bx_i^{T}\bbt-b_{k})+\log(1-\Phi(\bx_i^{T}\bbt-b_{k}))
    \right]\right\},\label{eq:composite_probit_likelihood}
\end{equation}
which is the weighted summation of log-likelihood functions of $K$ probit models with composite weights $w_k$'s and $h(\eta) = \log\left(\frac{\Phi(\eta)}{1-\Phi(\eta)}\right)$. The sparse estimation of $\bbt$ under high-dimensional cases can be obtained by applying sparse penalization to the composite likelihood function. It's worth noting that the estimation of $\bbt$ based on the composite probit regression method remains invariant to any transformation $g(\cdot)$ on the response and any choice of threshold values. Different transformation and choices of threshold values will only effect the interpretation of the intercepts $\bb$ but not the regression coefficients, which is of primary goal.

When implementing the composite probit regression estimation in practice, the sequence of thresholds and the composite weights need to be specified beforehand. \cite{zhou2023nonparametric} suggests to use the $5\%$, $10\%$, $15\%$, $\dots$, $95\%$ empirical percentiles of the response $\{y_i\}$'s and the equal composite weights, which amounts to $K=19$ and $w_1=\cdots=w_K=1/K$. 

\section{Partial Penalized Composite Probit Regression}\label{sec:partial_penalized_CPR}

In this section, we propose to apply a partial penalization on our composite probit regression model to accommodate the conflict between the minimal signal assumption and nonconcave penalization, as well as the power issue under local alternatives.  We then prove the rate of convergence and the asymptotic distribution of both the constrained and unconstrained partial penalized composite probit regression estimators. These estimators will serve as the basis for hypothesis testing, as detailed in Section \ref{sec:partial_penalized_tests}.

\subsection{Partial Penalized Composite Probit Regression Method}\label{subsec:partial_penalized_CPR_method}

Consider testing the following linear hypothesis:
\begin{equation}
	H_0:~\bC\bbt^*_{\cM}=\bt
    \label{eq:linear_hypothesis}
\end{equation}
where $\bbt^*_{\cM}$ is a sub-vector of the true regression coefficients $\bbt^*$, for a given $\cM\subseteq [1,\dots,p]$ with cardinality $\vert\cM\vert=m$, a matrix $\bC\in\mathbb{R}^{r\times m}$ with full row rank $rank(\bC)=r$, and a vector $\bt\in\mathbb{R}^{r}$. In this paper, we consider the ultra-high dimensional setting: $\log(p)=n^a$ for some $0<a<1$. 
\\
For sparse estimation of regression coefficients $\bbt$, we define the following \emph{partial penalized composite probit likelihood} function:
\begin{equation}
	Q_n(\bbt, \bb;\lambda)=
			\sum_{k=1}^{K} w_k 
					\left\{
					\frac{1}{n} \sum_{i=1}^{n} 
							\left[
							\tilde{y}_{ki}h({\bx_i}^{T}\bbt - b_{k})+\log(1-\Phi({{\bx_i}^{T}\bbt - b_{k}})
							\right]
					\right\}
			- \sum_{j \notin \cM}p_\lambda(\vert\beta_j\vert)
    \label{eq:partial_penalized_composite_probit_likelihood}
\end{equation}
where $p_\lambda(\cdot)$ is a penalty function with a tuning parameter $\lambda$. 
Further define the \emph{partial penalized composite probit regression estimator} $(\hat{\bbt}_a, \hat{\bb}_a)$ as the solution to the following unconstrained optimization problem:
\begin{equation}
	\max_{\bbt, \bb} Q_n(\bbt, \bb;\lambda),
        \label{eq:beta_a}
\end{equation}
and $(\hat{\bbt}_0, \hat{\bb}_0)$ be the solution to the following constrained optimization problem
\begin{equation}
	\max_{\bbt, \bb} Q_n(\bbt, \bb;\lambda)~~\text{s.t.}~ \bC\bbt_{\cM}=\bt.
        \label{eq:beta_0}
\end{equation}

Note that the partial penalized composite probit likelihood function \eqref{eq:partial_penalized_composite_probit_likelihood} does not apply penalties on parameters $\bbt_\cM$ involved in the testing constraints. This partial penalization avoids imposing minimal signal assumption on $\bbt^*_{\cM}$. As a result, we will show in section \ref{sec:partial_penalized_tests} that the corresponding tests will have non-trivial power at local alternatives. However, on the other hand, we sacrifice the variable selection capability on the subset $\cM$, which might be undesirable when $\bbt^*_{\cM}$ is sparse.

\subsection{Partial Penalized Composite Probit Regression Estimators}\label{subsec:partial_penalized_CPR_estimators}


There are two mainstream penalization methods, the LASSO \citep{tibshirani1996regression} and the concave penalization \citep{fan2001variable}, for sparse estimation in high-dimensional regression problems. We follow the definition given by \cite{fan2014strong} and assume that the penalty $p_\lambda(\vert t\vert)$ in our partial penalized composite probit regression is a general folded concave penalty function. The well-known penalty functions such as SCAD \citep{fan2001variable} and MCP \citep{zhang2010nearly} belong to this general family.

To simplify the notation, define 
\begin{equation*}
    \cB=\begin{pmatrix}
		\bbt
            \\
            \bb
	\end{pmatrix}\in\mathbb{R}^{p+K}, \quad
    \bx_i^{k}=\begin{pmatrix}
        \bx_i
        \\
	-\bm{e}_k
\end{pmatrix}, \quad
\bX^k=\begin{pmatrix}
	(\bx_1^k)^{T}
	\\
	\vdots
	\\
	(\bx_n^k)^{T}
\end{pmatrix},
\quad 
\bm{Y}^k=\begin{pmatrix}
\tilde{y}_{1k}\\\vdots\\\tilde{y}_{nk}
\end{pmatrix},
\end{equation*}
where $e_k\in\mathbb{R}^{K}$ with the $k$-th element 1 and all others 0.
Then the composite probit likelihood function \eqref{eq:composite_probit_likelihood} can be re-written as
\begin{equation}
	{M}_n(\cB)=\sum_{k=1}^{K}w_k\left\{
	\frac{1}{n}\sum_{i=1}^{n}\left[
        \tilde{y}_{ik}h((\bx_i^k)^{T}\cB)+\log(1-\Phi((\bx_i^k)^{T}\cB))\right]
	\right\}.
	\label{eq:CPR_likelihood_B}
\end{equation}

Let $\cK\subseteq [1,\dots,p+K]$ be the subset corresponding to the indices of the $K$ intercepts, i.e., $\cB_{\cK} = \bb$. Then $\cB^*_{\cK} = \bb^*$ denotes the true intercepts. And $\bX^k_{\cK}=-\bm{1}_n {\be_k}^T =:-\bm{E}_k$.

\subsubsection{Assumptions}\label{subsubsec:assumptions}

\paragraph{Assumptions on true $\cB^* = (\bbt^*, \bb^*)$:}

We assume that the true regression coefficient $\bbt^*$ is sparse and satisfies the linear equation $\bC\bbt^*_{\cM}-\bt = \bh_n$ for some $\bh_n \rightarrow 0$. Notice that when $\bh_n=\bm{0}$, the null holds; otherwise, the local alternative holds.  Define the partial support set of $\bbt^*$ as $S=\{j\in[1,\dots,p]\backslash \cM:\beta^*_{j}\neq 0\}$ with cardinality $s=\vert S\vert$. Let $d_n=\min_{j\in S}\vert \beta^*_{j}\vert/2$ be the half minimal signal of $\bbt^*_{S}$. Notice that the partial support set does not include nonzero coefficients in $\cM$ and thus the minimal signal $d_n$ defined above does not consider the signal strength on $\bbt^*_{\cM}$. Besides, we have $\cB^*_{(\cM\cup S\cup\cK)^c}=\bm{0}$ by the definition and assumptions.


\paragraph{Assumptions on the model:} We also introduce a neighborhood of $\cB^*$: $\cN^*=\{\cB\in\mathbb{R}^{p+K}:\Vert \cB - \cB^*\Vert_2\leq\sqrt{(m+s+K)\log{(n)}/n},~\cB_{(\cM\cup S\cup\cK)^c}=\bm{0}\}$.

(A1) Assume that
\begin{align*}
        &
        \Vert \bX_{\cM\cup S} \Vert_{2,\infty}
        = O\left(\sqrt{\frac{n}{(m+s+K)\log n}}\right),
        &
        \\
        &
        \Vert \bX_{(\cM\cup S\cup\cK)^c} \Vert_{1,2}
        =O(\sqrt{n}),
        &
        \\
        &
        \max_{1\leq i\leq n}\max_{1\leq k\leq K}\vert (\bx_{i}^{k})^T\cB^*\vert = O(1),
        &
        \\
        &
        \lambda_{\max}\left(
            \frac{1}{n}
            \sum_{k=1}^{K}w_k
            \begin{pmatrix}
                            \bm{X}_{\cM}^T
                            \\
                            (\bm{X}^k_{S\cup \cK})^T
                        \end{pmatrix}
            \begin{pmatrix}
                            \bm{X}_{\cM}^T
                            \\
                            (\bm{X}^k_{S\cup \cK})^T
                        \end{pmatrix}^T
        \right)
        = O(1),
        &
        \\
        &
        \max_{j\in{\cM\cup S}}
        \lambda_{\max}
                \left(
                    \frac{1}{n}
                    \sum_{k=1}^{K}w_k
                    \begin{pmatrix}
                            \bm{X}_{\cM}^T
                            \\
                            (\bm{X}^k_{S\cup \cK})^T
                        \end{pmatrix}
                    \text{diag}\left\{
                        \vert\bx_{(j)}\vert
                    \right\}
                    \begin{pmatrix}
                            \bm{X}_{\cM}^T
                            \\
                            (\bm{X}^k_{S\cup \cK})^T
                        \end{pmatrix}^T
                \right)
        = O(1),
        &
        \\
        &
        \inf_{\cB\in\cN^*}
        \lambda_{\min}
                \left(
                    \frac{1}{n}
                    \sum_{k=1}^{K}w_k
                    \begin{pmatrix}
                            \bm{X}_{\cM}^T
                            \\
                            (\bm{X}^k_{S\cup \cK})^T
                        \end{pmatrix}
                    \Sigma(\bX^k\cB)
                    \begin{pmatrix}
                            \bm{X}_{\cM}^T
                            \\
                            (\bm{X}^k_{S\cup \cK})^T
                        \end{pmatrix}^T
                \right)
        \geq c,
        \\
        &
        \lambda_{\min}
        \left(
            \frac{1}{n}
            \sum_{i=1}^{n}
            \sum_{k,k^{\prime}=1}^{K}
            w_{k} w_{k^{\prime}}
            \left(
                \bm{x}_{i,\cM\cup{S}\cup\cK}^{k}
            \right)
            \left(
                \bm{x}_{i,\cM\cup{S}\cup\cK}^{k^{\prime}}
            \right)^T
            \frac{
                \varphi\left(
                    ({\bm{x}_{i}^{k}})^T
                    {\cB^*}
                \right)
                \cdot\varphi\left(
                    ({\bm{x}_{i}^{k^{\prime}}})^T
                    {\cB^*}
                \right)
            }{
                \Phi\left(
                    (\bm{x}_{i}^{\min\{k,k^\prime\}})^T
                    {\cB^*}
                \right)
                \cdot
                \Phi\left(
                    -
                    (\bm{x}_{i}^{\max\{k,k^\prime\}})^T
                    {\cB^*}
                \right)
            }
        \right)
        \geq c,
    \end{align*}
where ${\Sigma}(\bm{X}^k\cB) = \text{diag}
                        \left\{
                            \frac{\varphi^2}{\Phi(1-\Phi)}((\bx_{i}^k)^T\cB)
                            ,~i=1,\dots,n
                        \right\}$, and $\varphi(\cdot)$ is the probability density function (pdf) of standard normal distribution.

Notice that since $(\bx_{i}^k)^T\cB^*=\bx_{i}^T\bbt^*-b^*_{k}$, the third condition of (A1) can be derived from $\max_{1\leq i\leq n}
    \vert
    \bx_{i}^T\bbt^*
    \vert
    = O(1)$
and 
$
    \Vert
    \bb^*
    \Vert_{\infty}
    = O(1)$. 
\cite{van2014asymptotically} and \cite{ning2017general} also assumed $\max_{1\leq i\leq n}
    \vert
    \bx_{i}^T\bbt^*
    \vert
    = O(1)$ 
to establish the asymptotic properties of their sparse estimators and test statistics under the high dimensional generalized linear models.

\paragraph{Assumptions on penalization $p_\lambda(\cdot)$:}
Let $\rho(t,\lambda) = p_\lambda(t)/\lambda$ for $\lambda>0$. We assume that $\rho(t,\lambda)$ is increasing and concave in $t\in[0,\infty)$ with $\rho(0,\lambda)=0$, and is continuously differentiable in $t\in(0,\infty)$. At $t=0$, we assume that $\rho^\prime(0,\lambda):=\rho^\prime(0+,\lambda)>0$ is a positive constant independent of $\lambda$. Given $t\in(0,\infty)$, we assume that $\rho^\prime (t,\lambda)$ is increasing in $\lambda\in(0,\infty)$. We define local concavity of $\rho$ at $\bv$ with $\Vert\bv\Vert_0=q$:
	\begin{equation*}
		\kappa(\rho, \bv, \lambda) =
			\lim_{\varepsilon\rightarrow 0}
			\max_{j:v_j\neq 0}
			\sup_{t_1<t_2\in(\vert v_j \vert - \varepsilon, \vert v_j \vert + \varepsilon)}
			-\frac{
					\rho^\prime(t_2,\lambda) - \rho^\prime(t_1,\lambda)
					}{
				  	 t_2 - t_1
			  	 	 }.
	\end{equation*}
Further define the maximal local concavity of $\rho$ at neighborhood $\cN^*$ of $\cB^*$:
\begin{equation*}
    \kappa^*_{j}=\max_{\cB\in\cN^*}\kappa(\rho,\cB,\lambda_{n,j})
\end{equation*}
for $j=0,a$, where $\lambda_{n,0}$ is tuning parameter for the constrained partial penalized CPR estimator $\hat{\cB}_0$ and $\lambda_{n,a}$ is the tuning parameter for the unconstrained partial penalized CPR estimator $\hat{\cB}_a$.

(A2) Assume that for $j=0,a$,
\begin{align*}
   &
   \lambda_{n,j}\kappa^*_{j}=o(1),
   &
   \\
   &
   p_{\lambda_{n,j}}^{\prime}(d_n)=o(1/\sqrt{(m+s)n}),
   &
   \\
   &
   d_n\gg\lambda_{n,j} \gg \max\{
   \sqrt{(s+m+K)/n},\sqrt{(\log p)/n}
   \}.
   &
\end{align*}

Notice that the third condition in (A2) imposes a minimum signal assumption on nonzero elements of $\bbt^*$ in $[1,\dots,p]\backslash\cM$ only, which is due to the partial penalization.


\paragraph{Assumption on testing linear equations:}
\quad

(A3) Assume that $\Vert \bh_n \Vert_2=O(\sqrt{\min(m+s+K-r, r)/n})$, and $\lambda_{\max}((\bC\bC^T)^{-1})=O(1)$.

\subsubsection{Asymptotic Results for Partial Penalized CPR Estimators}

\begin{theorem}\label{theorem:estimator_statistical_properties}
    Suppose that Conditions (A1) - (A3) hold, and $m+s+K = o(\sqrt{n})$, then the following holds:
    \begin{itemize}
        \item[(i)] With probability tending to 1, $\hat{\bbt}_0$ and $\hat{\bbt}_a$ defined in \eqref{eq:beta_0} and \eqref{eq:beta_a} must satisfy
        \begin{equation*}
            \hat{\bbt}_{0,(S\cup\cM)^c} = \hat{\bbt}_{a,(S\cup\cM)^c} = \bm{0}.
        \end{equation*}
        \item[(ii)] Their corresponding $\ell_2$ errors have rate
        \begin{equation*}
            \Vert
                \hat{\cB}_{a,S\cup\cM\cup\cK} - {\cB}^*_{S\cup\cM\cup\cK}
            \Vert_2 = 
            O_p \left(
                    \sqrt{\frac{m+s+K}{n}}
                \right),
        \end{equation*}
        and
        \begin{equation*}
            \Vert
                \hat{\cB}_{0,S\cup\cM\cup\cK} - {\cB}^*_{S\cup\cM\cup\cK}
            \Vert_2 = 
            O_p \left(
                    \sqrt{\frac{m+s+K-r}{n}}
                \right).
        \end{equation*}
    \end{itemize}
    If further $m+s+K=o(n^{1/3})$, then we have
    \begin{equation}
        \sqrt{n}
                \begin{pmatrix}
                            \hat{\cB}_{a,\cM} - \cB^*_{\cM} \\
                            \hat{\cB}_{a,S\cup\cK} - \cB^*_{S\cup\cK}
                        \end{pmatrix}
        =
        \frac{1}{\sqrt{n}}
                \bm{K}_n^{-1}
                \sum_{k=1}^{K}w_k
                    \begin{pmatrix}
                            \bm{X}_{\cM}^T
                            \\
                            (\bm{X}^k_{S\cup \cK})^T
                        \end{pmatrix}
                        \bm{H}(\bm{X}^k{{\cB^*}})
                        \left\{
                            \bm{Y}^k - \bm{\mu}(\bm{X}^k{{\cB^*}})
                        \right\}
         +
                o_p(1).
    \end{equation}
    and 
    \begin{align}
                \sqrt{n}
                \begin{pmatrix}
                            \hat{\cB}_{0,\cM} - \cB^*_{\cM} \\
                            \hat{\cB}_{0,S\cup\cK} - \cB^*_{S\cup\cK}
                        \end{pmatrix}
                & = 
                \frac{1}{\sqrt{n}}
                \bm{K}_n^{-1/2}
                (\bm{I}-\bm{P}_n)
                \bm{K}_n^{-1/2}
                \sum_{k=1}^{K}w_k
                    \begin{pmatrix}
                            \bm{X}_{\cM}^T
                            \\
                            (\bm{X}^k_{S\cup \cK})^T
                        \end{pmatrix}
                        \bm{H}(\bm{X}^k{{\cB^*}})
                        \left\{
                            \bm{Y}^k - \bm{\mu}(\bm{X}^k{{\cB^*}})
                        \right\}
                \nonumber
                \\
                & -
                \sqrt{n}
                \bm{K}_n^{-1/2}
                \bm{P}_n
                \bm{K}_n^{-1/2}
                \begin{pmatrix}
                    \bm{C}^T({\bm{C}\Omega_{mm}\bm{C}^T})^{-1}\bm{h}_n
                    \\
                    \bm{0}_{s+K}
                \end{pmatrix}
                +
                o_p(1).
            \end{align}
    where $\bm{H}(\bm{X}^k\cB)
                    = \text{diag}
                        \left\{
                            h^{\prime}((\bx_i^k)^T\cB)
                            = \frac{\varphi}{\Phi(1-\Phi)}((\bx_i^k)^T\cB),~i=1,\dots,n
                        \right\}$, 
    $\bm{I}$ is the identity matrix, $\bm{K}_n$ is the $(m+s+K)\times(m+s+K)$ matrix
    \begin{equation*}
                \bm{K}_n = 
                \frac{1}{n}
                \sum_{k=1}^{K}w_k 
                        \begin{pmatrix}
                            \bm{X}_{\cM}^T
                            \\
                            (\bm{X}^k_{S\cup \cK})^T
                        \end{pmatrix}
                        {\Sigma}(\bm{X}^k{\cB^*})
                        \begin{pmatrix}
                            \bm{X}_{\cM}^T
                            \\
                            (\bm{X}^k_{S\cup \cK})^T
                        \end{pmatrix}^T,
    \end{equation*}
    in which ${\Sigma}(\bm{X}^k\cB) = \text{diag}
                        \left\{
                            \frac{\varphi^2}{\Phi(1-\Phi)}((\bx_{i}^k)^T\cB)
                            ,~i=1,\dots,n
                        \right\}$, and $\bm{P}_n$ is the $(m+s+K)\times(m+s+K)$ projection matrix of rank $r$
    \begin{equation*}
                \bm{P}_n 
                = 
                \bm{K}_n^{-1/2}
                \begin{pmatrix}
                    \bC^T 
                            \\
                            \bm{0}_{r\times (s+K)}^T
                \end{pmatrix}
                \left(
                \begin{pmatrix}
                    \bC^T 
                            \\
                            \bm{0}_{r\times (s+K)}^T
                \end{pmatrix}^T
                \bm{K}_n^{-1}
                \begin{pmatrix}
                    \bC^T 
                            \\
                            \bm{0}_{r\times (s+K)}^T
                \end{pmatrix}
                \right)^{-1}
                \begin{pmatrix}
                    \bC^T 
                            \\
                            \bm{0}_{r\times (s+K)}^T
                \end{pmatrix}^T
                \bm{K}_n^{-1/2}.
    \end{equation*}
    
\end{theorem}

\paragraph{Remark 2.1.} If we set $\cM=\emptyset$, then $\mathcal{S}=\{j\in [1,\dots,p] :\beta^*_j \neq 0\}$ represents the conventional support set of $\bbt^*$. Consequently, Theorem \ref{theorem:estimator_statistical_properties} implies the convergence rate and asymptotic distribution of the standard folded concave penalized composite probit regression estimator with the penalty function $\sum_{j=1}^{p}p_\lambda(\vert\beta_j\vert)$, thus concluding the theoretical analysis of the statistical properties of the composite probit regression method in \cite{zhou2023nonparametric}.

\paragraph{Remark 2.2.} Since (A2) assumes that $d_n\gg\sqrt{(s+m+K)/n}$, Theorem \ref{theorem:estimator_statistical_properties} (ii) implies that each element in $\hat{\bbt}_{a,S}$ and $\hat{\bbt}_{0,S}$ is nonzero and has the same sign as the true coefficient ${\bbt^*}_{S}$ with probability tending to $1$. Along with results in (i), we proves the sign consistency of $\hat{\bbt}_{a,\cM^c}$ and $\hat{\bbt}_{0,\cM^c}$.

\paragraph{Remark 2.3.} Theorem \ref{theorem:estimator_statistical_properties} implies that when $\bh_n$ converges to $0$ at an appropriate rate (given in (A3)), constrained estimator $\hat{\cB}_{0}$ converges faster than the unconstrained estimator $\hat{\cB}_{a}$ when ${m+s+K-r} \ll {m+s+K}$, where $r$ is the number of independent linear constraints in the hypothesis.

\section{Partial Penalized Wald, Score and Likelihood-Ratio Tests}\label{sec:partial_penalized_tests}

Our approach to constructing the testing procedures aligns with the principles underlying Wald, score, and likelihood ratio tests for (generalized) linear regression models in the fixed $p$ case. In Section \ref{subsec:partial_penalized_CPR_method}, we introduced the partial penalized composite probit regression method for the sparse estimation of the non-parametric Box-Cox regression model \eqref{eq:non-parametric_Box-Cox_model} in high-dimensional settings. In this section, we will use the proposed constrained and unconstrained partial penalized composite probit regression estimators, $\hat{\cB}_{0}$ defined in \eqref{eq:beta_0} and $\hat{\cB}_{a}$ defined in \eqref{eq:beta_a},
to develop our testing procedures for high-dimensional linear hypotheses $H_0:\bC\bbt^*_{\cM}=\bt$.

\subsection{Test Statistics}

\paragraph{Partial Penalized Composite Likelihood Ratio Test Statistic}
Since the composite probit regression method in Section \ref{subsec:CPR} is based on the composite probit likelihood function, we introduce the \emph{partial penalized composite likelihood ratio test statistic},
\begin{equation}
    T_L
    =
    2n
    \left\{
    M_n(\hat{\cB}_a)
    -
    M_n(\hat{\cB}_0)
    \right\},
\end{equation}
where $M_n(\cB)$ is the composite probit regression likelihood function given in \eqref{eq:CPR_likelihood_B}.

\paragraph{Partial Penalized Wald Test Statistic}

The partial penalized Wald statistic is based on the quantity $\sqrt{n}(\bm{C}\hat{\bbt}_{a,\cM}-\bm{t})$. Let $\widehat{S}_a=\{j\in [1,\dots,p]\backslash \cM:\hat{\beta}_{a,j}\neq 0\}$, be the partial active set of unconstrained estimator $\hat{\bbt}_a$. By Theorem \ref{theorem:estimator_statistical_properties}, we have $\widehat{S}_a=S$ with probability tenting to $1$. Define $\widehat{\Omega}_a = ({\widehat{\bm{K}}_{n,a}})^{-1}$, where
\begin{equation}
    \widehat{\bm{K}}_{n,a}
    =
    \frac{1}{n}
    \sum_{k=1}^{K}w_k 
    \begin{pmatrix}
        \bm{X}_{\cM}^T
            \\
        \left( {\bm{X}^k_{{\widehat{S}_a}\cup \cK}} \right)^T
    \end{pmatrix}
    {\Sigma}(\bm{X}^k{\hat{\cB}_a})
    \begin{pmatrix}
        \bm{X}_{\cM}^T
        \\
        \left( {\bm{X}^k_{{\widehat{S}_a}\cup \cK}} \right)^T
    \end{pmatrix}^T
    ,
    \label{eq:test_statistic_Ka}
\end{equation}
is the plug-in estimator of the matix ${\bm{K}}_n$ under the alternative, and define $\widehat{\Omega}_{a,mm}$ as its submatrix formed by its first $m$ rows and $m$ columns. The \emph{partial penalized Wald test statistic} is defined as
\begin{equation}
    T_W 
    = 
    n
    (\bm{C}\hat{\bbt}_{a,\cM}-\bm{t})^T
    (
        \bm{C}
        \widehat{\Omega}_{a,mm}
        \bm{C}^T
    )^{-1}
    (\bm{C}\hat{\bbt}_{a,\cM}-\bm{t}).
\end{equation}

\paragraph{Partial Penalized Score Test Statistic}

The partial penalized score statistic is based on the score function evaluated at $\hat{\cB}_0$
\begin{equation}
    \mathcal{S}_n(\hat{\cB}_0)
    :=
    \sum_{k=1}^{K}w_k
                    \begin{pmatrix}
                            \bm{X}_{\cM}^T
                            \\
                            \left( {\bm{X}^k_{{\widehat{S}_0}\cup \cK}} \right)^T
                        \end{pmatrix}
                        \bm{H}(\bm{X}^k{\hat{\cB}_0})
                        \left\{
                            \bm{Y}^k - \bm{\mu}(\bm{X}^k{\hat{\cB}_0})
                        \right\},
\end{equation}
where $\widehat{S}_0=\{j\in [1,\dots,p]\backslash\cM :\hat{\beta}_{0,j}\neq 0\}$ is the partial active set of constrained estimator $\hat{\bbt}_0$. Define $\widehat{\Omega}_0 = (\widehat{\bm{K}}_{n,0})^{-1}$, where
\begin{equation}
    \widehat{\bm{K}}_{n,0}
    =
    \frac{1}{n}
    \sum_{k=1}^{K}w_k 
    \begin{pmatrix}
        \bm{X}_{\cM}^T
            \\
        \left( {\bm{X}^k_{{\widehat{S}_0}\cup \cK}} \right)^T
    \end{pmatrix}
    {\Sigma}(\bm{X}^k{\hat{\cB}_0})
    \begin{pmatrix}
        \bm{X}_{\cM}^T
        \\
        \left( {\bm{X}^k_{{\widehat{S}_0}\cup \cK}} \right)^T
    \end{pmatrix}^T,
    \label{eq:test_statistic_K0}
\end{equation}
is the plug-in estimator of matrix ${\bm{K}}_n$ under the null.
The \emph{partial penalized score test statistic} is defined as
\begin{equation}
    T_S
    =
    \frac{1}{n}
    \mathcal{S}_n(\hat{\cB}_0)^T
    \widehat{\Omega}_0
    \mathcal{S}_n(\hat{\cB}_0).
\end{equation}

\subsection{Limiting Distributions of the Test Statistics}
Define several useful quantities:
\begin{align*}
            \Psi_n 
            &
            = 
            \bm{C}
            {\Omega}_{mm}
            \bm{C}^T
            =
            \begin{pmatrix}
                    \bC^T 
                            \\
                            \bm{0}_{r\times (s+K)}^T
                \end{pmatrix}^T
            \bm{K}_n^{-1}
            \begin{pmatrix}
                    \bC^T 
                            \\
                            \bm{0}_{r\times (s+K)}^T
                \end{pmatrix},
            \\
            \mathcal{T}_n 
            &
            =
            \begin{pmatrix}
                    \bC^T 
                            \\
                            \bm{0}_{r\times (s+K)}^T
                \end{pmatrix}^T
            \bm{K}_n^{-1}
            \bm{V}_n
            \bm{K}_n^{-1}
            \begin{pmatrix}
                    \bC^T 
                            \\
                            \bm{0}_{r\times (s+K)}^T
                \end{pmatrix},
            \\
            \bm{V}_n
            &
            =
            \frac{1}{n}\sum_{i=1}^{n}
            \sum_{k,k^{\prime}=1}^{K}
            w_{k} w_{k^{\prime}}
            \left(
                \bm{x}_{i,\cM\cup{S}\cup\cK}^{k}
            \right)
            \left(
                \bm{x}_{i,\cM\cup{S}\cup\cK}^{k^{\prime}}
            \right)^T
            \frac{
                \varphi\left(
                    ({\bm{x}_{i}^{k}})^T
                    \cB^*
                \right)
                \cdot\varphi\left(
                    ({\bm{x}_{i}^{k^{\prime}}})^T
                    \cB^*
                \right)
            }{
                \Phi\left(
                    (\bm{x}_{i}^{\min\{k,k^\prime\}})^T
                    \cB^*
                \right)
                \cdot
                \Phi\left(
                    -
                    (\bm{x}_{i}^{\max\{k,k^\prime\}})^T
                    \cB^*
                \right)
            }.
\end{align*}
Notice that $\bm{V}_n=\frac{1}{{n}}\text{Cov}(\mathcal{S}_n(\cB^*))$, where $\mathcal{S}_n(\cB^*)$ is the score function of the composite probit likelihood evaluated at $\cB^*$:
\begin{equation*}
    \mathcal{S}_n({\cB}^*)
    =
    \sum_{k=1}^{K}w_k
                    \begin{pmatrix}
                            \bm{X}_{\cM}^T
                            \\
                            \left( {\bm{X}^k_{{{S}}\cup \cK}} \right)^T
                        \end{pmatrix}
                        \bm{H}(\bm{X}^k{{\cB^*}})
                        \left\{
                            \bm{Y}^k - \bm{\mu}(\bm{X}^k{{\cB^*}})
                        \right\}.   
\end{equation*}

\begin{theorem}\label{theorem:testing_statistic_distribution}
    Assume Conditions (A1) - (A3) hold, $m+s+K=o(n^{1/3})$. Further assume the following holds:
    \begin{equation}
        \frac{r^{1/4}}{n^{3/2}}\sum_{i=1}^{n}\sum_{k=1}^{K}
        w_k
        \left\{
            (\bm{x}_{i,\cM\cup{S}\cup\cK}^k)^T
            \bm{V}_n^{-1}
            (\bm{x}_{i,\cM\cup{S}\cup\cK}^k)
        \right\}^{3/2}
        \rightarrow 0.
        \label{eq:theorem(limiting_distribution)condition}
    \end{equation}
    Then we have
    \begin{equation}
        \sup_{x}
        \left\vert
            \mathbb{P}(T\leq x)
            -
            \mathbb{P}\left(
                \left\Vert
                \Psi_n^{-1/2}
                \left(
                {\mathcal{T}_n}^{1/2} \bm{Z}
                +
                \sqrt{n}\bm{h}_n
                \right)
            \right\Vert_2^2
            \leq x
            \right)
        \right\vert
        \rightarrow 0,
        \label{eq:limiting_distribution_result}
    \end{equation}
    for $T=T_{W}$, $T_{S}$ or $T_{L}$, where $\bm{Z}\sim{N}(\bm{0}_r,\bm{I}_{r})$.
\end{theorem}

\paragraph{Remark 4.1.} From \eqref{eq:limiting_distribution_result}, we proved that the partial penalized composite likelihood ratio test statistic $T_L$, partial penalized Wald test statistic $T_W$, and partial penalized score test statistic $T_S$ are asymptotically equivalent to the following quantity:
\begin{equation*}
    \mathcal{X}_n:=
    \left\Vert
                \Psi_n^{-1/2}
                \left(
                {\mathcal{T}_n}^{1/2} \bm{Z}
                +
                \sqrt{n}\bm{h}_n
                \right)
            \right\Vert_2^2.
\end{equation*}
When $K=1$, we have $\Psi_n={\mathcal{T}_n}$. Then $\mathcal{X}_n\sim\chi^2(r,\gamma_n)$, a chi-squared distribution with degrees of freedom $r$ and non-centrality parameter $\gamma_n=\Vert\sqrt{n}\Psi_n^{-1/2}\bh_n
\Vert_2^2$. This coincides with the limiting distribution of the partial penalized test statistics for generalized linear model developed by \cite{shi2019linear}.
When $K\neq 1$, we have $\Psi_n\neq{\mathcal{T}_n}$. Then $\mathcal{X}_n$ follows a generalized chi-squared distribution with $r$ degrees of freedom.

\subsection{Testing Procedure}\label{subsec:testing_procedure}
Under the null hypothesis, we have $\bm{h}_n=\bm{0}$, and hence, for $T=T_{W}$, $T_{S}$ or $T_{L}$,
\begin{equation*}
        \sup_{x}
        \left\vert
            \mathbb{P}(T\leq x)
            -
            \mathbb{P}\left(
                \bm{Z}^T
                {\mathcal{T}_n}^{1/2}
                \Psi_n^{-1}
                {\mathcal{T}_n}^{1/2} 
                \bm{Z}
            \leq x
            \right)
        \right\vert
        \rightarrow 0.
\end{equation*}
Given the significance level $\alpha\in(0,1)$, denote the $(1-\alpha)$-quantile of $\bm{Z}^T
                {\mathcal{T}_n}^{1/2}
                \Psi_n^{-1}
                {\mathcal{T}_n}^{1/2} 
                \bm{Z}$ 
as $\chi_{n,(1-\alpha)}$, i.e., $\chi_{n,(1-\alpha)}$ is the value satisfying the following equation:
\begin{equation*}
    \mathbb{P}\left(
                \bm{Z}^T
                {\mathcal{T}_n}^{1/2}
                \Psi_n^{-1}
                {\mathcal{T}_n}^{1/2} 
                \bm{Z}
            \leq \chi_{n,(1-\alpha)}
            \right)
            =
            1-\alpha.
\end{equation*}

When $K = 1$, we have $\Psi_n=\mathcal{T}_n$ and $\bm{Z}^T
                {\mathcal{T}_n}^{1/2}
                \Psi_n^{-1}
                {\mathcal{T}_n}^{1/2} 
                \bm{Z} = \bm{Z}^T\bm{Z}\sim\chi^2(r)$, which implies that $\chi_{n,(1-\alpha)}$ is the $(1-\alpha)$-quantile of the $\chi^2$ distribution with $r$ degrees of freedom.

    When $K > 1$, we have $\Psi_n \neq \mathcal{T}_n$, which implies that $\chi_{n,(1-\alpha)}$ is $(1-\alpha)\times 100\%$ quantile of a generalized $\chi^2$ distribution with $r$ degrees of freedom and the shape of the ellipsoid is given by the positive-definite matrix ${\mathcal{T}_n}^{1/2}
                \Psi_n^{-1}
                {\mathcal{T}_n}^{1/2} $,
    determined by the design matrix $\bm{X}$, the true coefficients $\cB^*$, and the testing matrix $\bC$. Since the true coefficients $\cB^*$ are unknown, the quantile $\chi_{n,(1-\alpha)}$ is unknown. To construct an approximated testing procedure, we plug in the unconstrained estimator, $\hat{\cB}_{a}$, to get the estimation of $\Psi_n$ and $\mathcal{T}_n$, given by
    \begin{align*}
    \widehat{\Psi}_{n,a}
    &
    :=
    \begin{pmatrix}
                    \bC^T 
                            \\
                            \bm{0}_{r\times (s+K)}^T
                \end{pmatrix}^T
            \widehat{\bm{K}}_{n,a}^{-1}
            \begin{pmatrix}
                    \bC^T 
                            \\
                            \bm{0}_{r\times (s+K)}^T
                \end{pmatrix}
                =
                \bm{C}
        \widehat{\Omega}_{a,mm}
        \bm{C}^T
        ,
    \\
    \widehat{\mathcal{T}}_{n,a}
    &
    :=
    \begin{pmatrix}
                    \bC^T 
                            \\
                            \bm{0}_{r\times (s+K)}^T
                \end{pmatrix}^T
            \widehat{\bm{K}}_{n,a}^{-1}
            \widehat{\bm{V}}_{n,a}
            \widehat{\bm{K}}_{n,a}^{-1}
            \begin{pmatrix}
                    \bC^T 
                            \\
                            \bm{0}_{r\times (s+K)}^T
                \end{pmatrix},
    \\
    \widehat{\bm{V}}_{n,a}
    &
    :=
    \frac{1}{n}\sum_{i=1}^{n}
            \sum_{k,k^{\prime}=1}^{K}
            w_{k} w_{k^{\prime}}
            \left(
                \bm{x}_{i,\cM\cup{\widehat{S}_a}\cup\cK}^{k}
            \right)
            \left(
                \bm{x}_{i,\cM\cup{\widehat{S}_a}\cup\cK}^{k^{\prime}}
            \right)^T
            \frac{
                \varphi\left(
                    \hat{\cB}_{a}^T
                    {\bm{x}_{i}^{k}}
                \right)
                \cdot\varphi\left(
                    \hat{\cB}_{a}^T
                    {\bm{x}_{i}^{k^{\prime}}}
                \right)
            }{
                \Phi\left(
                    \hat{\cB}_{a}^T
                    \bm{x}_{i}^{\min\{k,k^\prime\}}
                \right)
                \cdot
                \Phi\left(
                    -
                    \hat{\cB}_{a}^T
                    \bm{x}_{i}^{\max\{k,k^\prime\}}
                \right)
            },
    \end{align*}
where $\widehat{\bm{K}}_{n,a}$ is defined in equation \eqref{eq:test_statistic_Ka}. Then $\chi_{n,(1-\alpha)}$ can be approximated by $\hat{\chi}_{n,(1-\alpha)}$, which is defined as the solution to the following equation:
\begin{equation*}
    \mathbb{P}\left(\left.
                \bm{Z}^T
                {\widehat{\mathcal{T}}_{n,a}}^{1/2}
                {\widehat{\Psi}_{n,a}}^{-1}
                {\widehat{\mathcal{T}}_{n,a}}^{1/2} 
                \bm{Z}
            \leq 
            \hat{\chi}_{n,(1-\alpha)}
            \right\vert
            \bm{X},\bm{Y}
            \right)
            =
            1-\alpha,
\end{equation*}
where $\bm{Z}\sim{N}(\bm{0}_r,\bm{I}_r)$ and $\bm{Z} {\perp \!\!\! \perp} (\bm{X},\bm{Y})$. Notice that $\hat{\chi}_{n,(1-\alpha)}$ is a function of $(\bm{X},\bm{Y})$, thus is a random variable. Given a realization of $(\bX,\bm{Y})$, we can get a realization of $\hat{\chi}_{n,(1-\alpha)}$ using Monte Carlo simulation. We can prove that the testing procedure developed based on $\hat{\chi}_{n,(1-\alpha)}$ has the correct size asymptotically.

\begin{theorem}\label{theorem:type_I_error}
    For testing hypothesis $H_0:\bC\bbt^*_{\cM}=\bt$, we construct the testing procedure as follows: For a given significance level $\alpha\in(0,1)$, reject the null hypothesis when 
    \begin{equation}
        T>\hat{\chi}_{n,(1-\alpha)},
        \label{eq:testing_procedure}
    \end{equation}
    where $T\in\{T_L,~T_W,~T_S\}$.
    Assume that the conditions of Theorem \ref{theorem:testing_statistic_distribution} as well as $(s+m+K)r=o(\sqrt{n})$ hold. Then, under the null hypothesis, we have
    \begin{equation*}
    \lim_{n}\mathbb{P}\left(
                T
            > 
            \hat{\chi}_{n,(1-\alpha)}
            \right)=
            \alpha.
\end{equation*}
\end{theorem}

\section{Computing Algorithms}\label{sec:algorithm}
Define $\cky_{ki}=-1$ if $\tilde{y}_{ki}=0$ and $\cky_{ki}=1$ if $\tilde{y}_{ki}=1$. Then the negative composite probit likelihood function \eqref{eq:composite_probit_likelihood} can be reformulated as the weighted summation of loss of margins:
\begin{equation}
    -M_n(\bbt, \bb) = {\sum_{k=1}^{K}w_k\left\{
    \frac{1}{n}\sum_{i=1}^{n}L(t_{ki})
    \right\}},
    \label{eq:CPR_loss_of_margins}
\end{equation}
where $	L(t)=-\log(\Phi(t))$ is the large-margin loss function induced from the probit regression model, named as \textit{probit regression loss function}, and $t_{ki}=\check{y}_{ki}(\bx_i^{T}\bbt-b_{k})$ is the margin of the $(k,i)$-th pair of data. Some useful properties of this probit loss function are studied in \cite{zhou2023nonparametric}. In this section, we will develop the computing algorithms for solving constrained and unconstrained estimators, $\hat{\bbt}_0$ and $\hat{\bbt}_a$, within the framework defined in \eqref{eq:CPR_loss_of_margins}.

\subsection{Computing the partial penalized linearly-constrained CPR estimator $\hat{\bbt}_0$}
By the definition in \eqref{eq:beta_0}, $(\hat{\bbt}_0,{\hat{\bb}}_0)$ is the solution to the following linearly-constrained optimization problem:
\begin{align*}
    \min_{\bbt, \bb} 
    &
    ~{\sum_{k=1}^{K}w_k\left\{
    \frac{1}{n}\sum_{i=1}^{n}L(\check{y}_{ki}(\bx_i^{T}\bbt-b_{k}))
    \right\}}
    + 
    \sum_{j \notin \cM}p_\lambda(\vert\beta_j\vert).
    \\
    \text{s.t.}
    &
    ~ \bC\bbt_{\cM}=\bt,
\end{align*}
Since there remains a lack of computing algorithms that offer both computational efficiency and theoretical guarantees for directly finding local minima in this type of optimization problem, \cite{shi2019linear} introduced an additional variable $\bm{\theta}$ with an added constraint $\bbt_{\cM^c}=\bm{\theta}$. They decomposed the objective function into a smooth and convex function of $(\bbt,\bb)$ and a non-smooth and non-convex function of $\bm{\theta}$. Then implementing the augmented Largrangian method becomes the alternating direction method of multipliers (ADMM \citep{boyd2011distributed}). When updating $(\bbt,\bb)$ given $\bm{\theta}$, \cite{shi2019linear} suggested to use the Newton-Raphson algorithm to find the minimizer of the smooth and convex optimization subproblem. However, in (ultra-)high-dimensional settings, the use of second-order algorithms like Newton-Raphson can pose significant challenges related to computational efficiency and accuracy of computing Hessian matrices and their inverse, particularly when $p$ is extremely large and $\bbt^*$ is extremely sparse. To tackle optimization problems with constraints under such circumstances, we propose to employ a variant of the standard Augmented Lagrangian Method \citep{hestenes1969multiplier,powell1969method}, which incorporate the special structure of our problems and is inspired by the success of coordinate-decent type-of algorithm for regularized generalized linear models \citep{friedman2023glmnet}. The algorithm can be decomposed to the following nested loops.

\paragraph{Outer Loop: Augmented Lagrangian Method}
The augmented Lagrangian function of our linearly-constrained minimization problem is
\begin{align}
    \mathcal{L}(\bbt, \bb, \bm{v};\rho) 
    :=
    &~
    {\sum_{k=1}^{K}w_k\left\{
    \frac{1}{n}\sum_{i=1}^{n}L(\check{y}_{ki}(\bx_i^{T}\bbt-b_{k}))
    \right\}}
    + 
    \sum_{j \notin \cM}p_\lambda(\vert\beta_j\vert)
    \nonumber
    \\
    +
    &~
    {\bm{v}}^{T} (\bC\bbt_{\cM}-\bt)
    +
    \frac{\rho}{2}\Vert \bC\bbt_{\cM}-\bt \Vert_2^2,
    \label{eq:Augmented_Lagrangian_function}
\end{align}
where $\rho>0$ is the augmentation parameter and $\bm{v}\in\mathbb{R}^{r}$ is the dual variable.
Then the method of multiplies solves the problem by updating the primal variable $(\bbt, \bb)$ and dual variable $\bm{v}$ alternatively until convergence:
\begin{align}
        (\hat{\bbt}^{s+1},\hat{\bm{b}}^{s+1})
        &
        = \arg\min_{\bbt,\bm{b}} \mathcal{L}(\bbt, \bb, \hat{\bm{v}}^{s};\rho);
        \label{eq:AML_beta_update}
        \\
        \hat{\bm{v}}^{s+1} 
        &
        = \hat{\bm{v}}^{s} + \rho(\bC\hat{\bbt}_{\cM}^{s+1}-\bt);
        \label{eq:AML_v_update}
\end{align}

We observe that, given $\hat{\bm{v}}$, the augmented Lagrangian function \eqref{eq:Augmented_Lagrangian_function} can be separated into one smooth and convex part, and one non-smooth and non-convex part:
\begin{align}
    \mathcal{L}(\bbt, \bb, \hat{\bm{v}};\rho) 
    =
    \underbrace{
    -M_n(\bbt,\bb)
    + 
    \hat{\bm{v}}^{T} (\bC\bbt_{\cM}-\bt)
    +
    \frac{\rho}{2}\Vert \bC\bbt_{\cM}-\bt \Vert_2^2
    }_{\text{convex part}}
    +
    \underbrace{
    \sum_{j \notin \cM}p_\lambda(\vert\beta_j\vert)
    }_{\text{non-convex part}},
    \label{eq:AML_beta_objective}
\end{align}
where the convex part is the negative composite likelihood function with the quadratic augmentation terms, and the non-convex part is the folded-concave penalty function.
To mitigate the computational issues along with high-dimensionality, we develop a \emph{first-order coordinate-majorization-descent} type of algorithm that can efficiently solve the folded-concave penalized problem \eqref{eq:AML_beta_objective} without introducing redundant variable $\bm{\theta}$. The idea of coordinate-majorization-descent is borrowed from \cite{zhou2023nonparametric}, \cite{yang2013efficient}, and \cite{friedman2023glmnet}, has proven to be effective in optimizing unconstrained folded-concave penalized regression problems, providing a direct solution for computing $\hat{\bbt}_a$ (details shown in Appendix \ref{appendix_sec:compute_beta_a}). Our algorithm incorporates the principle of minimization-majorization for descent updates and leverages the  coordinate-descent updates to enhance computational efficiency. For the estimation of generalized linear models with convex penalties (e.g. LASSO, ridge, elastic net), the coordinate-descent type of algorithms developed by \cite{friedman2010regularization} and implemented in \texttt{glmnet} \citep{friedman2023glmnet} have been demonstrated to be highly efficient across various real-world applications, outperforming some well-known accelerated algorithms \citep{nesterov1983AGD} widely used in many machine learning applications. In high-dimensional settings, sparse estimation can benefit from heuristic techniques such as active set and warm start, as incorporated by \cite{friedman2010regularization} in their coordinate-descent algorithms. These techniques exploit sparsity within the feature set and therefore greatly improve the algorithm's performance on high-dimensional datasets. Below, we provide the details of the algorithm for solving problem \eqref{eq:AML_beta_objective}.

\paragraph{Middle Loop: local linear approximation}
The minimization-majorization principle is used to convert the folded-concave penalized problem to a sequence of weighted $\ell_1$ penalized problems.
Let $(\hat{\bbt}^\textit{curr},\hat{\bm{b}}^\textit{curr})$ be the current estimate. The folded concave penalty could be majorized by its local linear approximation function:
\begin{equation}
	\sum_{j\notin \cM}p_\lambda(\vert\beta_j\vert)\leq
	\sum_{j\notin \cM}p_\lambda(\vert\hat{\beta}_j^\textit{curr}\vert) + p_\lambda^{\prime}(\vert{\hat{\beta}_j^\textit{curr}}\vert)(\vert\beta_j\vert-\vert{\hat{\beta}_j^\textit{curr}}\vert),
\end{equation}
which is proven to be the best convex majorization of the concave penalty function (Theorem 2 of \citealt{zou2008one}). Then the objective function \eqref{eq:AML_beta_objective} could be majorized by the following function:
\begin{equation}
	-M_n(\bbt,\bm{b})
        +
        \hat{\bm{v}}^{T} (\bC\bbt_{\cM}-\bt)
        +
        \frac{\rho}{2}\Vert \bC\bbt_{\cM}-\bt \Vert_2^2
        +
        \sum_{j \notin \cM}p_\lambda(\vert\hat{\beta}_j^\textit{curr}\vert) + p_\lambda^{\prime}(\vert{\hat{\beta}_j^\textit{curr}}\vert)(\vert\beta_j\vert-\vert{\hat{\beta}_j^\textit{curr}}\vert),
\end{equation}
which is a weighted $\ell_1$ penalized optimization problem. The pseudocode for the LLA algorithm are summarized in Algorithm \ref{alg:partial_penalized_linearly_constrained_LLA}. 

\begin{algorithm}[ht!]
	\caption{The LLA Algorithm for Solving \eqref{eq:AML_beta_update}}\label{alg:partial_penalized_linearly_constrained_LLA}
	\begin{algorithmic}[1]
            \State Input the dual vector $\hat{\bm{v}}$.
		\State Initialize $\hat{\bbt}^{(0)}=\hat{\bbt}^{\text{initial}}$ and compute the adaptive weight
		\begin{equation*}
			\hat{\bm{\omega}}^{(0)}
                =
                \left(
                    \hat{\omega}_j^{(0)},~j\in [1,\dots,p]\backslash\cM 
                \right)
                =
                \left(
                    p_{\lambda}^\prime(|\hat{\beta}_j^{(0)}|),~j\in [1,\dots,p]\backslash\cM
			\right)
		\end{equation*}
		\State For $m=1,2,\dots,$ repeat the LLA iteration till convergence
		\begin{itemize}
			\item[(3.a)] Obtain $(\hat{\bbt}^{(m)},\hat{\bm{b}}^{(m)})$ by minimizing the following objective function
			\begin{equation}
				-M_n(\bbt,\bm{b})
                    +
                    \hat{\bm{v}}^{T} (\bC\bbt_{\cM}-\bt)
                    +
                    \frac{\rho}{2}\Vert \bC\bbt_{\cM}-\bt \Vert_2^2
                    +
                    \sum_{j\notin\cM}\hat{\omega}_{j}^{(m-1)}|\beta_j|,
				\label{eq:AML_cpr_weightedL1}
			\end{equation}
			\item[(3.b)] Update the adaptive weight vector $\hat{\bm{\omega}}^{(m)}$ with $\hat{\omega}_j^{(m)}=p_{\lambda}^\prime(|\hat{\beta}_j^{(m)}|)$.
		\end{itemize}
	\end{algorithmic}
\end{algorithm}

\paragraph{Inner Loop: coordinate majorization descent}
In each iteration within the LLA Algorithm \ref{alg:partial_penalized_linearly_constrained_LLA}, when solving our weighted $\ell_1$ penalized problem \eqref{eq:AML_cpr_weightedL1}, we leverage the efficiency of coordinate descent algorithm \citep{friedman2010regularization} which has proven successful in finding sparse solution for many high-dimensional regression models. In the case of our non-parametric Box-Cox model, solved using the composite of a sequence of probit regression models, it is important to be extremely careful about computing overflow errors that may arise when computing the cumulative distribution function (CDF) $\Phi(\cdot)$. During cyclic udates across each coordinate, conventional second-order optimization algorithms like Newton-Raphson are very sensitive to large values of the linear predictor \citep{demidenko2001computational}, leading to inaccurate and overly aggressive step sizes.

Therefore, our primary focus lies in developing a numerically stable and efficient coordinate-descent algorithm to solve \eqref{eq:AML_cpr_weightedL1}. Leveraging the good property of the probit regression loss, as demonstrated in Lemma 2 of  \cite{zhou2023nonparametric}—specifically, the second derivative of the probit regression loss function is bounded by $1$—we can effectively fix the computer overflow issue by employing the minimization-majorization principle. This algorithm relies solely on the first-order information (gradient) of the composite probit loss function. Further details are provided below.

Let $(\hat{\bbt}^\textit{curr},\hat{\bm{b}}^\textit{curr})$ be the current estimate. Define the current margin $\hat{r}_{ki}^\textit{curr}=\check{y}_{ki}(\bx_i^{T}{\hat{\bbt}^\textit{curr}}-\hat{{b}}_k^\textit{curr})$ for $k=1,\dots,K$, $i=1,\dots,n$  and current adaptive weights
$\hat{\omega}_j^\textit{curr}=p^\prime_{\lambda}(\vert\hat{\beta}^\textit{curr}_j\vert)$ for $j\in \{1,\dots,p\}\backslash\cM$. 

To update the $\zeta$-th coordinate of $\bbt$, where $\zeta\in\{1,\cdots,p\}\backslash\cM$, define the $F$ function:
\begin{equation}
    F(\beta_{\zeta}\vert \hat{\bbt}^\textit{curr},\hat{\bm{b}}^\textit{curr})
    :=
    \sum_{k=1}^{K}w_k\left\{\frac{1}{n}\sum_{i=1}^{n}L(\check{y}_{ki}x_{i{\zeta}}(\beta_{\zeta}-{\hat{\beta}^\textit{curr}}_{\zeta}) + \hat{r}_{ki}^\textit{curr})\right\} + \hat{\omega}^\textit{curr}_{\zeta}\vert\beta_{\zeta}\vert.
\end{equation}
By Lemma 2 of \cite{zhou2023nonparametric} and $\check{y}_{ki}^2=1$, this function can be majorized by a penalized quadratic function defined as
\begin{align*}
    &~
    \mathcal{F}(\beta_{\zeta}\vert\hat{\bbt}^\textit{curr},\hat{\bm{b}}^\textit{curr}):=
    \\
    &
    \sum_{k=1}^{K}w_k\cdot\frac{1}{n}\sum_{i=1}^{n}\left[
	L(\hat{r}_{ki}^\textit{curr}) + L^\prime(\hat{r}_{ki}^\textit{curr})\check{y}_{ki}x_{i{\zeta}}(\beta_{\zeta}-{\hat{\beta}^\textit{curr}}_{\zeta}) + \frac{1}{2}x_{i{\zeta}}^2(\beta_{\zeta}-{\hat{\beta}^\textit{curr}}_{\zeta})^2
    \right] 
    + 
    \hat{\omega}^\textit{curr}_{\zeta}\vert\beta_{\zeta}\vert.
\end{align*}
We can easily solve the minimizer of the majorization function by a simple soft thresholding rule:
\begin{equation}
    {\hat{\beta}^\textit{new}}_{\zeta}
    =\mathcal{S}\left(
	{\hat{\beta}^\textit{curr}}_{\zeta}
    -
    \frac{
        \sum_{k=1}^{K}w_k\sum_{i=1}^{n}L^\prime(\hat{r}_{ki}^\textit{curr})\check{y}_{ki}x_{i{\zeta}}
        }{
        \sum_{i=1}^{n}x_{i{\zeta}}^2
        }
    ,
    \frac{\hat{\omega}^\textit{curr}_{\zeta}}{\frac{1}{n}\sum_{i=1}^{n}x_{i{\zeta}}^2}
	\right),
\end{equation}
where $\mathcal{S}(z,t)  = (\vert z\vert -t)_+\cdot\text{sgn}(z)$. The updating of the margins is given by
\begin{equation}
    \hat{r}_{ki}^\textit{new} = \hat{r}_{ki}^\textit{curr} + \cky_{ki}x_{i{\zeta}}\cdot({\hat{\beta}^\textit{new}}_{\zeta} - {\hat{\beta}^\textit{curr}}_{\zeta}),~
    k=1,\dots,K,~
    i=1,\dots,n.
    \label{eq:margin_update_from_beta}
\end{equation}

To update $\zeta$-th coordinate of $\bbt$, where $\zeta\in\cM$, define the $F$ function:
\begin{align*}
    F(\beta_{\zeta}\vert \hat{\bbt}^\textit{curr},\hat{\bm{b}}^\textit{curr})
    :=
    &~
    \sum_{k=1}^{K}w_k\left\{\frac{1}{n}\sum_{i=1}^{n}L(\check{y}_{ki}x_{i{\zeta}}(\beta_{\zeta}-{\hat{\beta}^\textit{curr}}_{\zeta}) + \hat{r}_{ki}^\textit{curr})\right\} 
    +
    \frac{\rho}{2}
    \left\Vert
        C_{\zeta}(\beta_{\zeta}-\hat{\beta}_{\zeta}^\textit{curr})
        +
        \frac{\hat{\tau}^{\textit{curr}}}{\rho}
    \right\Vert_2^2,
\end{align*}
where $C_{\zeta}$ is the column of matrix $\bC$ corresponding to coefficient $\beta_{\zeta}$, and $\hat{\tau}^{\textit{curr}}=\rho(\bC\hat{\bbt}_{\cM}^{\textit{curr}} - \bt) + {\hat{\bm{v}}}$.
This function can be majorized by the following quadratic function:
\begin{align*}
    \mathcal{F}(\beta_{\zeta}\vert\hat{\bbt}^\textit{curr},\hat{\bm{b}}^\textit{curr})
    :=
    &~
    \sum_{k=1}^{K}w_k\cdot\frac{1}{n}\sum_{i=1}^{n}\left[
	L(\hat{r}_{ki}^\textit{curr}) + L^\prime(\hat{r}_{ki}^\textit{curr})\check{y}_{ki}x_{i{\zeta}}(\beta_{\zeta}-{\hat{\beta}^\textit{curr}}_{\zeta}) + \frac{1}{2}x_{is}^2(\beta_{\zeta}-{\hat{\beta}^\textit{curr}}_{\zeta})^2
    \right] 
    \\
    +
    &~
    \frac{\rho}{2}
    \left\Vert
        C_{\zeta}(\beta_{\zeta}-\hat{\beta}_{\zeta}^\textit{curr})
        +
        \frac{\hat{\tau}^{\textit{curr}}}{\rho}
    \right\Vert_2^2.
\end{align*}
whose minimizer is given by:
\begin{equation}
    \hat{\beta}_{\zeta}^\textit{new}=	{\hat{\beta}^\textit{curr}}_{\zeta}
    -
    \frac{
        \sum_{k=1}^{K}w_k\frac{1}{n}\sum_{i=1}^{n}L^\prime(\hat{r}_{ki}^\textit{curr})\check{y}_{ki}x_{i{\zeta}}
        +
        C_{\zeta}^{T} \hat{\tau}^{\textit{curr}}
    }
    {
        \frac{1}{n}\sum_{i=1}^{n}x_{i{\zeta}}^2
        +
        \rho\Vert C_{\zeta}\Vert_2^2
    }.
    \label{eq:beta_update_w_penalty}
\end{equation}
The updating of margins is also given by equation \eqref{eq:margin_update_from_beta}, and the updating of slackness vector $\tau$ is given by
\begin{equation*}
    \hat{\tau}^{\textit{new}}=\hat{\tau}^{\textit{curr}} + \rho C_{\zeta}(\hat{\beta}_{\zeta}^\textit{new} -	{\hat{\beta}^\textit{curr}}_{\zeta} ).
\end{equation*}

Similarly, for the $\zeta$-th intercept $b_{\zeta}$, we consider minimizing the quadratic majorization:
\begin{equation}
	\mathcal{F}(b_{\zeta}\vert\hat{\bbt}^\textit{curr},\hat{\bm{b}}^\textit{curr})
        :=
        \frac{1}{n}\sum_{i=1}^{n}\left\{
	   L(\hat{r}_{{\zeta}i}^\textit{curr}) + L^\prime(\hat{r}_{{\zeta}i}^\textit{curr})(-\check{y}_{{\zeta}i})(b_{\zeta}-\hat{b}_{\zeta}^\textit{curr}) + \frac{1}{2}(b_{\zeta}-\hat{b}_{\zeta}^\textit{curr})^2
	\right\},
\end{equation}
which has a minimizer 
\begin{equation}
    \hat{b}_{\zeta}^\textit{new}=\hat{b}_{\zeta}^\textit{curr}+\frac{1}{n}\sum_{i=1}^{n}L^\prime(r_{{\zeta}i})\check{y}_{{\zeta}i}.
\end{equation}
The updating of margins is given by
\begin{equation}
    \hat{r}_{ki}^\textit{new} =
    \begin{cases}
        \hat{r}_{{\zeta}i}^\textit{curr} - \cky_{{\zeta}i}(\hat{b}_{\zeta}^\textit{new} - \hat{b}_{\zeta}^\textit{curr}),
        &
        \text{if } k=\zeta;
        \\
        \hat{r}_{ki}^\textit{curr},
        &
        \text{if } k\neq \zeta.
    \end{cases}
    \label{eq:margin_update_from_b}
\end{equation}

To sum up, the CMD algorithm for solving problem \eqref{eq:AML_cpr_weightedL1} in each LLA iteration is given in Algorithm \ref{alg:CMD_for_ALM}.

\begin{algorithm}[ht!]
	\caption{The CMD algorithm for solving \eqref{eq:AML_cpr_weightedL1}.}\label{alg:CMD_for_ALM}
	\begin{algorithmic}[1]
		\State Input the dual vector $\hat{\bm{v}}$, the adaptive weight vector $\hat{\bm{\omega}}$.
		\State Initialize $(\hat{\bbt}, \hat{\bm{b}})$ and compute its corresponding margin $\hat{r}_{ki}=\check{y}_{ki}(\bx_i^{T}\hat{\bbt}-\hat{b}_{k})$, for $k=1,\dots,K$, and $i=1,\dots,n$, and the vector $\hat{\tau}=\rho(\bC\hat{\bbt}_\cM-\bt) + \hat{\bm{v}}$.
		\State Iterate (3.a)-(3.c) until convergence:
		\begin{itemize}
		    \item[(3.a)] Cyclic coordinate descent for penalized coefficients: for $j \in \{1,\dots,p\}\backslash \cM$,
		    \begin{itemize}
		        \item[(3.a.1)] Compute
                    \begin{equation*}
				        \hat{\beta}_{j}^{\textit{new}} = 
                                \mathcal{S}\left(
	                        \hat{\beta}_{j}-\frac{\sum_{k=1}^{K}w_k\sum_{i=1}^{n}L^\prime(\hat{r}_{ki})\check{y}_{ki}x_{ij}}{\sum_{i=1}^{n}x_{ij}^2},\frac{\hat{\omega}_{j}}{\frac{1}{n}\sum_{i=1}^{n}x_{ij}^2}
	                        \right),
			        \end{equation*}
                    \item[(3.a.2)] Compute the new margins: $\hat{r}_{ki}^\textit{new} = \hat{r}_{ki} + \cky_{ki}x_{ij}\cdot({\hat{\beta}^\textit{new}}_{j} - \hat{\beta}_{j})$, 
                    for $k=1,\dots,K$, $i=1,\dots,n$.
			    \item[(3.a.3)] Update $\hat{\beta}_{j} = {\hat{\beta}_{j}}^\textit{new}$, and $\hat{r}_{ki} = \hat{r}_{ki}^\textit{new}$, for $k=1,\dots,K$, $i=1,\dots,n$.
                \end{itemize}
                \item[(3.b)] Cyclic coordinate descent for constrained coefficients: for $j \in \cM$,
                \begin{itemize}
                    \item[(3.b.1)] Compute
			        \begin{equation*}
				        \hat{\beta}_{j}^{\textit{new}} = 
                                {\hat{\beta}}_{j}
                                    -
                                    \frac{
                                        \sum_{k=1}^{K}w_k\frac{1}{n}\sum_{i=1}^{n}L^\prime(\hat{r}_{ki})\check{y}_{ki}x_{ij}
                                        +
                                        C_{j}^{T} \hat{\tau}
                                    }
                                    {
                                        \frac{1}{n}\sum_{i=1}^{n}x_{ij}^2
                                        +
                                        \rho\Vert C_{j}\Vert_2^2
                                    },
			        \end{equation*}
                    \item[(3.b.2)] Compute the new margins: $\hat{r}_{ki}^\textit{new} = \hat{r}_{ki} + \cky_{ki}x_{ij}\cdot({\hat{\beta}^\textit{new}}_{j} - \hat{\beta}_{j})$, 
                    for $k=1,\dots,K$, $i=1,\dots,n$, and the vector: $\hat{\tau}^{\textit{new}}=\hat{\tau} + \rho C_{j}(\hat{\beta}_{j}^\textit{new} -	{\hat{\beta}}_{j} )$.
			    \item[(3.b.3)] Update $\hat{\beta}_{j} = {\hat{\beta}_{j}}^\textit{new}$, $\hat{r}_{ki} = \hat{r}_{ki}^\textit{new}$, for $k=1,\dots,K$, $i=1,\dots,n$, 
                    and $\hat{\tau}=\hat{\tau}^{\textit{new}}$.
		    \end{itemize}
		    \item[(3.c)] Cyclic coordinate descent for intercepts: for $k = 1, 2, \cdots, K$,
		    \begin{itemize}
		        \item[(3.c.1)] Compute
			        \begin{equation*}
				        \hat{b}_{k}^\textit{new}=\hat{b}_{k}+\frac{1}{n}\sum_{i=1}^{n}L^\prime(r_{ki})\check{y}_{ki}.
			        \end{equation*}
                    \item[(3.c.2)] Compute the new margins: $\hat{r}_{ki}^\textit{new} = \hat{r}_{ki} - \cky_{ki}(\hat{b}_{k}^\textit{new} - \hat{b}_k),$ for $i=1,\dots,n$.
			    \item[(3.c.3)] Update $\hat{b}_{k} = \hat{b}_{k}^\textit{new}$, and $\hat{r}_{ki}= \hat{r}_{ki}^\textit{new}$, for $i=1,\dots,n$.
		    \end{itemize}
		\end{itemize} 
	\end{algorithmic}
\end{algorithm}

\section{Numerical Studies}\label{sec:simulations}
In this section, we examine the finite sample performance of the proposed partial penalized composite likelihood ratio, Wald, and score tests in a variety of settings.

In our implementation, we use the SCAD penalty with $a=3.7$ for the folded-concave penalization, and set $\rho=1$ for the augmented Lagrangian parameter.

To obtain $\hat{\bbt}_0$ and $\hat{\bbt}_a$, we compute $\hat{\bbt}_0^{\lambda}$ and $\hat{\bbt}_a^{\lambda}$ for a series of log-spaced values in $[\lambda_{\min},\lambda_{\max}]$. Then we choose $\hat{\bbt}_0=\hat{\bbt}_0^{\hat{\lambda}_0}$ and $\hat{\bbt}_a=\hat{\bbt}_a^{\hat{\lambda}_a}$, where $\hat{\lambda}_0$ and $\hat{\lambda}_a$ are the minimizers of the following information criterion:
\begin{align*}
    \hat{\lambda}_0
    &=
    \arg\min_{\lambda}
    \left(
    -nM_n(\hat{\cB}_0^{\lambda})
    +
    c_n\left\Vert\hat{\bbt}_0^{\lambda}\right\Vert_0
    \right),
    \\
    \hat{\lambda}_a
    &=
    \arg\min_{\lambda}
    \left(
    -nM_n(\hat{\cB}_a^{\lambda})
    +
    c_n\left\Vert\hat{\bbt}_a^{\lambda}\right\Vert_0
    \right),
\end{align*}
where $c_n=\max\{\log(n),\log(\log(n))\log(p)\}$, which takes the larger one between the Bayesian Information Criterion  \citep{schwarz1978estimating} and the Generalized Information Criterion (GIC) proposed by \citep{schwarz1978estimating,fan2013tuning} for ultra-high dimensional cases where $p$ increases exponentially with $n$.

\paragraph{Simulation Settings}

Simulated data were generated from
\begin{equation*}
    g(Y) = 2X_1 - (2+h_1)X_2 + \varepsilon,
\end{equation*}
where $\varepsilon\sim N(0,1)$ and $\bX\sim N(\bm{0}_p,\Sigma)$, $g(\cdot)$ is some monotone increasing function, and $h_1$ is a constant. The true regression coefficient is $\bbt^*=(2, -(2+h_1), \mathbf{0}_{p-2})$. We focus on testing the following four pairs of hypotheses
\begin{align*}
    H_0^{(i)}:
    \beta_1+\beta_2 = 0,
    \quad &v.s.\quad
    H_a^{(i)}:
    \beta_1+\beta_2 \neq 0;
    \\
    H_0^{(ii)}:
    \beta_2 = -2,
    \quad &v.s.\quad
    H_a^{(ii)}:
    \beta_2 \neq -2;
    \\
    H_0^{(iii)}:
    \sum_{j=1}^{4}\beta_{j}= 0,
    \quad &v.s.\quad
    H_a^{(iii)}:
    \sum_{j=1}^{4}\beta_{j}\neq 0;
    \\
    H_0^{(iv)}:
    \begin{pmatrix}
        1&1&0&0\\
        0&1&0&0\\
        1&1&1&1
    \end{pmatrix}
    \begin{pmatrix}
        \beta_1\\
        \beta_2\\
        \beta_3\\
        \beta_4
    \end{pmatrix}
    =
    \begin{pmatrix}
        0\\
        -2\\
        0
    \end{pmatrix},
    \quad &v.s.\quad
    H_a^{(iv)}:
    \begin{pmatrix}
        1&1&0&0\\
        0&1&0&0\\
        1&1&1&1
    \end{pmatrix}
    \begin{pmatrix}
        \beta_1\\
        \beta_2\\
        \beta_3\\
        \beta_4
    \end{pmatrix}
    \neq
    \begin{pmatrix}
        0\\
        -2\\
        0
    \end{pmatrix};
\end{align*}
with significance level $\alpha=0.05$. Notice that for all four pairs of testing hypotheses given above, the null hypotheses hold if and only if $h_1=0$. And the fourth hypothesis $H_0^{(iv)}$ is a multiple testing involving the simultaneous testing of the first three hypotheses.

The number of observations has two choices, $n=200$ and $n=400$. For the dimension of the regression coefficients, we consider two different cases for each value of $n$, $p=1.25n$ and $p=2.5n$. For the covariance matrix $\Sigma$ when generating the design matrix, we consider the first-order autoregressive
structure $\Sigma = \left( \rho^{\vert i-j\vert} \right)_{p\times p}$ with $\rho=0.5$ or $0.8$. We consider the monotone transformation function: $g_1^{-1}(t) = \frac{1}{2}(2t-1)^{1/3}+\frac{1}{2}$ and $g_2^{-1}(t) = \frac{\exp((t-2)/3)}{1+0.5((t-2)/3)^2}$. This yields a total of 16 settings. For each hypothesis and each setting, we further consider four scenarios, by setting $h_1 =
0,0.1,0.2,0.4$. Therefore, the null holds when $h_1=0$ and the alternative holds when $h_1=0.1,0.2,0.4$. 

When using the composite probit regression method to construct the testing statistics under our non-parametric Box-Cox model, we 
set the number of thresholds to be $K=19$ and use the $5\%,~10\%,~15\%,\dots,~95\%$ empirical percentiles of the responses to dichotomize the response variable and combine the $19$ different probit regression models with the equal weights. Since the composite probit regression method for regression coefficient estimation is invariant against the unknown monotone transformation function, the corresponding rejection probabilities would also be invariant against any $g(\cdot)$. Therefore, we could combine the simulation results under $g_1(\cdot)$ and $g_2(\cdot)$, degenerating to 8 data-generating settings. The rejection probabilities for $H_0^{(i)}$, $H_0^{(ii)}$, $H_0^{(iii)}$ and $H_0^{(iv)}$ under the settings with $n=200$ and $400$, and $\Sigma = \left( 0.5^{\vert i-j\vert} \right)_{p\times p}$ and $\left( 0.8^{\vert i-j\vert} \right)_{p\times p}$ under the non-parametric Box-Cox model assumption are summarzied in Table \ref{table:rejection_probabilities_NBC_200_0.5} -- Table \ref{table:rejection_probabilities_NBC_400_0.8}.

Based on the results, it can be seen that under these null hypotheses, the Type-I error rates of the three tests are well controlled and close to the significance level $\alpha=0.05$. Under the alternative hypotheses, the powers of these three test statistics increase as $h_1$ increases, showing the consistency of our testing procedure. Moreover, the empirical rejection rates among these three test statistics are very close across all different scenarios and settings. This is consistent with our findings in Theorem \ref{theorem:testing_statistic_distribution} that these statistics are asymptotically equivalent.
One interesting observation is that under the alternative hypotheses, the estimated powers are different for different scenarios and settings, differing from the findings in \cite{shi2019linear}. This is because we prove in Theorem \ref{theorem:testing_statistic_distribution} that our test statistics are asymptotically equivalent to a generalized chi-squared distribution, with the shape of the ellipsoid determined by the design matrix $\bX$, the true regression coefficients, and the testing matrix $\bC$. Consequently, the power is determined by various factors such as the data we generated and the hypothesis we formulated.

\paragraph{Model Mis-specification:}  Under the assumption that the underlying true model is the linear model, we can apply the testing procedures for (generalized) linear model  \citep{shi2019linear} on our simulated data sets. The rejection probabilities for $H_0^{(i)}$, $H_0^{(ii)}$, $H_0^{(iii)}$ and $H_0^{(iv)}$ under the setting where $\Sigma = \left( 0.5^{\vert i-j\vert} \right)_{p\times p}$ and $p=250$ summarzied in Table \ref{table:rejection_probabilities_LS_0.5}. One can see that under model mis-specification, we might make wrong conclusions.

\begin{table}[ht!]
	\centering
	\caption{Rejection probabilities $(\%)$ of the partial penalized likelihood ratio, wald and score statistics with standard errors in parenthesis $(\%)$, under the setting where $n=200$, $\Sigma=(0.5^{\vert i-j\vert})_{(p\times p)}$.}
    \label{table:rejection_probabilities_NBC_200_0.5}
	\begin{tabular}{lcccccc}
		\toprule
		\multirow{2}{*}{}&\multicolumn{3}{c}{ $p=250$}&\multicolumn{3}{c}{$p=500$} \\
            \cmidrule(lr){2-4} \cmidrule(lr){5-7}
		& $T_L$ & $T_S$ & $T_W$ & $T_L$ & $T_S$ & $T_W$ \\
		\midrule
		$h_1$ & \multicolumn{6}{c}{$H_0^{(i)}$}\\
            \cmidrule{2-7}
	    $0$ & $~6.50(1.01)$ & $~6.50(1.01)$ & $~6.33(0.99)$ & $~6.50(1.01)$ & $~6.50(1.01)$ & $~6.33(0.99)$ \\
            $0.1$ & $18.00(1.57)$ & $18.17(1.57)$ & $17.83(1.56)$ & $21.83(1.68)$ & $21.67(1.68)$ & $21.50(1.68)$ \\
            $0.2$ & $64.17(1.96)$ & $63.83(1.96)$ & $64.17(1.96)$ & $63.00(1.97)$ & $62.83(1.97)$ & $62.67(1.97)$ \\
            $0.4$ & $99.00(0.41)$ & $98.83(0.44)$ & $99.00(0.41)$ & $99.50(0.29)$ & $99.50(0.29)$ & $99.50(0.29)$ \\
		\midrule
            $h_1$ & \multicolumn{6}{c}{$H_0^{(ii)}$}\\
            \cmidrule{2-7}
	    $0$ & $~5.00(0.89)$ & $~5.00(0.89)$ & $~5.00(0.89)$ & $~5.33(0.92)$ & $~5.00(0.89)$ & $~5.33(0.92)$ \\
            $0.1$ & $14.00(1.42)$ & $11.50(1.30)$ & $14.00(1.42)$ & $12.67(1.36)$ & $10.00(1.22)$ & $12.17(1.33)$ \\
            $0.2$ & $28.67(1.85)$ & $23.67(1.74)$ & $29.67(1.86)$ & $30.33(1.88)$ & $25.67(1.78)$ & $30.67(1.88)$ \\
            $0.4$ & $77.50(1.70)$ & $71.83(1.84)$ & $83.00(1.53)$ & $75.50(1.76)$ & $71.67(1.84)$ & $80.83(1.61)$ \\
		\midrule
            $h_1$ & \multicolumn{6}{c}{$H_0^{(iii)}$}\\
            \cmidrule{2-7}
	    $0$ & $~4.83(0.88)$ & $~4.83(0.88)$ & $~4.83(0.88)$ & $~4.67(0.86)$ & $~4.67(0.86)$ & $~4.67(0.86)$ \\
            $0.1$ & $16.67(1.52)$ & $16.67(1.52)$ & $16.00(1.50)$ & $18.33(1.58)$ & $18.33(1.58)$ & $17.50(1.55)$ \\
            $0.2$ & $46.83(2.04)$ & $46.33(2.04)$ & $46.17(2.04)$ & $46.00(2.03)$ & $45.67(2.03)$ & $45.83(2.03)$ \\
            $0.4$ & $94.67(0.92)$ & $94.83(0.90)$ & $94.50(0.93)$ & $95.17(0.88)$ & $95.17(0.88)$ & $95.17(0.88)$ \\
		\midrule
            $h_1$ & \multicolumn{6}{c}{$H_0^{(iv)}$}\\
            \cmidrule{2-7}
	    $0$ & $~6.50(1.01)$ & $~6.00(0.97)$ & $~6.00(0.97)$ & $~6.17(0.98)$ & $~6.17(0.98)$ & $~6.00(0.97)$ \\
            $0.1$ & $17.50(1.55)$ & $15.00(1.46)$ & $15.67(1.48)$ & $17.00(1.53)$ & $15.00(1.45)$ & $15.33(1.47)$ \\
            $0.2$ & $46.50(2.04)$ & $44.83(2.03)$ & $45.33(2.03)$ & $48.17(2.04)$ & $44.33(2.03)$ & $45.17(2.03)$ \\
            $0.4$ & $97.00(0.70)$ & $96.33(0.77)$ & $96.83(0.71)$ & $98.00(0.57)$ & $98.33(0.52)$ & $98.17(0.55)$ \\
		\bottomrule
	\end{tabular}
\end{table}

\begin{table}[ht!]
	\centering
	\caption{Rejection probabilities $(\%)$ of the partial penalized likelihood ratio, wald and score statistics with standard errors in parenthesis $(\%)$, under the setting where $n=400$, $\Sigma=(0.5^{\vert i-j\vert})_{(p\times p)}$.}
    \label{table:rejection_probabilities_NBC_400_0.5}
	\begin{tabular}{lcccccc}
		\toprule
		\multirow{2}{*}{}&\multicolumn{3}{c}{ $p=500$}&\multicolumn{3}{c}{$p=1000$} \\
            \cmidrule(lr){2-4} \cmidrule(lr){5-7}
		& $T_L$ & $T_S$ & $T_W$ & $T_L$ & $T_S$ & $T_W$ \\
		\midrule
		$h_1$ & \multicolumn{6}{c}{$H_0^{(i)}$}\\
            \cmidrule{2-7}
	    $0$ & $~3.33(0.73)$ & $~3.50(0.75)$ & $~3.17(0.71)$ & $~4.33(0.83)$ & $~4.17(0.81)$ & $~4.33(0.83)$ \\
            $0.1$ & $37.50(1.97)$ & $37.33(1.97)$ & $37.50(1.97)$ & $37.50(1.97)$ & $37.33(1.97)$ & $37.50(1.97)$ \\
            $0.2$ & $88.00(1.32)$ & $88.00(1.32)$ & $88.17(1.32)$ & $87.00(1.37)$ & $87.00(1.37)$ & $87.00(1.37)$ \\
            $0.4$ & $100.00(0.00)$ & $100.00(0.00)$ & $100.00(0.00)$ & $100.00(0.00)$ & $100.00(0.00)$ & $100.00(0.00)$ \\
		\midrule
            $h_1$ & \multicolumn{6}{c}{$H_0^{(ii)}$}\\
            \cmidrule{2-7}
	    $0$ & $~6.33(0.99)$ & $~6.50(1.00)$ & $~6.67(1.01)$ & $~3.33(0.73)$ & $~3.50(0.75)$ & $~3.50(0.75)$ \\
            $0.1$ & $16.67(1.52)$ & $14.67(1.44)$ & $17.00(1.53)$ & $14.83(1.45)$ & $12.83(1.36)$ & $14.33(1.43)$ \\
            $0.2$ & $53.83(2.03)$ & $48.50(2.04)$ & $54.50(2.03)$ & $47.00(2.03)$ & $44.00(2.03)$ & $48.83(2.04)$ \\
            $0.4$ & $96.33(0.76)$ & $95.50(0.84)$ & $97.33(0.65)$ & $97.00(0.69)$ & $96.67(0.83)$ & $97.83(0.59)$ \\
		\midrule
            $h_1$ & \multicolumn{6}{c}{$H_0^{(iii)}$}\\
            \cmidrule{2-7}
	    $0$ & $~3.83(0.78)$ & $~3.83(0.78)$ & $~3.83(0.78)$ & $~3.50(0.75)$ & $~3.33(0.73)$ & $~3.50(0.75)$ \\
            $0.1$ & $26.00(1.79)$ & $25.50(1.78)$ & $26.00(1.79)$ & $26.33(1.79)$ & $25.67(1.78)$ & $26.50(1.80)$ \\
            $0.2$ & $73.83(1.79)$ & $73.50(1.80)$ & $73.50(1.80)$ & $73.00(1.81)$ & $72.33(1.82)$ & $72.67(1.82)$ \\
            $0.4$ & $100.00(0.00)$ & $100.00(0.00)$ & $100.00(0.00)$ & $100.00(0.00)$ & $100.00(0.00)$ & $100.00(0.00)$ \\
		\midrule
            $h_1$ & \multicolumn{6}{c}{$H_0^{(iv)}$}\\
            \cmidrule{2-7}
	    $0$ & $~6.33(0.99)$ & $~5.50(0.93)$ & $~5.50(0.93)$ & $~3.00(0.69)$ & $~3.17(0.71)$ & $~2.83(0.67)$ \\
            $0.1$ & $27.33(1.82)$ & $25.50(1.78)$ & $25.50(1.78)$ & $25.00(1.76)$ & $23.83(1.74)$ & $23.83(1.74)$ \\
            $0.2$ & $79.33(1.65)$ & $77.33(1.71)$ & $78.00(1.69)$ & $77.83(1.69)$ & $76.83(1.72)$ & $76.67(1.72)$ \\
            $0.4$ & $100.00(0.00)$ & $100.00(0.00)$ & $100.00(0.00)$ & $100.00(0.00)$ & $100.00(0.00)$ & $100.00(0.00)$ \\
		\bottomrule
	\end{tabular}
\end{table}

\begin{table}[ht!]
	\centering
	\caption{Rejection probabilities $(\%)$ of the partial penalized likelihood ratio, wald and score statistics with standard errors in parenthesis $(\%)$, under the setting where $n=200$, $\Sigma=(0.8^{\vert i-j\vert})_{(p\times p)}$.}
    \label{table:rejection_probabilities_NBC_200_0.8}
	\begin{tabular}{lcccccc}
		\toprule
		\multirow{2}{*}{}&\multicolumn{3}{c}{ $p=250$}&\multicolumn{3}{c}{$p=500$} \\
            \cmidrule(lr){2-4} \cmidrule(lr){5-7}
		& $T_L$ & $T_S$ & $T_W$ & $T_L$ & $T_S$ & $T_W$ \\
		\midrule
		$h_1$ & \multicolumn{6}{c}{$H_0^{(i)}$}\\
            \cmidrule{2-7}
	    $0$ & $~5.67(0.94)$ & $~5.33(0.92)$ & $~5.67(0.94)$ & $~4.17(0.81)$ & $~4.17(0.81)$ & $~3.83(0.78)$ \\
            $0.1$ & $26.17(1.79)$ & $25.50(1.78)$ & $25.67(1.78)$ & $24.50(1.75)$ & $23.83(1.74)$ & $24.50(1.75)$ \\
            $0.2$ & $71.33(1.85)$ & $71.17(1.85)$ & $71.00(1.85)$ & $71.33(1.85)$ & $71.17(1.85)$ & $71.00(1.85)$ \\
            $0.4$ & $99.83(0.17)$ & $99.83(0.17)$ & $99.83(0.17)$ & $99.83(0.17)$ & $99.83(0.17)$ & $99.83(0.17)$ \\
		\midrule
            $h_1$ & \multicolumn{6}{c}{$H_0^{(ii)}$}\\
            \cmidrule{2-7}             
	    $0$ & $~6.50(1.01)$ & $~6.33(0.99)$ & $~5.83(0.96)$ & $~5.83(0.96)$ & $~5.67(0.94)$ & $~5.33(0.92)$ \\
            $0.1$ & $~8.50(1.14)$ & $~6.50(1.01)$ & $~8.33(1.13)$ & $~9.67(1.21)$ & $~6.83(1.03) $ & $~8.83(1.16)$ \\
            $0.2$ & $26.67(1.80)$ & $22.33(1.70)$ & $28.00(1.83)$ & $22.83(1.71)$ & $19.00(1.60)$ & $23.67(1.73)$ \\
            $0.4$ & $66.17(1.93)$ & $60.00(2.00)$ & $70.50(1.86)$ & $66.83(1.92)$ & $62.83(1.97)$ & $70.33(1.86)$ \\
		\midrule      
            $h_1$ & \multicolumn{6}{c}{$H_0^{(iii)}$}\\
            \cmidrule{2-7}
	    $0$ & $~4.17(0.82)$ & $~4.17(0.82)$ & $~4.00(0.80)$ & $~4.00(0.80)$ & $~4.00(0.80)$ & $~3.83(0.78)$ \\
            $0.1$ & $24.83(1.76)$ & $24.50(1.75)$ & $24.00(1.74)$ & $21.33(1.67)$ & $20.83(1.66)$ & $21.00(1.66)$ \\
            $0.2$ & $62.50(1.98)$ & $62.67(1.97)$ & $61.83(1.98)$ & $64.50(1.95)$ & $64.00(1.96)$ & $64.00(1.96)$ \\     
            $0.4$ & $99.33(0.33)$ & $99.50(0.29)$ & $99.33(0.33)$ & $99.17(0.37)$ & $99.17(0.37)$ & $99.17(0.37)$ \\
		\midrule
            $h_1$ & \multicolumn{6}{c}{$H_0^{(iv)}$}\\
            \cmidrule{2-7}     
	    $0$ & $~6.33(0.99)$ & $~5.83(0.95)$ & $~5.67(0.94)$ & $~4.67(0.86)$ & $~4.00(0.80)$ & $~3.67(0.77)$ \\
            $0.1$ & $16.83(1.53)$ & $16.33(1.51)$ & $16.17(1.50)$ & $14.67(1.44)$ & $13.17(1.38)$ & $13.50(1.39)$ \\
            $0.2$ & $55.00(2.03)$ & $52.83(2.04)$ & $53.33(2.04)$ & $57.33(2.02)$ & $54.83(2.03)$ & $55.67(2.03)$ \\      
            $0.4$ & $99.33(0.33)$ & $99.00(0.41)$ & $99.00(0.41)$ & $99.33(0.33)$ & $99.17(0.37)$ & $99.17(0.37)$ \\
		\bottomrule
	\end{tabular}
\end{table}

\begin{table}[ht!]
	\centering
	\caption{Rejection probabilities $(\%)$ of the partial penalized likelihood ratio, wald and score statistics with standard errors in parenthesis $(\%)$, under the setting where $n=400$, $\Sigma=(0.8^{\vert i-j\vert})_{(p\times p)}$.}
    \label{table:rejection_probabilities_NBC_400_0.8}
	\begin{tabular}{lcccccc}
		\toprule
		\multirow{2}{*}{}&\multicolumn{3}{c}{ $p=500$}&\multicolumn{3}{c}{$p=1000$} \\
            \cmidrule(lr){2-4} \cmidrule(lr){5-7}
		& $T_L$ & $T_S$ & $T_W$ & $T_L$ & $T_S$ & $T_W$ \\
		\midrule
		$h_1$ & \multicolumn{6}{c}{$H_0^{(i)}$}\\
            \cmidrule{2-7}
	    $0$ & $~5.00(0.89)$ & $~4.83(0.87)$ & $~5.00(0.89)$ & $~6.17(0.98)$ & $~6.00(0.97)$ & $~6.00(0.97)$ \\
            $0.1$ & $45.17(2.03)$ & $45.17(2.03)$ & $45.00(2.03)$ & $42.00(2.01)$ & $42.00(2.01)$ & $42.00(2.01)$ \\
            $0.2$ & $96.67(0.73)$ & $96.50(0.75)$ & $96.33(0.76)$ & $93.67(0.99)$ & $93.67(0.99)$ & $93.67(0.99)$ \\
            $0.4$ & $100.00(0.00)$ & $100.00(0.00)$ & $100.00(0.00)$ & $100.00(0.00)$ & $100.00(0.00)$ & $100.00(0.00)$ \\
		\midrule
            $h_1$ & \multicolumn{6}{c}{$H_0^{(ii)}$}\\
            \cmidrule{2-7}             
	    $0$ & $~4.83(0.87)$ & $~4.83(0.87)$ & $~4.67(0.86)$ & $~5.00(0.89)$ & $~5.17(0.90)$ & $~5.33(0.92)$ \\
            $0.1$ & $15.00(1.45)$ & $12.00(1.32)$ & $13.50(1.39)$ & $15.00(1.45)$ & $13.33(1.38) $ & $15.33(1.47)$ \\
            $0.2$ & $38.17(1.98)$ & $35.17(1.95)$ & $39.33(1.99)$ & $37.67(1.98)$ & $33.50(1.93)$ & $39.33(1.99)$ \\
            $0.4$ & $91.67(1.12)$ & $89.17(1.26)$ & $93.83(0.98)$ & $91.00(1.16)$ & $89.17(1.58)$ & $93.00(1.04)$ \\
		\midrule      
            $h_1$ & \multicolumn{6}{c}{$H_0^{(iii)}$}\\
            \cmidrule{2-7}
	    $0$ & $~4.83(0.87)$ & $~4.83(0.87)$ & $~4.83(0.87)$ & $~6.00(0.97)$ & $~5.83(0.95)$ & $~6.00(0.97)$ \\
            $0.1$ & $38.17(1.76)$ & $38.33(1.98)$ & $38.00(1.98)$ & $36.67(1.98)$ & $36.00(1.96)$ & $36.33(1.96)$ \\
            $0.2$ & $91.00(1.17)$ & $90.83(1.18)$ & $90.67(1.18)$ & $88.00(1.32)$ & $88.00(1.32)$ & $87.67(1.34)$ \\     
            $0.4$ & $100.00(0.00)$ & $100.00(0.00)$ & $100.00(0.00)$ & $100.00(0.00)$ & $100.00(0.00)$ & $100.00(0.00)$ \\
		\midrule
            $h_1$ & \multicolumn{6}{c}{$H_0^{(iv)}$}\\
            \cmidrule{2-7}     
	    $0$ & $~5.50(0.93)$ & $~5.17(0.90)$ & $~5.17(0.90)$ & $~4.50(0.84)$ & $~4.83(0.87)$ & $~4.50(0.84)$ \\
            $0.1$ & $33.33(1.92)$ & $31.00(1.88)$ & $32.00(1.90)$ & $31.88(1.90)$ & $31.33(1.89)$ & $31.17(1.89)$ \\
            $0.2$ & $86.50(1.39)$ & $86.00(1.41)$ & $85.67(1.43)$ & $83.33(1.52)$ & $81.67(1.58)$ & $82.17(1.56)$ \\      
            $0.4$ & $100.00(0.00)$ & $100.00(0.00)$ & $100.00(0.00)$ & $100.00(0.00)$ & $100.00(0.00)$ & $100.00(0.00)$ \\
		\bottomrule
	\end{tabular}
\end{table}

\begin{table}[ht!]
	\centering
	\caption{Model mis-specification (assuming linear model ignoring transformation on response): Rejection probabilities $(\%)$ of the partial penalized likelihood ratio, wald and score statistics with standard errors in parenthesis $(\%)$, under the setting where $\Sigma=(0.5^{\vert i-j\vert})_{(p\times p)}$, $p=250$.}
    \label{table:rejection_probabilities_LS_0.5}
	\begin{tabular}{lcccccc}
		\toprule
		\multirow{2}{*}{}&\multicolumn{3}{c}{ $g_1$}&\multicolumn{3}{c}{$g_2$} \\
            \cmidrule(lr){2-4} \cmidrule(lr){5-7}
		& $T_L$ & $T_S$ & $T_W$ & $T_L$ & $T_S$ & $T_W$ \\
		\midrule
		$h_1$ & \multicolumn{6}{c}{$H_0^{(i)}$}\\
            \cmidrule{2-7}
	    $0$ & $~~0.00(0.00)$ & $~~0.00(0.00)$ & $~~0.00(0.00)$ & $~~0.00(0.00)$ & $~~0.00(0.00)$ & $~~0.00(0.00)$ \\
            $0.1$ & $~~0.00(0.00)$ & $~~0.00(0.00)$ & $~~0.00(0.00)$ & $~~0.00(0.00)$ & $~~0.00(0.00)$ & $~~0.00(0.00)$ \\
            $0.2$ & $~~0.00(0.00)$ & $~~0.00(0.00)$ & $~~0.00(0.00)$ & $~~0.00(0.00)$ & $~~0.00(0.00)$ & $~~0.00(0.00)$ \\
            $0.4$ & $~~8.83(1.16)$ & $~~8.83(1.16)$ & $~~8.83(1.16)$ & $~~0.00(0.00)$ & $~~0.00(0.00)$ & $~~0.00(0.00)$ \\
		\midrule
            $h_1$ & \multicolumn{6}{c}{$H_0^{(ii)}$}\\
            \cmidrule{2-7}
	    $0$ & $100.00(0.00)$ & $100.00(0.00)$ & $100.00(0.00)$ & $100.00(0.00)$ & $100.00(0.00)$ & $100.00(0.00)$ \\
            $0.1$ & $100.00(0.00)$ & $100.00(0.00)$ & $100.00(0.00)$ & $100.00(0.00)$ & $100.00(0.00)$ & $100.00(0.00)$ \\
            $0.2$ & $100.00(0.00)$ & $100.00(0.00)$ & $100.00(0.00)$ & $100.00(0.00)$ & $100.00(0.00)$ & $100.00(0.00)$ \\
            $0.4$ & $100.00(0.00)$ & $100.00(0.00)$ & $100.00(0.00)$ & $100.00(0.00)$ & $100.00(0.00)$ & $100.00(0.00)$ \\
		\midrule
            $h_1$ & \multicolumn{6}{c}{$H_0^{(iii)}$}\\
            \cmidrule{2-7}
	    $0$ & $~~0.00(0.00)$ & $~~0.00(0.00)$ & $~~0.00(0.00)$ & $~~0.00(0.00)$ & $~~0.00(0.00)$ & $~~0.00(0.00)$ \\
            $0.1$ & $~~0.00(0.00)$ & $~~0.00(0.00)$ & $~~0.00(0.00)$ & $~~0.00(0.00)$ & $~~0.00(0.00)$ & $~~0.00(0.00)$ \\
            $0.2$ & $~~0.00(0.00)$ & $~~0.00(0.00)$ & $~~0.00(0.00)$ & $~~0.00(0.00)$ & $~~0.00(0.00)$ & $~~0.00(0.00)$ \\
            $0.4$ & $~~1.67(0.52)$ & $~~1.67(0.52)$ & $~~1.67(0.52)$ & $~~0.00(0.00)$ & $~~0.00(0.00)$ & $~~0.00(0.00)$ \\
		\midrule
            $h_1$ & \multicolumn{6}{c}{$H_0^{(iv)}$}\\
            \cmidrule{2-7}
	    $0$ & $100.00(0.00)$ & $100.00(0.00)$ & $100.00(0.00)$ & $100.00(0.00)$ & $100.00(0.00)$ & $100.00(0.00)$ \\
            $0.1$ & $100.00(0.00)$ & $100.00(0.00)$ & $100.00(0.00)$ & $100.00(0.00)$ & $100.00(0.00)$ & $100.00(0.00)$ \\
            $0.2$ & $100.00(0.00)$ & $100.00(0.00)$ & $100.00(0.00)$ & $100.00(0.00)$ & $100.00(0.00)$ & $100.00(0.00)$ \\
            $0.4$ & $100.00(0.00)$ & $100.00(0.00)$ & $100.00(0.00)$ & $100.00(0.00)$ & $100.00(0.00)$ & $100.00(0.00)$ \\
		\bottomrule
	\end{tabular}
\end{table}

\clearpage

\section{Discussion}\label{sec:discussion}

In this article, we are concerned about the challenges of hypothesis testing in high-dimensional regression, where the true underlying model may not be the typical low-dimensional linear regression model. To address issues of model mis-specification, we have introduced the non-parametric Box-Cox model with unspecified monotone transformation on the response as an alternative approach to robustify the testing procedures. This non-parametric Box-Cox framework provides a more flexible setting than the typical normal linear models for high-dimensional linear hypothesis testing while preserving the interpretability of predictors. 
The partial penalized composite probit regression estimation has been proposed to construct our partial penalized composite likelihood ratio, score, and Wald tests without knowing the non-parametric transformation function. Additionally, we have developed an efficient and stable algorithm to compute the estimators with linear constraints. Simulation studies have demonstrated the efficacy of our proposed testing procedures in both controlling Type-I error rates and achieving non-negligible power under local alternatives. Both numerical and empirical studies have shown the discrepancies between our method and standard high-dimensional methods, highlighting the importance of our robust approach.

There are several worthy discussion points. Following the context of composite likelihood inference, one can construct the composite likelihood versions of Wald and score statistics based on the Godambe information matrix \citep{godambe1960optimum} similar to the test statistics constructed for testing $H_0:\bbt=\bbt_0$ in \citep{molenberghs2005models}, and they will have the usual asymptotic chi-squared distribution. However, they are not asymptotically equivalent to the composite likelihood likelihood ratio statistic. For testing linear hypothesis $H_0:\bC\bbt_{0,\cM}=\bt$, the Wald-type statistic is given by
\begin{equation}
    \tilde{T}_W 
    = 
    n
    (\bm{C}\hat{\bbt}_{a,\cM}-\bm{t})^T
    (
        \bm{C}
        {\widehat{\cT}}_{a,mm}
        \bm{C}^T
    )^{-1}
    (\bm{C}\hat{\bbt}_{a,\cM}-\bm{t}),
\end{equation}
where ${\widehat{\cT}}_{a,mm}$ is the $m\times m$ submatrix of the Godambe information substitution estimator ${\widehat{\cT}}_{n,a}$
and the score-type statistic is given by
\begin{equation}
    T_S
    =
    \frac{1}{n}
    \mathcal{S}_n(\hat{\cB}_0)^T
    \widehat{\bK}_{n,0}^{-1}
    \begin{pmatrix}
                    \bC^T 
                            \\
                            \bm{0}_{r\times (s+K)}^T
                \end{pmatrix}
    \left(
    \begin{pmatrix}
                    \bC^T 
                            \\
                            \bm{0}_{r\times (s+K)}^T
                \end{pmatrix}^T
            \widehat{\bK}_{n,0}^{-1}
            \widehat{\bm{V}}_{n,0}
            \widehat{\bK}_{n,0}^{-1}
            \begin{pmatrix}
                    \bC^T 
                            \\
                            \bm{0}_{r\times (s+K)}^T
                \end{pmatrix}
    \right)^{-1}
    \begin{pmatrix}
                    \bC^T 
                            \\
                            \bm{0}_{r\times (s+K)}^T
                \end{pmatrix}^T
                \widehat{\bK}_{n,0}^{-1}
    \mathcal{S}_n(\hat{\cB}_0),
\end{equation}
both of which can be proved to have asymptotic chi-squared distribution.

To fix ideas, we have used the equal weights in composite probit regression in both the numerical and empirical studies. While our current findings are encouraging, it is interesting to study the optimal weights that maximize the estimation efficiency for a given problem. Estimators with higher efficiency typically exhibit greater power in hypothesis testing. One potential approach involves deriving the asymptotic covariance matrix of the estimator based on its asymptotic distribution for any choice of weights, followed by identifying the optimal choice of weights that minimize the trace of that matrix. A similar problem was studied in \cite{Bradic2011} for penalized composite quasi-likelihood estimators. This aspect requires further investigation and could be a topic for future research efforts.

\bibliographystyle{agsm}
\setcitestyle{authoryear,open={},close={}}
\bibliography{main}

\clearpage
\appendix
\section{Appendix / supplemental material}
\subsection{Lemma}
\begin{lemma}\label{lemma0}
    Suppose $\bU$ and $\bV$ are two $n\times n$ symmetric real positive-definite matrices, then we have
    \begin{equation}
        \left\Vert\bU^{1/2} - \bV^{1/2}\right\Vert_2
        \leq
        \Vert\bU - \bV\Vert_2
        \cdot
        \frac{1}{\lambda_{\min}(\bU^{1/2} +
        \bV^{1/2})}.
        \label{eq:lemma0}
    \end{equation}
\end{lemma}
\begin{proof}
    Let $\bx\in\mathbb{R}^{n}$ be eigenvector of $(\bU^{1/2} - \bV^{1/2})$ with eigenvalue $\mu$ satisfying $\Vert\bx\Vert_2=1$ and $\vert\mu\vert=\left\Vert\bU^{1/2} - \bV^{1/2}\right\Vert_2$. Then we have
    \begin{align*}
        \left\Vert
        \bU - \bV
        \right\Vert_2
        &\geq
        \left\vert
        \bx^T
        (\bU-\bV)
        \bx
        \right\vert
        \\
        &=
        \left\vert
        \bx^T
        \bU^{1/2} (\bU^{1/2} -\bV^{1/2} )
        \bx
        +
        \bx^T
        (\bU^{1/2} -\bV^{1/2} )\bV^{1/2}
        \bx
        \right\vert
        \\
        &=
        \left\vert
        \bx^T
        \bU^{1/2}
        \mu\bx
        +
        \mu\bx^T
        \bV^{1/2}
        \bx
        \right\vert
        \\
        &=
        \vert\mu\vert
        \cdot
        \left\vert
        \bx^T
        (\bU^{1/2}
        +
        \bV^{1/2})
        \bx
        \right\vert
        \\
        &\geq
        \vert\mu\vert
        \cdot
        \lambda_{\min}
        (\bU^{1/2}
        +
        \bV^{1/2})
        \\
        &=
        \left\Vert\bU^{1/2} - \bV^{1/2}\right\Vert_2
        \cdot
        \lambda_{\min}
        (\bU^{1/2}
        +
        \bV^{1/2}),
    \end{align*}
    as desired.
\end{proof}

\begin{lemma}\label{lemma1}
    Under the conditions in Theorem \ref{theorem:testing_statistic_distribution}, we have
    \begin{align}
        \left\Vert
                \bm{I} - 
                    \Psi_n^{1/2}
                    ({\bm{C}\widehat{\Omega}_{a,mm}\bm{C}^T})^{-1}
                    \Psi_n^{1/2}
            \right\Vert_2
            &=
            O_p
            \left(
            \frac{s+m+K}{\sqrt{n}}
            \right),
        \label{eq:lemma_result1}
        \\
        \left\Vert
                \bm{I} - 
                    \Omega_n^{-1/2}
                    \widehat{\Omega}_0
                    \Omega_n^{-1/2}
            \right\Vert_2
        &=
        O_p
            \left(
            \frac{s+m+K}{\sqrt{n}}
            \right),
        \label{eq:lemma_result2}
        \\
        \left\Vert
                \widehat{\bK}_{n,a}^{-1}
                \widehat{\bV}_{n,a}
                \widehat{\bK}_{n,a}^{-1}
                - 
                {\bK}_{n}^{-1}
                {\bV}_{n}
                {\bK}_{n}^{-1}
            \right\Vert_2
        &=
        O_p
            \left(
            \frac{s+m+K}{\sqrt{n}}
            \right),
        \label{eq:lemma_result3}
        \\
        \lambda_{\max}\left(
                    \widehat{\bK}_{n,a}^{-1}
                    \widehat{\bV}_{n,a}
                    \widehat{\bK}_{n,a}^{-1}
                \right)
        &=
        O_p(1)
        \label{eq:lemma_result3+},
        \\
        \lambda_{\max}\left(\widehat{\bK}_{n,a}
                    \widehat{\bV}_{n,a}^{-1}
                    \widehat{\bK}_{n,a}\right)
        &=
        O_p(1)
        \label{eq:lemma_result3++},
        \\
        \left\Vert
                \bm{I} - 
                    \cT_n^{1/2}
                    \widehat{\cT}_{n,a}^{-1}
                    \cT_n^{1/2}
            \right\Vert_2
        &=
        O_p
            \left(
            \frac{s+m+K}{\sqrt{n}}
            \right),
        \label{eq:lemma_result4}
        \\
        \left\Vert
            \Psi_n
            \left(
            \widehat{\cT}_{n,a}^{-1}
            -
            \cT_n^{-1}
            \right)
        \right\Vert_2
        &=
        O_p
            \left(
            \frac{s+m+K}{\sqrt{n}}
            \right),
        \label{eq:lemma_result5}
        \\
        \left\Vert                        {\widehat{\mathcal{T}}_{n,a}}^{1/2}
                {\widehat{\Psi}_{n,a}}^{-1}
                {\widehat{\mathcal{T}}_{n,a}}^{1/2} - 
                {{\mathcal{T}}_{n}}^{1/2}
                {\Psi}_{n}^{-1}
                {{\mathcal{T}}_{n}}^{1/2}
                \right\Vert_2
        &=O_p
            \left(
            \frac{s+m+K}{\sqrt{n}}
            \right),
        \label{eq:lemma_result6}
        \\
        \lambda_{\max}
        \left(
                {{\mathcal{T}}_{n}}^{1/2}
                {\Psi}_{n}^{-1}
                {{\mathcal{T}}_{n}}^{1/2}
        \right)
        &=O(1),
        \label{eq:lemma_result6+}
        \\
        \lambda_{\max}
        \left(
                {{\mathcal{T}}_{n}}^{-1/2}
                {\Psi}_{n}
                {{\mathcal{T}}_{n}}^{-1/2}
        \right)
        &=O(1)
        \label{eq:lemma_result6++}
        \\
        \left\Vert
        {\Psi}_{n}^{-1/2}\bh_n
        \right\Vert_2
        &=
        O\left(
            \sqrt{r/n}
            \right)
        \label{eq:lemma_result7}
    \end{align}
\end{lemma}

\paragraph{Proof of equation \eqref{eq:lemma_result1}}
\begin{proof}
    By Cauchy-Schwartz inequality
    \begin{align*}
        \left\Vert
                \bm{I} - 
                    \Psi_n^{1/2}
                    ({\bm{C}\widehat{\Omega}_{a,mm}\bm{C}^T})^{-1}
                    \Psi_n^{1/2}
            \right\Vert_2
            &=
        \left\Vert
            \left(
            \Psi_n^{-1/2}
                    {\bm{C}\widehat{\Omega}_{a,mm}\bm{C}^T}
                    \Psi_n^{-1/2}
            -
            \bm{I} 
            \right)
            \Psi_n^{1/2}
                    ({\bm{C}\widehat{\Omega}_{a,mm}\bm{C}^T})^{-1}
                    \Psi_n^{1/2}     
            \right\Vert_2
        \\
        &
        \leq
        \left\Vert
            \Psi_n^{-1/2}
            {\bm{C}}
            \left(
            \widehat{\Omega}_{a,mm}
            -
            \Omega_{mm} 
            \right) 
            {\bm{C}}^T
            \Psi_n^{-1/2}
            \right\Vert_2
        \left\Vert
            \Psi_n^{1/2}
                    ({\bm{C}\widehat{\Omega}_{a,mm}\bm{C}^T})^{-1}
                    \Psi_n^{1/2}
        \right\Vert_2
        .
    \end{align*}
    We will prove that
    \begin{equation}
        \left\Vert
            \Psi_n^{-1/2}
            {\bm{C}}
            \left(
            \widehat{\Omega}_{a,mm}
            -
            \Omega_{mm} 
            \right) 
            {\bm{C}}^T
            \Psi_n^{-1/2}
            \right\Vert_2
        =
        O_p
            \left(
            \frac{s+m+K}{\sqrt{n}}
            \right),
        \label{eq:lemma-I}
    \end{equation}
    and
    \begin{equation}
        \left\Vert
            \Psi_n^{1/2}
                    ({\bm{C}\widehat{\Omega}_{a,mm}\bm{C}^T})^{-1}
                    \Psi_n^{1/2}
        \right\Vert_2
        =
        O_p(1).
        \label{eq:lemma-II}
    \end{equation}
    
    Define $\Tilde{\Omega}_{a}=(\Tilde{{\bm{K}}}_{n,a})^{-1}$, where
    \begin{equation*}
        \Tilde{{\bm{K}}}_{n,a}
        =
        \frac{1}{n}
        \sum_{k=1}^{K}w_k 
        \begin{pmatrix}
            \bm{X}_{\cM}^T
                \\
            \left( {\bm{X}^k_{S\cup \cK}} \right)^T
        \end{pmatrix}
        {\Sigma}(\bm{X}^k{\hat{\cB}_a})
        \begin{pmatrix}
            \bm{X}_{\cM}^T
            \\
            \left( {\bm{X}^k_{S\cup \cK}} \right)^T
        \end{pmatrix}^T,
    \end{equation*}
    and $\Tilde{\Omega}_{a,mm}$ as the submatrix of $\Tilde{{\bm{K}}}_{n,a}$ formed by its first $m$ rows and $m$ columns.
    Assume for now we have
    \begin{equation}
        \Vert
        {\bm{K}}_n - \Tilde{{\bm{K}}}_{n,a}
        \Vert_2
        =
        O_p
            \left(
            \frac{s+m+K}{\sqrt{n}}
            \right).
        \label{eq:lemma-1}
    \end{equation}
    Notice that
    \begin{align*}
        \lim\inf_{n}\lambda_{\min}(\Tilde{{\bm{K}}}_{n,a})
        &
        \geq
        \lim\inf_{n}\lambda_{\min}({{\bm{K}}}_{n})
        -
        \lim\sup_{n}\Vert
        {\bm{K}}_n - \Tilde{{\bm{K}}}_{n,a}
        \Vert_2.    
    \end{align*}
    Under \textcolor{black}{the 6th condition of (A1)}, we have $\lim\inf_n\lambda_{\min}(\bm{K}_n)>0$, this together with condition $s+m+K=o(n^{1/2})$ and \eqref{eq:lemma-1} implies that
    \begin{equation}
        \lim\inf_{n}\lambda_{\min}(\Tilde{{\bm{K}}}_{n,a})>0,
        \label{eq:lemma-2}
    \end{equation}
    with probability tending to $1$. Or equivalently, $\lambda_{\max}(\Tilde{{\bm{K}}}_{n,a}^{-1})=O_p(1)$. Hence, we have
    \begin{align}
        \Vert
        {\bm{K}}_n^{-1} - \Tilde{{\bm{K}}}_{n,a}^{-1}
        \Vert_2
        &=
        \Vert
        {\bm{K}}_n^{-1}({\bm{K}}_n - \Tilde{{\bm{K}}}_{n,a})\Tilde{{\bm{K}}}_{n,a}^{-1}
        \Vert_2
        \nonumber
        \\
        &
        \leq
        \lambda_{\max}({\bm{K}}_n^{-1})
        \Vert
        {\bm{K}}_n - \Tilde{{\bm{K}}}_{n,a}
        \Vert_2
        \lambda_{\max}(\Tilde{{\bm{K}}}_{n,a}^{-1})
        \nonumber
        \\
        &=
        O_p\left(
            \frac{s+m+K}{\sqrt{n}}
            \right).
            \label{eq:lemma-3}
    \end{align}
    This together with $\Vert
        \Omega_{mm} - \Tilde{\Omega}_{a,mm}
        \Vert_2
        \leq \Vert
        {\bm{K}}_n^{-1} - \Tilde{{\bm{K}}}_{n,a}^{-1}
        \Vert_2$ implies that $\Vert
        \Omega_{mm} - \Tilde{\Omega}_{a,mm}
        \Vert_2
        =O_p\left(
            \frac{s+m+K}{\sqrt{n}}
            \right)$.
    Under event $\left\{\widehat{S}_{a}=S\right\}$, we have $\Tilde{\Omega}_{a}=\widehat{\Omega}_{a}$ and $\Tilde{\Omega}_{a,mm}=\widehat{\Omega}_{a,mm}$. According to Theorem \ref{theorem:estimator_statistical_properties}, $\widehat{S}_{a}=S$ with probability tending to $1$. Therefore, we have 
    \begin{equation}
        \Vert
        \Omega_{mm} - \widehat{\Omega}_{a,mm}
        \Vert_2
        =O_p\left(
            \frac{s+m+K}{\sqrt{n}}
            \right).
        \label{eq:lemma-4}
    \end{equation}   
    Then, 
    \begin{align*}
        &~
        \left\Vert
            \Psi_n^{-1/2}
            {\bm{C}}
            \left(
            \widehat{\Omega}_{a,mm}
            -
            \Omega_{mm} 
            \right) 
            {\bm{C}}^T
            \Psi_n^{-1/2}
            \right\Vert_2
        \\
        =&~
        \left\Vert
            \Psi_n^{-1/2}
            {\bm{C}}\Omega_{mm}^{1/2}
            \Omega_{mm}^{-1/2}
            \left(
            \widehat{\Omega}_{a,mm}
            -
            \Omega_{mm} 
            \right) 
            \Omega_{mm}^{-1/2}
            \Omega_{mm}^{1/2}
            {\bm{C}}^T
            \Psi_n^{-1/2}
            \right\Vert_2
        \\
        \leq&~
        \left\Vert
            \Psi_n^{-1/2}
            {\bm{C}}\Omega_{mm}^{1/2}
        \right\Vert_2^2
        \left\Vert
            \Omega_{mm}^{-1/2}
            \left(
            \widehat{\Omega}_{a,mm}
            -
            \Omega_{mm} 
            \right) 
            \Omega_{mm}^{-1/2}
        \right\Vert_2
        \\
        \leq&~
        1\cdot
        \Vert
         \Omega_{mm}^{-1/2}
         \Vert_2^2
         \left\Vert
            \widehat{\Omega}_{a,mm}
            -
            \Omega_{mm} 
        \right\Vert_2
        \\
        \leq&~
        O_p\left(
            \frac{s+m+K}{\sqrt{n}}
            \right),
    \end{align*}
    where $\Vert
         \Omega_{mm}^{-1/2}
         \Vert_2^2=O(1)$
    is because \textcolor{black}{the 4th condition in (A1)} implies that $\lambda_{\max}(\bm{K}_n)=O(1)$ and thus $\lim\inf_n\lambda_{\min}(\Omega_{mm})>0$. Thus, we proved the equation \eqref{eq:lemma-I}. To show equation \eqref{eq:lemma-II}, notice that
    \begin{align*}
        &~
        \lim\inf_{n}\lambda_{\min}
        \left(
            \Psi_n^{-1/2}
                    {\bm{C}\widehat{\Omega}_{a,mm}\bm{C}^T}
                    \Psi_n^{-1/2}
        \right)
        \\
        \geq&~
        \lim\inf_{n}\lambda_{\min}
        \left(
            \Psi_n^{-1/2}
                    {\bm{C}{\Omega}_{mm}\bm{C}^T}
                    \Psi_n^{-1/2}
        \right)
        -
        \lim\sup_{n}
        \left\Vert
        \Psi_n^{-1/2}
                    {\bm{C}({\Omega}_{mm}-\widehat{\Omega}_{a,mm})\bm{C}^T}
                    \Psi_n^{1/2}
        \right\Vert_2
        \\
        =&~
        1-
        \lim\sup_{n}
        \left\Vert
        \Psi_n^{-1/2}
                    {\bm{C}({\Omega}_{mm}-\widehat{\Omega}_{a,mm})\bm{C}^T}
                    \Psi_n^{1/2}
        \right\Vert_2.
    \end{align*}
    This together with \eqref{eq:lemma-I} and the condition $s+m+K=o(\sqrt{n})$ implies that
    \begin{equation}
        \lim\inf_{n}\lambda_{\min}
        \left(
            \Psi_n^{-1/2}
                    {\bm{C}\widehat{\Omega}_{a,mm}\bm{C}^T}
                    \Psi_n^{-1/2}
        \right)
        >0
    \end{equation}
    with probability tending to $1$. Or equivalently, $\lambda_{\max}\left(
            \Psi_n^{1/2}
                    ({\bm{C}\widehat{\Omega}_{a,mm}\bm{C}^T})^{-1}
                    \Psi_n^{1/2}
        \right)=O_p(1)$, as desired.

    It remains to show \eqref{eq:lemma-1}. 
    By definition and symmetry of ${\bm{K}}_n$ and $\Tilde{{\bm{K}}}_{n,a}$, it suffices to show that $\Vert
        {\bm{K}}_n - \Tilde{{\bm{K}}}_{n,a}
        \Vert_\infty=O_p((s+m+K)/\sqrt{n})$.
    \begin{align*}
        \Vert
        {\bm{K}}_n - \Tilde{{\bm{K}}}_{n,a}
        \Vert_\infty
        &=
        \max_{j\in {S\cup \cM\cup\cK}}
        \left\Vert
        \sum_{k=1}^{K}w_k \frac{1}{n}
                        \begin{pmatrix}
                            \bm{X}_{\cM}^T
                            \\
                            (\bm{X}^k_{S\cup \cK})^T
                        \end{pmatrix}
                         \left({\Sigma}(\bm{X}^k\cB^*) - {\Sigma}(\bm{X}^k\hat{\cB}_a)\right)
                        \begin{pmatrix}
                            \bm{X}_{\cM}^T
                            \\
                            (\bm{X}^k_{S\cup \cK})^T
                        \end{pmatrix}^T
        \right\Vert_1
                \\
        &\leq
        \max_{j\in {S\cup \cM\cup\cK}}
        \left\Vert
        \sum_{k=1}^{K}w_k \frac{1}{n}
                        \begin{pmatrix}
                            \bm{X}_{\cM}^T
                            \\
                            (\bm{X}^k_{S\cup \cK})^T
                        \end{pmatrix}
                         \left({\Sigma}(\bm{X}^k\cB^*) - {\Sigma}(\bm{X}^k\hat{\cB}_a)\right)
                        \begin{pmatrix}
                            \bm{X}_{\cM}^T
                            \\
                            (\bm{X}^k_{S\cup \cK})^T
                        \end{pmatrix}^T
        \right\Vert_2\sqrt{s+m+K}
    \end{align*}
    By Taylor's expansion, for $j\in S\cup \cM$,
    \begin{align*}
        &~
        \sum_{k=1}^{K}w_k
                        \bx_{(j)}^T
                         \left({\Sigma}(\bm{X}^k\cB^*) - {\Sigma}(\bm{X}^k\hat{\cB}_a)\right)
                        \begin{pmatrix}
                            \bm{X}_{\cM}^T
                            \\
                            (\bm{X}^k_{S\cup \cK})^T
                        \end{pmatrix}^T
        \\
        =&~
        \sum_{k=1}^{K}w_k 
                    \begin{pmatrix}
                            \bm{X}_{\cM}^T
                            \\
                            (\bm{X}^k_{S\cup \cK})^T
                        \end{pmatrix}
                    \text{diag}\left\{
                        \bx_{(j)}\circ{\Sigma}^\prime(\bX^k\bar{\cB}_j)
                    \right\}
                    \begin{pmatrix}
                            \bm{X}_{\cM}^T
                            \\
                            (\bm{X}^k_{S\cup \cK})^T
                        \end{pmatrix}^T
        (\hat{\cB}_0 - {\cB^*} ),
    \end{align*}
    for some $\bar{\cB}_j$ on the line segment joining $\cB^*$ and $\hat{\cB}_a$. 
    Since $\sup_{t}\vert{\Sigma}^\prime(t)\vert=O(1)$ and by the \textcolor{black}{5th condition of (A1)}, we have
    \begin{align*}
        &~
        \left\Vert
        \sum_{k=1}^{K}w_k
                        \bx_{(j)}^T
                         \left({\Sigma}(\bm{X}^k\cB^*) - {\Sigma}(\bm{X}^k\hat{\cB}_a)\right)
                        \begin{pmatrix}
                            \bm{X}_{\cM}^T
                            \\
                            (\bm{X}^k_{S\cup \cK})^T
                        \end{pmatrix}^T
        \right\Vert_2
        \\
        \leq&~
        \left\Vert
        \sum_{k=1}^{K}w_k 
                    \begin{pmatrix}
                            \bm{X}_{\cM}^T
                            \\
                            (\bm{X}^k_{S\cup \cK})^T
                        \end{pmatrix}
                    \text{diag}\left\{
                        \vert\bx_{(j)}\vert
                        \circ
                        \vert{\Sigma}^\prime(\bX^k\bar{\cB}_j)\vert
                    \right\}
                    \begin{pmatrix}
                            \bm{X}_{\cM}^T
                            \\
                            (\bm{X}^k_{S\cup \cK})^T
                        \end{pmatrix}^T
        \right\Vert_2
        \left\Vert
        \hat{\cB}_a - {\cB^*}
        \right\Vert_2
        \\
        \leq&~
        \sup_{t}\vert{\Sigma}^\prime(t)\vert
        \cdot
        \lambda_{\max}
        \left(
        \sum_{k=1}^{K}w_k 
                    \begin{pmatrix}
                            \bm{X}_{\cM}^T
                            \\
                            (\bm{X}^k_{S\cup \cK})^T
                        \end{pmatrix}
                    \text{diag}\left\{
                        \vert\bx_{(j)}\vert
                    \right\}
                    \begin{pmatrix}
                            \bm{X}_{\cM}^T
                            \\
                            (\bm{X}^k_{S\cup \cK})^T
                        \end{pmatrix}^T
        \right)
        \left\Vert
        \hat{\cB}_a - {\cB^*}
        \right\Vert_2
        \\
        =&~
        O(n)\left\Vert
        \hat{\cB}_a - {\cB^*}
        \right\Vert_2.
    \end{align*}
    This together with $\left\Vert
        \hat{\cB}_a - {\cB^*}
        \right\Vert_2=O_p(\sqrt{(s+m+K)/n})$ from Theorem \eqref{theorem:estimator_statistical_properties} implies that for $j\in S\cup\cM$,
    \begin{equation*}
        \left\Vert
        \sum_{k=1}^{K}w_k \frac{1}{n}
                        \bx_{(j)}^T
                         \left({\Sigma}(\bm{X}^k\cB^*) - {\Sigma}(\bm{X}^k\hat{\cB}_a)\right)
                        \begin{pmatrix}
                            \bm{X}_{\cM}^T
                            \\
                            (\bm{X}^k_{S\cup \cK})^T
                        \end{pmatrix}^T
        \right\Vert_2
        =
        O_p\left(
            \sqrt{
            \frac{s+m+K}{{n}}
            }
            \right).
    \end{equation*}
    Similarly, we can prove along with the \textcolor{black}{4th condition of (A1)} that for $j\in \cK$,
    \begin{equation*}
        \left\Vert
        \sum_{k=1}^{K}w_k \frac{1}{n}
                        (\bx_{(j)}^{k})^T
                         \left({\Sigma}(\bm{X}^k\cB^*) - {\Sigma}(\bm{X}^k\hat{\cB}_a)\right)
                        \begin{pmatrix}
                            \bm{X}_{\cM}^T
                            \\
                            (\bm{X}^k_{S\cup \cK})^T
                        \end{pmatrix}^T
        \right\Vert_2
        =
        O_p\left(
            \sqrt{
            \frac{s+m+K}{{n}}
            }
            \right).
    \end{equation*}
    Combining those together, we have
    \begin{align*}
        \Vert
        {\bm{K}}_n - \Tilde{{\bm{K}}}_{n,a}
        \Vert_\infty
        &\leq
        \max_{j\in {S\cup \cM\cup\cK}}
        \left\Vert
        \sum_{k=1}^{K}w_k \frac{1}{n}
                        (\bx_{(j)}^{k})^T
                         \left({\Sigma}(\bm{X}^k\cB^*) - {\Sigma}(\bm{X}^k\hat{\cB}_a)\right)
                        \begin{pmatrix}
                            \bm{X}_{\cM}^T
                            \\
                            (\bm{X}^k_{S\cup \cK})^T
                        \end{pmatrix}^T
        \right\Vert_2
        \sqrt{m+s+K}
        \\
        &=
        O_p\left(
            \sqrt{
            \frac{s+m+K}{{n}}
            }
            \right)\sqrt{s+m+K} = 
        O_p\left(
        \frac{s+m+K}{\sqrt{n}}
        \right).
    \end{align*}
    This ends the prove of equation \eqref{eq:lemma_result1}.
\end{proof}

\paragraph{Proof of equation \eqref{eq:lemma_result2}}
\begin{proof}
    By Cauchy-Schwartz inequality
    \begin{align*}
        \left\Vert
                \bm{I} - 
                    \Omega_n^{-1/2}
                    \widehat{\Omega}_0
                    \Omega_n^{-1/2}
            \right\Vert_2
            &\leq
        \left\Vert
            \Omega_n^{-1/2}
            (\Omega_n - \widehat{\Omega}_0)
            \Omega_n^{-1/2}    
            \right\Vert_2
        \\
        &
        \leq
        \left\Vert
            \Omega_n^{-1}    
            \right\Vert_2
        \left\Vert
            \Omega_n - \widehat{\Omega}_0
        \right\Vert_2
        .
    \end{align*}
    By \textcolor{black}{the 4th condition in (A1)}, we have $\lambda_{\max}(\bm{K}_n)=O(1)$. It remains to bound $\left\Vert
            \bm{K}_n^{-1} - \widehat{\bm{K}}_{n,0}^{-1}
        \right\Vert_2$.
    Assume for now, we have
    \begin{equation}
        \left\Vert
            \bm{K}_n- \widehat{\bm{K}}_{n,0}
        \right\Vert_2
        =
        O_p\left(\frac{s+m+K}{\sqrt{n}}\right),
        \label{eq:lemma-5}
    \end{equation}
    this together with $\lambda_{\max}(\bm{K}_n^{-1})=O(1)$ and $s+m+K=o(\sqrt{n})$ implies that with probability tending to $1$,
    \begin{align*}
        \lim\inf_{n}\lambda_{\min}(\widehat{\bm{K}}_{n,0})
        &
        \geq
        \lim\inf_{n}\lambda_{\min}(\bm{K}_{n})
        -
        \lim\sup_{n}\Vert
        \bm{K}_n - \widehat{\bm{K}}_{n,0}
        \Vert_2
        >0,
    \end{align*}
    which is equivalent to $\lambda_{\max}(\widehat{\bm{K}}_{n,0}^{-1})=O_p(1)$. Therefore,
    \begin{align*}
        \left\Vert
            \bm{K}_n^{-1}- \widehat{\bm{K}}_{n,0}^{-1}
        \right\Vert_2
        &=
        \left\Vert
            \bm{K}_n^{-1}(\bm{K}_n- \widehat{\bm{K}}_{n,0})
            \widehat{\bm{K}}_{n,0}^{-1}
        \right\Vert_2
        \\
        &\leq
        \left\Vert
            \bm{K}_n^{-1}
        \right\Vert_2
        \left\Vert
            \bm{K}_n- \widehat{\bm{K}}_{n,0}
        \right\Vert_2
        \left\Vert
            \widehat{\bm{K}}_{n,0}^{-1}
        \right\Vert_2
        =
        O_p\left(\frac{s+m+K}{\sqrt{n}}\right).
    \end{align*}
    It remains to prove equation \eqref{eq:lemma-5}. Similar to the prove of equation \eqref{eq:lemma-1}, we can show that under event $\{S=\hat{S}_0\}$, $\Vert\bm{K}_n-\Tilde{\bm{K}}_{n,0}\Vert_2=O_p\left(\frac{s+m+K}{\sqrt{n}}\right)$, where
    \begin{equation*}
        \Tilde{{\bm{K}}}_{n,0}
        =
        \frac{1}{n}
        \sum_{k=1}^{K}w_k 
        \begin{pmatrix}
            \bm{X}_{\cM}^T
                \\
            \left( {\bm{X}^k_{S\cup \cK}} \right)^T
        \end{pmatrix}
        {\Sigma}(\bm{X}^k{\hat{\cB}_0})
        \begin{pmatrix}
            \bm{X}_{\cM}^T
            \\
            \left( {\bm{X}^k_{S\cup \cK}} \right)^T
        \end{pmatrix}^T.
    \end{equation*}
    This together with the fact that $\mathbb{P}(S=\hat{S}_0)\rightarrow 1$ proves the equation \eqref{eq:lemma-5}. This ends the proof of equation \eqref{eq:lemma_result2}.
    
\end{proof}

\paragraph{Proof of equation \eqref{eq:lemma_result3}, \eqref{eq:lemma_result3+},
\eqref{eq:lemma_result3++}}
\begin{proof}
    We first prove that $\lambda_{\max}(\bV_n)=O(1)$ by showing that $\lambda_{\max}(\bV_n)\leq\lambda_{\max}(\bK_n)$:
    \begin{align*}
        \lambda_{\max}(\bV_n)
        &=
        \sup_{\ba:\Vert\ba\Vert_2=1}
        \ba^T\bV_{n}\ba
        \\
        &=
        \sup_{\ba:\Vert\ba\Vert_2=1}
        \frac{1}{n}
        \mathbb{E}
        \left[
        \left\Vert
            \sum_{k=1}^{K}w_k
                \ba^T
                        \begin{pmatrix}
                                \bm{X}_{\cM}^T
                                \\
                                \left( {\bm{X}^k_{{{S}}\cup \cK}} \right)^T
                            \end{pmatrix}
                            \bm{H}(\bm{X}^k{{\cB}_0})
                            \left\{
                                \bm{Y}^k - \bm{\mu}(\bm{X}^k{{\cB^*}})
                            \right\}
        \right\Vert_2^2
        \right]
        \\
        &\leq
        \sup_{\ba:\Vert\ba\Vert_2=1}
        \frac{1}{n}
        \sum_{k=1}^{K}w_k
        \mathbb{E}
        \left[
        \left\Vert
                \ba^T
                        \begin{pmatrix}
                                \bm{X}_{\cM}^T
                                \\
                                \left( {\bm{X}^k_{{{S}}\cup \cK}} \right)^T
                            \end{pmatrix}
                            \bm{H}(\bm{X}^k{{\cB}_0})
                            \left\{
                                \bm{Y}^k - \bm{\mu}(\bm{X}^k{{\cB^*}})
                            \right\}
        \right\Vert_2^2
        \right]
        \\
        &=
        \sup_{\ba:\Vert\ba\Vert_2=1}
        \frac{1}{n}
        \sum_{k=1}^{K}w_k
        \ba^T
            \begin{pmatrix}
                                \bm{X}_{\cM}^T
                                \\
                                (\bm{X}^k_{S\cup \cK})^T
                            \end{pmatrix}
                            {\Sigma}(\bm{X}^k{\cB^*})
                            \begin{pmatrix}
                                \bm{X}_{\cM}^T
                                \\
                                (\bm{X}^k_{S\cup \cK})^T
                            \end{pmatrix}^T
        \ba
        \\
        &=
        \sup_{\ba:\Vert\ba\Vert_2=1}
        \ba^T
        \bK_n
        \ba
        =
        \lambda_{\max}(\bK_n),
    \end{align*}
    where the inequality is from the convexity of function $\Vert\cdot\Vert_2^2$. Since $\lambda_{\max}(\bK_n)=O(1)$ by \textcolor{black}{the 4th condition in (A1)}, we have $\lambda_{\max}(\bV_n)=O(1)$. And by \textcolor{black}{the 7th condition in (A1)}, we have $\lambda_{\max}(\bV_n^{-1})=O(1)$. And similar to the proof equation \eqref{eq:lemma-1} and equation \eqref{eq:lemma-5}, we can show that
    \begin{equation}
            \left\Vert
                \bV_n- \widehat{\bV}_{n,a}
            \right\Vert_2
            =
            O_p\left(\frac{s+m+K}{\sqrt{n}}\right),
            \label{eq:lemma-6}
        \end{equation}
    
    We next prove equation \eqref{eq:lemma_result3} by splitting $\widehat{\bK}_{n,a}^{-1}
                    \widehat{\bV}_{n,a}
                    \widehat{\bK}_{n,a}^{-1}
                    - 
                    {\bK}_{n}^{-1}
                    {\bV}_{n}
                    {\bK}_{n}^{-1}$
    into several parts:
    \begin{align*}
        &~
        \widehat{\bK}_{n,a}^{-1}
                    \widehat{\bV}_{n,a}
                    \widehat{\bK}_{n,a}^{-1}
                    - 
                    {\bK}_{n}^{-1}
                    {\bV}_{n}
                    {\bK}_{n}^{-1}
        \nonumber
        \\
        =&~
        \underbrace
        {
        \widehat{\bK}_{n,a}^{-1}
        (\widehat{\bV}_{n,a}-{\bV}_{n})
        \widehat{\bK}_{n,a}^{-1}
        }_{I}
        +
        \underbrace
        {
        (\widehat{\bK}_{n,a}^{-1}-{\bK}_{n}^{-1})
        {\bV}_{n}
        (\widehat{\bK}_{n,a}^{-1}-{\bK}_{n}^{-1})
        }_{II}
        \\
        &~
        +
        \underbrace
        {
        (\widehat{\bK}_{n,a}^{-1}-{\bK}_{n}^{-1})
        {\bV}_{n}
        {\bK}_{n}^{-1}
        +
        {\bK}_{n}^{-1}
        {\bV}_{n}
        (\widehat{\bK}_{n,a}^{-1}-{\bK}_{n}^{-1})
        }_{III}
    \end{align*}
    For the term $\widehat{\bK}_{n,a}^{-1}$ in $I$, we proved in equation \eqref{eq:lemma-2} that $\lambda_{\max}(\tilde{\bK}_{n,a}^{-1})=O_p(1)$. This together with the fact that $\mathbb{P}(S=\hat{S}_{a})\rightarrow 1$ proves that $\lambda_{\max}(\widehat{\bK}_{n,a}^{-1})=O_p(1)$. Then along with equation \eqref{eq:lemma-6}, we have
    \begin{equation*}
        \Vert
        {
        \widehat{\bK}_{n,a}^{-1}
        (\widehat{\bV}_{n,a}-{\bV}_{n})
        \widehat{\bK}_{n,a}^{-1}
        }
        \Vert_2
        \leq
        \Vert
        \widehat{\bK}_{n,a}^{-1}
        \Vert_2^2
        \cdot
        \Vert
        {
        \widehat{\bV}_{n,a}-{\bV}_{n}
        }
        \Vert_2
        =
        O_p\left(\frac{s+m+K}{\sqrt{n}}\right).
    \end{equation*}
    For term $(\widehat{\bK}_{n,a}^{-1}-{\bK}_{n}^{-1})$ in $II$, we proved in equation \eqref{eq:lemma-3} that $\Vert {\bm{K}}_n^{-1} - \Tilde{{\bm{K}}}_{n,a}^{-1} \Vert_2=O_p\left(\frac{s+m+K}{\sqrt{n}}\right)$. This together with the fact that $\mathbb{P}(S=\hat{S}_{a})\rightarrow 1$ proves that $\Vert {\bm{K}}_n^{-1} - \widehat{{\bm{K}}}_{n,a}^{-1} \Vert_2=O_p\left(\frac{s+m+K}{\sqrt{n}}\right)$. Then along with $\lambda_{\max}(\bV_n)=O(1)$ and $s+m+K=o(\sqrt{n})$, we have
    \begin{equation*}
        \Vert
        {
        {
        (\widehat{\bK}_{n,a}^{-1}-{\bK}_{n}^{-1})
        {\bV}_{n}
        (\widehat{\bK}_{n,a}^{-1}-{\bK}_{n}^{-1})
        }
        }
        \Vert_2
        \leq
        \Vert
        {\bV}_{n}
        \Vert_2
        \cdot
        \Vert
        {
        \widehat{\bK}_{n,a}^{-1}-{\bK}_{n}^{-1}
        }
        \Vert_2^2
        =
        o_p\left(\frac{s+m+K}{\sqrt{n}}\right).
    \end{equation*}
    For bounding $III$, we proved that $\lambda_{\max}(\bK_n^{-1})=O(1)$. Then we have
    \begin{equation*}
        \Vert
        {
        (\widehat{\bK}_{n,a}^{-1}-{\bK}_{n}^{-1})
        {\bV}_{n}
        {\bK}_{n}^{-1}
        +
        {\bK}_{n}^{-1}
        {\bV}_{n}
        (\widehat{\bK}_{n,a}^{-1}-{\bK}_{n}^{-1})
        }
        \Vert_2
        \leq
        2
        \Vert
        {\bK}_{n}^{-1}
        \Vert_2
        \Vert
        {\bV}_{n}
        \Vert_2
        \cdot
        \Vert
        {
        \widehat{\bK}_{n,a}^{-1}-{\bK}_{n}^{-1}
        }
        \Vert_2
        =
        O_p\left(\frac{s+m+K}{\sqrt{n}}\right).
    \end{equation*}
    Combining the above three results, we proved the desired equation \eqref{eq:lemma_result3}. This, together with the fact that $s+m+K=o(\sqrt{n})$ as well as $\left\Vert
                    {\bK}_{n}^{-1}
                    {\bV}_{n}
                    {\bK}_{n}^{-1}
                \right\Vert_2=O(1)$ proves the result \eqref{eq:lemma_result3+}.
    We can further prove result \eqref{eq:lemma_result3++}
    by proving that $\lim\inf_{n}\lambda_{\min}\left(\widehat{\bK}_{n,a}^{-1}
                    \widehat{\bV}_{n,a}
                    \widehat{\bK}_{n,a}^{-1}\right)>0$ with probability tending to $1$. Notice that
    \begin{align*}
        &~
        \lim\inf_{n}\lambda_{\min}\left(\widehat{\bK}_{n,a}^{-1}
                    \widehat{\bV}_{n,a}
                    \widehat{\bK}_{n,a}^{-1}\right)
        \\
        \geq &~
        \lim\inf_{n}\lambda_{\min}\left({\bK}_{n}^{-1}
                    {\bV}_{n}
                    {\bK}_{n}^{-1}\right)
        -
        \lim\sup_{n}
        \left\Vert
                    \widehat{\bK}_{n,a}^{-1}
                    \widehat{\bV}_{n,a}
                    \widehat{\bK}_{n,a}^{-1}
                    - 
                    {\bK}_{n}^{-1}
                    {\bV}_{n}
                    {\bK}_{n}^{-1}
                \right\Vert_2.
    \end{align*}
    This together with the fact that $\lambda_{\max}({\bK}_{n})=O(1)$, $\lambda_{\max}({\bV}_{n}^{-1})=O(1)$, $s+m+K=o(\sqrt{n})$, as well as the equation \eqref{eq:lemma_result3} proved above, we have $\lim\inf_{n}\lambda_{\min}\left(\widehat{\bK}_{n,a}^{-1}
                    \widehat{\bV}_{n,a}
                    \widehat{\bK}_{n,a}^{-1}\right)>0$ with probability tending to $1$.
    
\end{proof}

\paragraph{Proof of equation \eqref{eq:lemma_result4}}
\begin{proof}
    We first prove that 
    \begin{equation}
        \left\Vert
        \bm{I} - 
                    \cT_n^{-1/2}
                    \widehat{\cT}_{n,a}
                    \cT_n^{-1/2}
        \right\Vert_2 = O_p\left(\frac{s+m+K}{\sqrt{n}}\right).
        \label{eq:lemma-9}
    \end{equation}
    Notice that
    \begin{align}
        &~
        \bm{I} - 
                    \cT_n^{-1/2}
                    \widehat{\cT}_{n,a}
                    \cT_n^{-1/2}
        \nonumber
        \\
        =&~
        \cT_n^{-1/2}
        ( \cT_n - \widehat{\cT}_{n,a})
        \cT_n^{-1/2}
        \nonumber
        \\
        =&~
        \cT_n^{-1/2}
        \begin{pmatrix}
                    \bC^T 
                            \\
                            \bm{0}_{r\times (s+K)}^T
                \end{pmatrix}^T
                (
                {\bK}_{n}^{-1}
                    {\bV}_{n}
                    {\bK}_{n}^{-1}
                -
                \widehat{\bK}_{n,a}^{-1}
                    \widehat{\bV}_{n,a}
                    \widehat{\bK}_{n,a}^{-1}
                )
            \begin{pmatrix}
                    \bC^T 
                            \\
                            \bm{0}_{r\times (s+K)}^T
                \end{pmatrix}
        \cT_n^{-1/2}
        \nonumber
        \\
        =&~
        \cT_n^{-1/2}
        \begin{pmatrix}
                    \bC^T 
                            \\
                            \bm{0}_{r\times (s+K)}^T
                \end{pmatrix}^T
                {\bK}_{n}^{-1}
                {\bV}_{n}^{1/2}
                \cdot
                    {\bV}_{n}^{-1/2}
                    {\bK}_{n}
                (
                {\bK}_{n}^{-1}
                    {\bV}_{n}
                    {\bK}_{n}^{-1}
                -
                \widehat{\bK}_{n,a}^{-1}
                    \widehat{\bV}_{n,a}
                    \widehat{\bK}_{n,a}^{-1}
                )
                {\bK}_{n}
                {\bV}_{n}^{-1/2}
                \cdot
                {\bV}_{n}^{1/2}
                {\bK}_{n}^{-1}
            \begin{pmatrix}
                    \bC^T 
                            \\
                            \bm{0}_{r\times (s+K)}^T
                \end{pmatrix}
        \cT_n^{-1/2}
        \nonumber.
    \end{align}
    This implies that
    \begin{align*}
        &~
        \left\Vert
        \bm{I} - 
                    \cT_n^{-1/2}
                    \widehat{\cT}_{n,a}
                    \cT_n^{-1/2}
        \right\Vert_2
        \nonumber
        \\
        \leq &~
        \left\Vert
        {\bK}_{n}^{-1}
                    {\bV}_{n}
                    {\bK}_{n}^{-1}
                -
                \widehat{\bK}_{n,a}^{-1}
                    \widehat{\bV}_{n,a}
                    \widehat{\bK}_{n,a}^{-1}
        \right\Vert_2
        \cdot
        \left\Vert
        {\bK}_{n}
                {\bV}_{n}^{-1/2}
        \right\Vert_2^2
        \cdot
        \left\Vert
        {\bV}_{n}^{1/2}
                {\bK}_{n}^{-1}
            \begin{pmatrix}
                    \bC^T 
                            \\
                            \bm{0}_{r\times (s+K)}^T
                \end{pmatrix}
        \cT_n^{-1/2}
        \right\Vert_2^2
        \\
        =&~
        \left\Vert
        {\bK}_{n}^{-1}
                    {\bV}_{n}
                    {\bK}_{n}^{-1}
                -
                \widehat{\bK}_{n,a}^{-1}
                    \widehat{\bV}_{n,a}
                    \widehat{\bK}_{n,a}^{-1}
        \right\Vert_2
        \cdot
        \left\Vert
                {\bK}_{n}
                {\bV}_{n}^{-1}
                {\bK}_{n}
        \right\Vert_2
        \cdot
        \left\Vert
        \cT_n^{-1/2}
        \cT_n
        \cT_n^{-1/2}
        \right\Vert_2
        \\
        =&~
        O_p\left(\frac{s+m+K}{\sqrt{n}}\right),
    \end{align*}
    where the last equality is from the definition of $\cT_n$, the proved equation \eqref{eq:lemma_result3} and the fact that $\lambda_{\max}(\bK_n)=O(1)$, $\lambda_{\max}(\bV_n^{-1})=O(1)$. This ends the proof of equation \eqref{eq:lemma-9}. We then prove that
    \begin{equation}
        \lambda_{\max}
        \left(
                    \cT_n^{1/2}
                    \widehat{\cT}_{n,a}^{-1}
                    \cT_n^{1/2}
        \right)
        =O_p(1)
        \label{eq:lemma-8}
    \end{equation}
    by proving that $\lim\inf_{n}\lambda_{\min}\left(\cT_n^{-1/2}
                    \widehat{\cT}_{n,a}
                    \cT_n^{-1/2}\right)>0$ with probability tending to $1$. Notice that
    \begin{align*}
        &~
        \lim\inf_{n}\lambda_{\min}\left(
                    \cT_n^{-1/2}
                    \widehat{\cT}_{n,a}
                    \cT_n^{-1/2}
        \right)
        \\
        \geq &~
        \lim\inf_{n}\lambda_{\min}\left(
                    \cT_n^{-1/2}
                    {\cT}_{n}
                    \cT_n^{-1/2}\right)
        -
        \lim\sup_{n}
        \left\Vert
                    \cT_n^{-1/2}
                    (
                    {\cT}_{n}
                    -
                    \widehat{\cT}_{n,a}
                    )
                    \cT_n^{-1/2}
                \right\Vert_2
        \\
        = &~
        1 - 
        \lim\sup_{n}
        \left\Vert
                    \bm{I} - 
                    \cT_n^{-1/2}
                    \widehat{\cT}_{n,a}
                    \cT_n^{-1/2}
                \right\Vert_2
    \end{align*}
    This together with the fact that $s+m+K=o(\sqrt{n})$, as well as the equation \eqref{eq:lemma-9} proved above, implies that $\lim\inf_{n}\lambda_{\min}\left(\cT_n^{1/2}
                    \widehat{\cT}_{n,a}^{-1}
                    \cT_n^{1/2}\right)>0$ with probability tending to $1$. Finally, combining the results \eqref{eq:lemma-9} and \eqref{eq:lemma-8}, we have
    \begin{align*}
        \left\Vert
                \bm{I} - 
                    \cT_n^{1/2}
                    \widehat{\cT}_{n,a}^{-1}
                    \cT_n^{1/2}
            \right\Vert_2
        &=
        \left\Vert
        \left(
                \cT_n^{-1/2}
                    \widehat{\cT}_{n,a}
                    \cT_n^{-1/2}
                    -
                \bm{I}
        \right)
                    \cT_n^{1/2}
                    \widehat{\cT}_{n,a}^{-1}
                    \cT_n^{1/2}
            \right\Vert_2
        \\
        &\leq
        \left\Vert
                \cT_n^{-1/2}
                    \widehat{\cT}_{n,a}
                    \cT_n^{-1/2}
                    -
                \bm{I}
            \right\Vert_2
        \left\Vert
                    \cT_n^{1/2}
                    \widehat{\cT}_{n,a}^{-1}
                    \cT_n^{1/2}
            \right\Vert_2
        \\
        &
        =O_p\left(\frac{s+m+K}{\sqrt{n}}\right).
    \end{align*}
    This ends the proof of equation \eqref{eq:lemma_result4}.
\end{proof}

\paragraph{Proof of equation \eqref{eq:lemma_result5}}
\begin{proof}
    \begin{align*}
        &~
        \left\Vert
            \Psi_n
            \left(
            \widehat{\cT}_{n,a}^{-1}
            -
            \cT_n^{-1}
            \right)
        \right\Vert_2
        \\
        =&~
        \left\Vert
            \Psi_n
            \cT_n^{-1/2}
            \left(
            \cT_n^{1/2}
            \widehat{\cT}_{n,a}^{-1}
            \cT_n^{1/2}
            -
            \bm{I}
            \right)
            \cT_n^{-1/2}
        \right\Vert_2
        \\
        \leq &~
        \left\Vert
            \Psi_n
            \cT_n^{-1/2}
        \right\Vert_2
        \left\Vert
            \cT_n^{1/2}
            \widehat{\cT}_{n,a}^{-1}
            \cT_n^{1/2}
            -
            \bm{I}
        \right\Vert_2
        \left\Vert
            \cT_n^{-1/2}
        \right\Vert_2.
    \end{align*}
    Notice that
    \begin{align*}
        &~
        \left\Vert
            \Psi_n
            \cT_n^{-1/2}
        \right\Vert_2
        \\
        =&~
        \left\Vert
        \begin{pmatrix}
                    \bC^T 
                            \\
                            \bm{0}_{r\times (s+K)}^T
                \end{pmatrix}^T
            \bm{K}_n^{-1}
            \begin{pmatrix}
                    \bC^T 
                            \\
                            \bm{0}_{r\times (s+K)}^T
                \end{pmatrix}
        \cT_n^{-1/2}
        \right\Vert_2
        \\
        =&~
        \left\Vert
        \begin{pmatrix}
                    \bC^T 
                            \\
                            \bm{0}_{r\times (s+K)}^T
                \end{pmatrix}^T
            \bm{K}_n^{-1}
            \bV_{n}^{1/2}
            \cdot
            \bV_{n}^{-1/2}
            \bm{K}_n
            \bV_{n}^{-1/2}
            \cdot
            \bV_{n}^{1/2}
            \bm{K}_n^{-1}
            \begin{pmatrix}
                    \bC^T 
                            \\
                            \bm{0}_{r\times (s+K)}^T
                \end{pmatrix}
        \cT_n^{-1/2}
        \right\Vert_2
        \\
        \leq&~
        \left\Vert
        \begin{pmatrix}
                    \bC^T 
                            \\
                            \bm{0}_{r\times (s+K)}^T
                \end{pmatrix}^T
            \bm{K}_n^{-1}
            \bV_{n}^{1/2}
            \right\Vert_2
            \cdot
            \left\Vert
            \bV_{n}^{-1/2}
            \bm{K}_n
            \bV_{n}^{-1/2}
            \right\Vert_2
            \cdot
            \left\Vert
            \bV_{n}^{1/2}
            \bm{K}_n^{-1}
            \begin{pmatrix}
                    \bC^T 
                            \\
                            \bm{0}_{r\times (s+K)}^T
                \end{pmatrix}
        \cT_n^{-1/2}
        \right\Vert_2
        \\
        =&~
        \left\Vert
        \cT_n^{1/2}
        \right\Vert_2
        \cdot
            \left\Vert
            \bV_{n}^{-1/2}
            \bm{K}_n
            \bV_{n}^{-1/2}
            \right\Vert_2
            \cdot
            \left\Vert
             \cT_n^{-1/2}
              \cT_n
        \cT_n^{-1/2}
        \right\Vert_2.
    \end{align*}
    Then we have
    \begin{align*}
        &~
        \left\Vert
            \Psi_n
            \left(
            \widehat{\cT}_{n,a}^{-1}
            -
            \cT_n^{-1}
            \right)
        \right\Vert_2
        \\
        \leq &~
        \left\Vert
        \cT_n^{1/2}
        \right\Vert_2
        \cdot
            \left\Vert
            \bV_{n}^{-1/2}
            \bm{K}_n
            \bV_{n}^{-1/2}
            \right\Vert_2
            \cdot
        \left\Vert
            \cT_n^{1/2}
            \widehat{\cT}_{n,a}^{-1}
            \cT_n^{1/2}
            -
            \bm{I}
        \right\Vert_2
        \cdot
        \left\Vert
            \cT_n^{-1/2}
        \right\Vert_2
        \\
        =&~
        \left\Vert
            \bV_{n}^{-1/2}
            \bm{K}_n
            \bV_{n}^{-1/2}
            \right\Vert_2
            \cdot
        \left\Vert
            \cT_n^{1/2}
            \widehat{\cT}_{n,a}^{-1}
            \cT_n^{1/2}
            -
            \bm{I}
        \right\Vert_2
        \\
        =&~
        O_p\left(\frac{s+m+K}{\sqrt{n}}\right),
    \end{align*}
    where the last equation is from the fact that $\lambda_{\max}(\bK_n)=O(1)$, $\lambda_{\max}(\bV_n^{-1})=O(1)$ as well as equation \eqref{eq:lemma_result4}. This ends the proof of equation \eqref{eq:lemma_result5}.
\end{proof}

\paragraph{Proof of equation \eqref{eq:lemma_result6}, \eqref{eq:lemma_result6+},
\eqref{eq:lemma_result6++}}
\begin{proof}
    Notice that
    \begin{align}
        &~
        {\widehat{\mathcal{T}}_{n,a}}^{1/2}
                {\widehat{\Psi}_{n,a}}^{-1}
                {\widehat{\mathcal{T}}_{n,a}}^{1/2} - 
                {{\mathcal{T}}_{n}}^{1/2}
                {\Psi}_{n}^{-1}
                {{\mathcal{T}}_{n}}^{1/2}
        \nonumber
        \\
        =&~
        {\widehat{\mathcal{T}}_{n,a}}^{1/2}
        \left(
                {\widehat{\Psi}_{n,a}}^{-1}
                -
                {\Psi}_{n}^{-1}
        \right)
                {\widehat{\mathcal{T}}_{n,a}}^{1/2}
        +
        \left(
        {\widehat{\mathcal{T}}_{n,a}}^{1/2}
                {\Psi}_{n}^{-1}
                {\widehat{\mathcal{T}}_{n,a}}^{1/2}
                -
                {{\mathcal{T}}_{n}}^{1/2}
                {\Psi}_{n}^{-1}
                {{\mathcal{T}}_{n}}^{1/2}
        \right).
        \label{eq:lemma-10}
    \end{align}
    For the first term, we have
    \begin{align*}
        &~
        \left\Vert
        {\widehat{\mathcal{T}}_{n,a}}^{1/2}
        \left(
                {\widehat{\Psi}_{n,a}}^{-1}
                -
                {\Psi}_{n}^{-1}
        \right)
                {\widehat{\mathcal{T}}_{n,a}}^{1/2}
        \right\Vert_2
        \\
        =&~
        \left\Vert
        {\widehat{\mathcal{T}}_{n,a}}^{1/2}
        {\Psi}_{n}^{-1/2}
        \left(
                {\Psi}_{n}^{1/2}
                {\widehat{\Psi}_{n,a}}^{-1}
                {\Psi}_{n}^{1/2}
                -
                \bm{I}
        \right)
       {\Psi}_{n}^{-1/2}
                {\widehat{\mathcal{T}}_{n,a}}^{1/2}
        \right\Vert_2
        \\
        \leq &~
        \left\Vert
        {\Psi}_{n}^{1/2}
                {\widehat{\Psi}_{n,a}}^{-1}
                {\Psi}_{n}^{1/2}
                -
                \bm{I}
        \right\Vert_2
        \left\Vert
       {\Psi}_{n}^{-1/2}
                {\widehat{\mathcal{T}}_{n,a}}^{1/2}
        \right\Vert_2^2.
    \end{align*}
    By result \eqref{eq:lemma_result1} , we have $\left\Vert
        {\Psi}_{n}^{1/2}
                {\widehat{\Psi}_{n,a}}^{-1}
                {\Psi}_{n}^{1/2}
                -
                \bm{I}
        \right\Vert_2 = O_p
            \left(
            \frac{s+m+K}{\sqrt{n}}
            \right)$. For $\left\Vert
       {\Psi}_{n}^{-1/2}
                {\widehat{\mathcal{T}}_{n,a}}^{1/2}
        \right\Vert_2^2
        =
        \left\Vert
       {\Psi}_{n}^{-1/2}
                {\widehat{\mathcal{T}}_{n,a}}
                {\Psi}_{n}^{-1/2}
        \right\Vert_2$,
    notice that
    \begin{align*}
        &~
        \left\Vert
       {\Psi}_{n}^{-1/2}
                {\widehat{\mathcal{T}}_{n,a}}
                {\Psi}_{n}^{-1/2}
        \right\Vert_2
        \\
        =&~
        \left\Vert
       {\Psi}_{n}^{-1/2}
                \begin{pmatrix}
                    \bC^T 
                            \\
                            \bm{0}_{r\times (s+K)}^T
                \end{pmatrix}^T
                \widehat{\bK}_{n,a}^{-1}
                    \widehat{\bV}_{n,a}
                    \widehat{\bK}_{n,a}^{-1}
            \begin{pmatrix}
                    \bC^T 
                            \\
                            \bm{0}_{r\times (s+K)}^T
                \end{pmatrix}
                {\Psi}_{n}^{-1/2}
        \right\Vert_2
        \\
        =&~
        \left\Vert
       {\Psi}_{n}^{-1/2}
                \begin{pmatrix}
                    \bC^T 
                            \\
                            \bm{0}_{r\times (s+K)}^T
                \end{pmatrix}^T
                \bK_n^{-1/2}
                \cdot
                \bK_n^{1/2}
                \widehat{\bK}_{n,a}^{-1}
                    \widehat{\bV}_{n,a}
                    \widehat{\bK}_{n,a}^{-1}
                \bK_n^{1/2}
                \cdot
                \bK_n^{-1/2}
            \begin{pmatrix}
                    \bC^T 
                            \\
                            \bm{0}_{r\times (s+K)}^T
                \end{pmatrix}
                {\Psi}_{n}^{-1/2}
        \right\Vert_2
        \\
        \leq &~
        \left\Vert
                \bK_n^{1/2}
                \widehat{\bK}_{n,a}^{-1}
                    \widehat{\bV}_{n,a}
                    \widehat{\bK}_{n,a}^{-1}
                \bK_n^{1/2}
        \right\Vert_2
                \cdot
        \left\Vert
                \bK_n^{-1/2}
            \begin{pmatrix}
                    \bC^T 
                            \\
                            \bm{0}_{r\times (s+K)}^T
                \end{pmatrix}
                {\Psi}_{n}^{-1/2}
        \right\Vert_2^2
        \\
        \leq&~
        \left\Vert
                \bK_n
        \right\Vert_2
        \left\Vert
                \widehat{\bK}_{n,a}^{-1}
                    \widehat{\bV}_{n,a}
                    \widehat{\bK}_{n,a}^{-1}
        \right\Vert_2
                \cdot
        \left\Vert
            {\Psi}_{n}^{-1/2}
            {\Psi}_{n}
                {\Psi}_{n}^{-1/2}
        \right\Vert_2
        \\
        =&~
        O_p
            \left(
            \frac{s+m+K}{\sqrt{n}}
            \right),
    \end{align*}
    where the last equation is from the definition of $\Psi_n$ as well as the fact that $\lambda_{\max}(\bK_n)=O(1)$ and result \eqref{eq:lemma_result3+} proved before. Therefore, we proved that $\left\Vert
        {\widehat{\mathcal{T}}_{n,a}}^{1/2}
        \left(
                {\widehat{\Psi}_{n,a}}^{-1}
                -
                {\Psi}_{n}^{-1}
        \right)
                {\widehat{\mathcal{T}}_{n,a}}^{1/2}
        \right\Vert_2=O_p
            \left(
            \frac{s+m+K}{\sqrt{n}}
            \right)$.

    For the second term in equation \eqref{eq:lemma-10}, assume for now that we have the following result
    \begin{equation}
        \left\Vert
        {\widehat{\mathcal{T}}_{n,a}}^{-1/2}
                {\Psi}_{n}
                {\widehat{\mathcal{T}}_{n,a}}^{-1/2}
                -
                {{\mathcal{T}}_{n}}^{-1/2}
                {\Psi}_{n}
                {{\mathcal{T}}_{n}}^{-1/2}
        \right\Vert_2
        =
        O_p
            \left(
            \frac{s+m+K}{\sqrt{n}}
            \right).
        \label{eq:lemma-11}
    \end{equation}
    Next, we prove that $\left\Vert
                {{\mathcal{T}}_{n}}^{1/2}
                {\Psi}_{n}^{-1}
                {{\mathcal{T}}_{n}}^{1/2}
        \right\Vert_2=O(1)$ and $\left\Vert
        {\widehat{\mathcal{T}}_{n,a}}^{1/2}
                {\Psi}_{n}^{-1}
                {\widehat{\mathcal{T}}_{n,a}}^{1/2}
        \right\Vert_2=O_p(1)$. Notice that
    \begin{align}
        &~
        \lim\inf_{n}\lambda_{\min}
        \left(
                {{\mathcal{T}}_{n}}^{-1/2}
                {\Psi}_{n}
                {{\mathcal{T}}_{n}}^{-1/2}
        \right)
        \nonumber
        \\
        =&~
        \lim\inf_{n}
        \min_{\ba:\Vert\ba\Vert_2=1}
         \ba^T
                {{\mathcal{T}}_{n}}^{-1/2}
                \begin{pmatrix}
                    \bC^T 
                            \\
                            \bm{0}_{r\times (s+K)}^T
                \end{pmatrix}^T
                {\bK}_{n}^{-1}
            \begin{pmatrix}
                    \bC^T 
                            \\
                            \bm{0}_{r\times (s+K)}^T
                \end{pmatrix}
                {{\mathcal{T}}_{n}}^{-1/2}
        \ba
        \nonumber
        \\
        =&~
        \lim\inf_{n}
        \min_{\ba:\Vert\ba\Vert_2=1}
         \ba^T
                {{\mathcal{T}}_{n}}^{-1/2}
                \begin{pmatrix}
                    \bC^T 
                            \\
                            \bm{0}_{r\times (s+K)}^T
                \end{pmatrix}^T
                {\bK}_{n}^{-1}
                \bV_{n}^{1/2}
                \cdot
                \bV_{n}^{-1/2}
                {\bK}_{n}
                 \bV_{n}^{-1/2}
                 \cdot
                 \bV_{n}^{1/2}
                 {\bK}_{n}^{-1}
            \begin{pmatrix}
                    \bC^T 
                            \\
                            \bm{0}_{r\times (s+K)}^T
                \end{pmatrix}
                {{\mathcal{T}}_{n}}^{-1/2}
        \ba
        \nonumber
        \\
        \geq&~
        \lim\inf_{n}
        \lambda_{\min}
        \left(
                \bV_{n}^{-1/2}
                {\bK}_{n}
                 \bV_{n}^{-1/2}
        \right)
        \lambda_{\min}
        \left(
                {{\mathcal{T}}_{n}}^{-1/2}
                 {{\mathcal{T}}_{n}}
                {{\mathcal{T}}_{n}}^{-1/2}
        \right)
        >0,
        \label{eq:lemma-13}
    \end{align}
    where the last inequality is from the definition of $\Psi_n$ as well as the fact that $\lambda_{\max}(\bK_n^{-1})=O(1)$ and $\lambda_{\max}(\bV_n)=O(1)$. Similarly, we can prove result \eqref{eq:lemma_result6++}.  By the definition of $\Psi_n$ and ${\widehat{\mathcal{T}}_{n,a}}$, we have
    \begin{align}
        &~
         \lim\inf_{n}\lambda_{\min}
        \left(
        {\widehat{\mathcal{T}}_{n,a}}^{-1/2}
                {\Psi}_{n}
                {\widehat{\mathcal{T}}_{n,a}}^{-1/2}
        \right)
        \nonumber
        \\
        =&~
        \lim\inf_{n}
        \min_{\ba:\Vert\ba\Vert_2=1}
         \ba^T
                {\widehat{\mathcal{T}}_{n,a}}^{-1/2}
                \begin{pmatrix}
                    \bC^T 
                            \\
                            \bm{0}_{r\times (s+K)}^T
                \end{pmatrix}^T
                \widehat{\bK}_{n,a}^{-1}
                \widehat{\bV}_{n,a}^{1/2}
                \cdot
                \widehat{\bV}_{n,a}^{-1/2}
                 \widehat{\bK}_{n,a}
                 {\bK}_{n}^{-1}
                 \widehat{\bK}_{n,a}
                 \widehat{\bV}_{n,a}^{-1/2}
                 \cdot
                 \widehat{\bV}_{n,a}^{1/2}
                 \widehat{\bK}_{n,a}^{-1}
            \begin{pmatrix}
                    \bC^T 
                            \\
                            \bm{0}_{r\times (s+K)}^T
                \end{pmatrix}
                {\widehat{\mathcal{T}}_{n,a}}^{-1/2}
        \ba
        \nonumber
        \\
        \geq&~
        \lim\inf_{n}
        \lambda_{\min}
        \left(
                \widehat{\bV}_{n,a}^{-1/2}
                 \widehat{\bK}_{n,a}
                 {\bK}_{n}^{-1}
                 \widehat{\bK}_{n,a}
                 \widehat{\bV}_{n,a}^{-1/2}
        \right)>0,
        \label{eq:lemma-14}
    \end{align}
    where the last inequality is because
    \begin{align*}
        \left\Vert
                \widehat{\bV}_{n,a}^{1/2}
                 \widehat{\bK}_{n,a}^{-1}
                 {\bK}_{n}
                 \widehat{\bK}_{n,a}^{-1}
                 \widehat{\bV}_{n,a}^{1/2}
        \right\Vert_2
        \leq
        \left\Vert
                 {\bK}_{n}
        \right\Vert_2
        \left\Vert
                 \widehat{\bK}_{n,a}^{-1}
                 \widehat{\bV}_{n,a}
                 \widehat{\bK}_{n,a}^{-1}
        \right\Vert_2=O_p(1)
    \end{align*}
    along with the fact that $\lambda_{\max}(\bK_n^{-1})=O(1)$ and result \eqref{eq:lemma_result3+} proved before.
    Then, we have
    \begin{align*}
        &~
        \left\Vert
        {\widehat{\mathcal{T}}_{n,a}}^{1/2}
                {\Psi}_{n}^{-1}
                {\widehat{\mathcal{T}}_{n,a}}^{1/2}
                -
                {{\mathcal{T}}_{n}}^{1/2}
                {\Psi}_{n}^{-1}
                {{\mathcal{T}}_{n}}^{1/2}
        \right\Vert_2
        \\
        =&~
        \left\Vert
        {\widehat{\mathcal{T}}_{n,a}}^{1/2}
                {\Psi}_{n}^{-1}
                {\widehat{\mathcal{T}}_{n,a}}^{1/2}
        (
        {\widehat{\mathcal{T}}_{n,a}}^{-1/2}
                {\Psi}_{n}
                {\widehat{\mathcal{T}}_{n,a}}^{-1/2}
                -
        {{\mathcal{T}}_{n}}^{-1/2}
                {\Psi}_{n}
                {{\mathcal{T}}_{n}}^{-1/2}
        )
        {\widehat{\mathcal{T}}_{n,a}}^{1/2}
                {\Psi}_{n}^{-1}
                {\widehat{\mathcal{T}}_{n,a}}^{1/2}  
        \right\Vert_2
        \\
        \leq&~
        \left\Vert
        {\widehat{\mathcal{T}}_{n,a}}^{1/2}
                {\Psi}_{n}^{-1}
                {\widehat{\mathcal{T}}_{n,a}}^{1/2}
        \right\Vert_2
        \left\Vert
        {\widehat{\mathcal{T}}_{n,a}}^{-1/2}
                {\Psi}_{n}
                {\widehat{\mathcal{T}}_{n,a}}^{-1/2}
                -
        {{\mathcal{T}}_{n}}^{-1/2}
                {\Psi}_{n}
                {{\mathcal{T}}_{n}}^{-1/2}
        \right\Vert_2
        \left\Vert
                {\widehat{\mathcal{T}}_{n,a}}^{1/2}
                {\Psi}_{n}^{-1}
                {\widehat{\mathcal{T}}_{n,a}}^{1/2} 
        \right\Vert_2
        \\
        =&~
        O_p(1)\cdot
        O_p
            \left(
            \frac{s+m+K}{\sqrt{n}}
            \right)
            \cdot
        O(1)=
        O_p
            \left(
            \frac{s+m+K}{\sqrt{n}}
            \right).
    \end{align*}
    So it remains to prove equation \eqref{eq:lemma-11}. We start from proving that
    \begin{equation}
        \left\Vert
        {\widehat{\mathcal{T}}_{n,a}}^{-1}
                {\Psi}_{n}^2
                {\widehat{\mathcal{T}}_{n,a}}^{-1}
                -
                {{\mathcal{T}}_{n}}^{-1}
                {\Psi}_{n}^2
                {{\mathcal{T}}_{n}}^{-1}
        \right\Vert_2
        =
        O_p
            \left(
            \frac{s+m+K}{\sqrt{n}}
            \right).
        \label{eq:lemma-12}
    \end{equation}
    Notice that
    \begin{align*}
        &~
        \left\Vert
        {\widehat{\mathcal{T}}_{n,a}}^{-1}
                {\Psi}_{n}^2
                {\widehat{\mathcal{T}}_{n,a}}^{-1}
                -
                {{\mathcal{T}}_{n}}^{-1}
                {\Psi}_{n}^2
                {{\mathcal{T}}_{n}}^{-1}
        \right\Vert_2
        \\
        =&~
        \left\Vert
        (
        {\widehat{\mathcal{T}}_{n,a}}^{-1}
        -
        {{\mathcal{T}}_{n}}^{-1}
        )
        {\Psi}_{n}^2
        (
        {\widehat{\mathcal{T}}_{n,a}}^{-1}
        -
        {{\mathcal{T}}_{n}}^{-1}
        )
        +
        (
        {\widehat{\mathcal{T}}_{n,a}}^{-1}
        -
        {{\mathcal{T}}_{n}}^{-1}
        )
        {\Psi}_{n}^2
        {{\mathcal{T}}_{n}}^{-1}
        +
        {{\mathcal{T}}_{n}}^{-1}
        {\Psi}_{n}^2
        (
        {\widehat{\mathcal{T}}_{n,a}}^{-1}
        -
        {{\mathcal{T}}_{n}}^{-1}
        )
        \right\Vert_2
        \\
        \leq&~
        \left\Vert
        {\Psi}_{n}
        (
        {\widehat{\mathcal{T}}_{n,a}}^{-1}
        -
        {{\mathcal{T}}_{n}}^{-1}
        )
        \right\Vert_2^2
        +
        2
        \left\Vert
        {{\mathcal{T}}_{n}}^{-1}
        {\Psi}_{n}
        \right\Vert_2
        \left\Vert
        {\Psi}_{n}
        (
        {\widehat{\mathcal{T}}_{n,a}}^{-1}
        -
        {{\mathcal{T}}_{n}}^{-1}
        )
        \right\Vert_2
        \\
        =&~
        O_p
            \left(
            \frac{(s+m+K)^2}{n}
            \right)
        +
        \left\Vert
        {{\mathcal{T}}_{n}}^{-1}
        {\Psi}_{n}
        \right\Vert_2
                O_p
            \left(
            \frac{s+m+K}{\sqrt{n}}
            \right),
    \end{align*}
    where the last equation is from the result \eqref{eq:lemma_result5} proved before. 
    And $\left\Vert
        {{\mathcal{T}}_{n}}^{-1}
        {\Psi}_{n}
        \right\Vert_2=O(1)$ because
    \begin{align*}
        &~
        \left\Vert
        {{\mathcal{T}}_{n}}^{-1}
        {\Psi}_{n}
        \right\Vert_2
        \\
        \leq&~
        \left\Vert
        {{\mathcal{T}}_{n}}^{-1/2}
        \right\Vert_2
        \left\Vert
        {{\mathcal{T}}_{n}}^{-1/2}
        \begin{pmatrix}
                    \bC^T 
                            \\
                            \bm{0}_{r\times (s+K)}^T
                \end{pmatrix}^T
                {\bK}_{n}^{-1}
                {\bV}_{n}^{1/2}
                \cdot
                {\bV}_{n}^{-1/2}
                {\bK}_{n}
                {\bV}_{n}^{-1/2}
                \cdot
                {\bV}_{n}^{1/2}
                {\bK}_{n}^{-1}
            \begin{pmatrix}
                    \bC^T 
                            \\
                            \bm{0}_{r\times (s+K)}^T
                \end{pmatrix}
        \right\Vert_2
        \\
        \leq&~
        \left\Vert
        {{\mathcal{T}}_{n}}^{-1/2}
        \right\Vert_2
        \left\Vert
        {{\mathcal{T}}_{n}}^{-1/2}
        \begin{pmatrix}
                    \bC^T 
                            \\
                            \bm{0}_{r\times (s+K)}^T
                \end{pmatrix}^T
                {\bK}_{n}^{-1}
                {\bV}_{n}^{1/2}
        \right\Vert_2
        \left\Vert
        {\bV}_{n}^{-1/2}
                {\bK}_{n}
                {\bV}_{n}^{-1/2}
        \right\Vert_2
        \left\Vert
        {\bV}_{n}^{1/2}
                {\bK}_{n}^{-1}
            \begin{pmatrix}
                    \bC^T 
                            \\
                            \bm{0}_{r\times (s+K)}^T
                \end{pmatrix}
        \right\Vert_2
        \\
        =&~
        \left\Vert
        {{\mathcal{T}}_{n}}^{-1/2}
        \right\Vert_2
        \cdot
        1
        \cdot
        O(1)
        \cdot
        \left\Vert
        {{\mathcal{T}}_{n}}^{1/2}
        \right\Vert_2=O(1).
    \end{align*}
    Along with the fact that $s+m+K=o(\sqrt{n})$, we proved the desired equation \eqref{eq:lemma-12}. By applying Lemma \ref{lemma0}, we have
    \begin{align*}
        &~
        \left\Vert
        {\widehat{\mathcal{T}}_{n,a}}^{-1/2}
                {\Psi}_{n}
                {\widehat{\mathcal{T}}_{n,a}}^{-1/2}
                -
                {{\mathcal{T}}_{n}}^{-1/2}
                {\Psi}_{n}
                {{\mathcal{T}}_{n}}^{-1/2}
        \right\Vert_2
        \\
        \leq&~
        \left\Vert
        {\widehat{\mathcal{T}}_{n,a}}^{-1}
                {\Psi}_{n}^2
                {\widehat{\mathcal{T}}_{n,a}}^{-1}
                -
                {{\mathcal{T}}_{n}}^{-1}
                {\Psi}_{n}^2
                {{\mathcal{T}}_{n}}^{-1}
        \right\Vert_2
        \frac{1}{\lambda_{\min}({\widehat{\mathcal{T}}_{n,a}}^{-1/2}
                {\Psi}_{n}
                {\widehat{\mathcal{T}}_{n,a}}^{-1/2}
                +
                {{\mathcal{T}}_{n}}^{-1/2}
                {\Psi}_{n}
                {{\mathcal{T}}_{n}}^{-1/2})}.
    \end{align*}
    From results \eqref{eq:lemma-13} and \eqref{eq:lemma-14}, we have $1/{\lambda_{\min}({\widehat{\mathcal{T}}_{n,a}}^{-1/2}
                {\Psi}_{n}
                {\widehat{\mathcal{T}}_{n,a}}^{-1/2}
                +
                {{\mathcal{T}}_{n}}^{-1/2}
                {\Psi}_{n}
                {{\mathcal{T}}_{n}}^{-1/2})}=O_p(1)$.
    Along with result \eqref{eq:lemma-12}, we proved the equation \eqref{eq:lemma-11} as desired. This ends the proof for equation \eqref{eq:lemma_result6}.
    
\end{proof}

\paragraph{Proof of equation \eqref{eq:lemma_result7}}
\begin{proof}
    First proof that $\lambda_{\max}(\cT_n^{-1})=O(1)$. By definition of $\cT$, we have
    \begin{align*}
        \lim\inf_{n}\lambda_{\min}(\cT_n)
        &=
        \lim\inf_{n}
        \min_{\ba:\Vert\ba\Vert_2=1}
         \ba^T
        \cT_n
        \ba
        \\
        &=
        \lim\inf_{n}
        \min_{\ba:\Vert\ba\Vert_2=1}
        \ba^T
        \begin{pmatrix}
                    \bC^T 
                            \\
                            \bm{0}_{r\times (s+K)}^T
                \end{pmatrix}^T
                {\bK}_{n}^{-1}
                {\bV}_{n}
                {\bK}_{n}^{-1}
            \begin{pmatrix}
                    \bC^T 
                            \\
                            \bm{0}_{r\times (s+K)}^T
                \end{pmatrix}
        \ba
        \\
        &\geq
        \lim\inf_{n}
        \lambda_{\min}
        \left(
            {\bK}_{n}^{-1}
            {\bV}_{n}
            {\bK}_{n}^{-1}
        \right)
        \cdot
        \lim\inf_{n}
        \lambda_{\min}
        \left(
            \bC\bC^T
        \right)>0,
    \end{align*}
    where the last inequality is from the fact that $\lambda_{\max}(\bK_n)=O(1)$, $\lambda_{\max}(\bV_n^{-1})=O(1)$, as well as $\lambda_{\max}((\bC\bC^T)^{-1})=O(1)$ given in \textcolor{black}{assumption (A3)}. Along with the fact that $\Vert\bh_n\Vert_2=O(\sqrt{\min(m+s+K-r, r)/n})\leq O(\sqrt{r/n})$ given in \textcolor{black}{assumption (A3)}, we proved the desired result \eqref{eq:lemma_result7}.
\end{proof}

\subsection{Proof of Theorem \ref{theorem:estimator_statistical_properties}}

\begin{proof}
        Since the true coefficients satisfy $\cB^*_{(\cM\cup S\cup\cK)^c}=\bm{0}$, we define an oracle estimator $\hat{\cB}$ as solution to
        \begin{equation}
            \max_{\cB}Q_n(\cB)~~\text{s.t.}~ \bC\cB_\cM = \bt,~\cB_{(\cM\cup S\cup\cK)^c}=\bm{0}.
            \label{eq:oracle_optimization}
        \end{equation}
        \paragraph{Step 1: Show that $\Vert\hat{\cB} - \cB^*\Vert_2=O_p(\sqrt{(m+s+K-r)/n})$.}

            Define a $(p+K)$-dimensional vector $\cB_0$ as
            \begin{equation*}
                \begin{cases}
                    \cB_{0,\cM}
                        & = \cB_\cM^* - \bC^T(\bC\bC^T)^{-1}\bh_n,
                        \\
                    
                    \cB_{0,\cM^c}
                        & = \cB_{\cM^c}^*.
                \end{cases}
            \end{equation*}
            It is immediate to see that
            \begin{equation*}
                \bC\cB_{0,\cM} - \bt = \bm{0}.
            \end{equation*}
            And it follows from \textcolor{black}{(A3)} that 
            \begin{equation}
                \Vert\cB_0 - \cB^*\Vert_2^2 
                = \Vert \bC^T(\bC\bC^T)^{-1}\bh_n \Vert_2^2
                \leq \lambda_{\max}((\bC\bC^T)^{-1})\Vert \bh_n\Vert_2^2
                \leq O((m+s+K-r)/n).
                \label{eq:step1-1}
            \end{equation}
            So it suffices to show that $\Vert \hat{\cB} - \cB_0 \Vert_2^2 = O_p((m+s+K-r)/n)$.
            \\
            Observe that $\bC\cB_{0,\cM}=\bt$, then for any $\cB$ with $\bC\cB_{\cM}=\bt$, we have $\bC(\cB_{\cM} - \cB_{0,\cM})=\bm{0}$, i.e., $\cB_{\cM} - \cB_{0,\cM}\in Null(\bC)$. Take basis matrix $\bm{Z}\in\mathbb{R}^{m\times (m-r)}$ of $Null(\bC)$, i.e., $\bC\bm{Z}=\bm{0}$, and $\bm{Z}^T\bm{Z}=\bm{I}_{m-r}$. Then for any $\cB_\cM$ s.t. $\bC\cB_{\cM}=\bt$, $\exists \bm{v}\in\mathbb{R}^{m-r}$, s.t. $\cB_\cM = \cB_{0,\cM}+\bm{Z}\bm{v}$. Then define $\overline{Q}_n(\bm{\delta}) = Q_n(\cB(\bm{\delta}))$ for any $\bm{\delta}\in\mathbb{R}^{(m-r)+s+K}$ where $\cB(\bm{\delta})$ is defined as
            \begin{equation*}
                \begin{cases}
                    \cB(\bm{\delta})_{\cM} 
                        & = \cB_{0,\cM} + \bm{Z} \bm{\delta}_{J_0},
                    \\
                    \cB(\bm{\delta})_{S\cup\cK}
                        & = \cB_{0,S\cup\cK} + \bm{\delta}_{J_0^c},
                    \\
                    \cB(\bm{\delta})_{(\cM \cup S\cup\cK)^c}
                        & = \cB_{0,(\cM \cup S\cup\cK)^c} ,
                \end{cases}
            \end{equation*}
            where $J_0=[1,2,\dots, m-r]$. Then solving the constrained optimization problem \eqref{eq:oracle_optimization} is equivalent to solving $\max_{
            \bm{\delta}\in\mathbb{R}^{(m-r)+s+K}}\overline{Q}_n(\bm{\delta})$.
            Since $\Vert \bm{Z}\bm{\delta}_{J_0}\Vert_2^2 = \Vert \bm{\delta}_{J_0}\Vert_2^2$, it suffices to show that there exists a local maximizer $\bm{\delta}_0$ of $\overline{Q}_n(\bm{\delta})$ such that $\Vert \bm{\delta}_0 \Vert_2 = O_p(\sqrt{(m+s+K-r)/n})$.
            
            Define a neighborhood $\cN_\tau = \{\bm{\delta}\in\mathbb{R}^{(m-r)+s+K}: \Vert \bm{\delta}\Vert_2\leq \tau\sqrt{(m+s+K-r)/n}\}$, and an event $H_n:=\left\{\overline{Q}_n(\bm{0}) > \max_{\bm{\delta}\in {\partial\cN_\tau}}\overline{Q}_n(\bm{\delta})\right\}$, where ${\partial\cN_\tau}$ denotes the boundary of ${\cN_\tau}$. On the event $H_n$, there must exist a local maximizer in $\cN_\tau$. Hence, it suffices to show $\mathbb{P}(H_n)\rightarrow 1$ as $n\rightarrow \infty$ and some sufficiently large $\tau$.
            
            For any $\bm{\delta}\in\partial\cN_\tau$, by the second-order Taylor's expansion
            \begin{equation}
                \overline{Q}_n(\bm{\delta}) - \overline{Q}_n(\bm{0})
                    =
                \bm{\delta}^T\bm{\nu} - \frac{1}{2} \bm{\delta}^T \bm{D} \bm{\delta},
                \label{eq:step1-2}
            \end{equation}
            where
            \begin{align*}
                \bm{\nu} 
                    & = \sum_{k=1}^{K}w_k \frac{1}{n}
                        \begin{pmatrix}
                            \bm{Z}^T \bm{X}_{\cM}^T
                            \\
                            \bm{X}_{S}^T
                            \\
                            -\bm{E}_k^T
                        \end{pmatrix}
                        \bm{H}(\bm{X}^k\cB_0)
                        \left\{
                            \bm{Y}^k - \bm{\mu}(\bm{X}^k\cB_0)
                        \right\}
                        -
                        \begin{pmatrix}
                            \bm{0}_{m-r}
                            \\
                            \lambda_{n,0}\bar{\rho}(\cB_{0,S},\lambda_{n,0})
                            \\
                            \bm{0}_{K}
                        \end{pmatrix},
                \\
                \bm{D}
                    & = \sum_{k=1}^{K}w_k \frac{1}{n}
                        \begin{pmatrix}
                            \bm{Z}^T \bm{X}_{\cM}^T
                            \\
                            \bm{X}_{S}^T
                            \\
                            -\bm{E}_k^T
                        \end{pmatrix}
                        \hat{\Sigma}(\bm{X}^k\cB^{**})
                        \begin{pmatrix}
                            \bm{Z}^T \bm{X}_{\cM}^T
                            \\
                            \bm{X}_{S}^T
                            \\
                            -\bm{E}_k^T
                        \end{pmatrix}^T
                        -
                        \begin{pmatrix}
                            \bm{0}_{(m-r)\times(m-r)} & \bm{0}_{(m-r)\times s} & \bm{0}_{(m-r)\times K}
                            \\
                            \bm{0}_{s \times(m-r)} & \Lambda_0 & \bm{0}_{s \times K} 
                            \\
                            \bm{0}_{K \times(m-r)} & \bm{0}_{K \times s} & \bm{0}_{K \times K}
                        \end{pmatrix},
            \end{align*}
            where $\cB^{**}$ lies in the line segment joining $\cB(\bm{\delta})$ and $\cB_0$, and $\Lambda_0$ is a $r\times r$ diagonal matrix with non-negative diagonal elements and 
            \begin{align*}
                \bar{\rho}(\cB_{S},\lambda_{n,0})
                    & = \left(
                    \text{sign}(\cB_{j})\rho^{\prime}(\vert\cB_{j}\vert,\lambda_{n,0}),j\in S
                    \right)^T
                \\
                \bm{H}(\bm{X}^k\cB)
                    & = \text{diag}
                        \left\{
                            h^{\prime}((\bx_i^k)^T\cB),~i=1,\dots,n
                        \right\},
                \\
                \bm{\mu}(\bm{X}^k\cB)
                    & = \begin{pmatrix}
                            \Phi((\bx_1^k)^T\cB),\cdots,\Phi((\bx_n^k)^T\cB)
                        \end{pmatrix}^T,
                \\
                \hat{\Sigma}(\bm{X}^k\cB)
                    & = \text{diag}
                        \left\{
                            \tilde{y}_{ki}\ell^{\prime\prime}((\bx_{i}^k)^T\cB)+(1-\tilde{y}_{ki})\ell^{\prime\prime}(-(\bx_{i}^k)^T\cB),~i=1,\dots,n
                        \right\},
            \end{align*}
            where $\ell(t)=-\log(\Phi(t))$ is the probit loss function introduced and studied in \cite{zhou2023nonparametric} with $\ell(t)^{\prime\prime}$ bounded within $[0,1]$. By the definition of $\cB(\bm{\delta})$ and $\cB_0$, we have $\cB^{**}_{(\cM\cup S\cup \cK)^c}=\bm{0}$. Moreover, it follows from equation \eqref{eq:step1-1} and $\bm{\delta}\in\partial\cN_\tau$ that
            \begin{align}
                \Vert \cB^{**} - \cB^* \Vert_2
                & \leq
                \tau\sqrt{(s-r+m+K)/n} + O(\sqrt{(s-r+m+K)/n})
                \nonumber
                \\
                & \ll
                \sqrt{(s+m+K)\log(n)/n},
                \label{eq:step1-B**}
            \end{align}
            for $n$ sufficiently large and for $\tau\ll \sqrt{\log{n}}$. Therefore, we have $\cB^{**}\in\cN^*$. By definition of $\kappa_{0,0}$, we have $\lambda_{\max}(\Lambda_0)\leq\lambda_{n,0}\kappa_{0,0}$. Let
            \begin{equation*}
                \bm{L} = 
                    \begin{pmatrix}
                        \bm{Z} & \bm{0}_{m\times(s+K)}
                        \\
                        \bm{0}_{(s+K)\times (m-r)} & \bm{I}_{s+K}
                    \end{pmatrix},
            \end{equation*}
            we have
            \begin{equation*}
                \sum_{k=1}^{K}w_k \frac{1}{n}
                        \begin{pmatrix}
                            \bm{Z}^T \bm{X}_{\cM}^T
                            \\
                            (\bm{X}^k_{S\cup \cK})^T
                        \end{pmatrix}
                        \hat{\Sigma}(\bm{X}^k\cB^{**})
                        \begin{pmatrix}
                            \bm{Z}^T \bm{X}_{\cM}^T
                            \\
                            (\bm{X}^k_{S\cup \cK})^T
                        \end{pmatrix}^T
                =
                \sum_{k=1}^{K}w_k \frac{1}{n}
                        \bm{L}^T
                        \begin{pmatrix}
                            \bm{X}_{\cM}^T
                            \\
                            (\bm{X}^k_{S\cup \cK})^T
                        \end{pmatrix}
                        \hat{\Sigma}(\bm{X}^k\cB^{**})
                        \begin{pmatrix}
                            \bm{X}_{\cM}^T
                            \\
                            (\bm{X}^k_{S\cup \cK})^T
                        \end{pmatrix}^T
                        \bm{L}.
            \end{equation*}
            Then on the event defined below:
            \begin{equation}
            E_n:=
            \left\{
                \left\Vert
                    \sum_{k=1}^{K}w_k \frac{1}{n}
                        \begin{pmatrix}
                            \bm{X}_{\cM}^T
                            \\
                            (\bm{X}^k_{S\cup \cK})^T
                        \end{pmatrix}
                        \left(
                        \hat{\Sigma}(\bm{X}^k\cB^*)
                        -
                        {\Sigma}(\bm{X}^k\cB^*)
                        \right)
                        \begin{pmatrix}
                            \bm{X}_{\cM}^T
                            \\
                            (\bm{X}^k_{S\cup \cK})^T
                        \end{pmatrix}^T
                \right\Vert_2
                < \frac{c}{8}
            \right\},
            \label{eq:step1_event}
            \end{equation}
            where ${\Sigma}(\bm{X}^k\cB) = \text{diag}
                        \left\{
                            \frac{\varphi^2}{\Phi(1-\Phi)}((\bx_{i}^k)^T\cB)
                            ,~i=1,\dots,n
                        \right\}$, and $c$ is the constant in \textcolor{black}{the 6th condition of (A1)},
            we must have 
            \begin{equation}
                \lambda_{\min}
                    \left\{
                        \sum_{k=1}^{K}w_k \frac{1}{n}
                        \begin{pmatrix}
                            \bm{Z}^T \bm{X}_{\cM}^T
                            \\
                            (\bm{X}^k_{S\cup \cK})^T
                        \end{pmatrix}
                        \hat{\Sigma}(\bm{X}^k\cB^{**})
                        \begin{pmatrix}
                            \bm{Z}^T \bm{X}_{\cM}^T
                            \\
                            (\bm{X}^k_{S\cup \cK})^T
                        \end{pmatrix}^T
                    \right\}\geq \frac{3}{4}c,
                    \label{eq:step1-D1}
            \end{equation}
            for sufficiently large $n$. This is because for any $\bm{\delta}\in\mathbb{R}^{m-r+s+K}$, notice that $\Vert\bm{L}\bm{\delta}\Vert_2=\Vert\bm{\delta}\Vert_2$, then
            \begin{align}
                &~
                \bm{\delta}^T
                \sum_{k=1}^{K}w_k \frac{1}{n}
                        \begin{pmatrix}
                            \bm{Z}^T \bm{X}_{\cM}^T
                            \\
                            (\bm{X}^k_{S\cup \cK})^T
                        \end{pmatrix}
                        \hat{\Sigma}(\bm{X}^k\cB^{**})
                        \begin{pmatrix}
                            \bm{Z}^T \bm{X}_{\cM}^T
                            \\
                            (\bm{X}^k_{S\cup \cK})^T
                        \end{pmatrix}^T
                \bm{\delta}
                \nonumber
                \\
                =
                &~
                (\bm{L}\bm{\delta})^T
                \sum_{k=1}^{K}w_k \frac{1}{n}
                        \begin{pmatrix}
                            \bm{X}_{\cM}^T
                            \\
                            (\bm{X}^k_{S\cup \cK})^T
                        \end{pmatrix}
                        \left(
                        {\Sigma}(\bm{X}^k\cB^{**})
                        +
                        \hat{\Sigma}(\bm{X}^k\cB^{**})
                        -
                        {\Sigma}(\bm{X}^k\cB^{**})
                        \right)
                        \begin{pmatrix}
                            \bm{X}_{\cM}^T
                            \\
                            (\bm{X}^k_{S\cup \cK})^T
                        \end{pmatrix}^T
                (\bm{L}\bm{\delta})
                \nonumber
                \\
                \geq
                &~
                \lambda_{\min}
                \left(
                    \sum_{k=1}^{K}w_k \frac{1}{n}
                        \begin{pmatrix}
                            \bm{X}_{\cM}^T
                            \\
                            (\bm{X}^k_{S\cup \cK})^T
                        \end{pmatrix}
                        {\Sigma}(\bm{X}^k\cB^{**})
                        \begin{pmatrix}
                            \bm{X}_{\cM}^T
                            \\
                            (\bm{X}^k_{S\cup \cK})^T
                        \end{pmatrix}^T
                \right)
                \Vert \bm{L}\bm{\delta}\Vert_2^2
                \nonumber
                \\
                -
                &~
                \lambda_{\max}
                \left(
                    \sum_{k=1}^{K}w_k \frac{1}{n}
                        \begin{pmatrix}
                            \bm{X}_{\cM}^T
                            \\
                            (\bm{X}^k_{S\cup \cK})^T
                        \end{pmatrix}
                        \left(
                        \hat{\Sigma}(\bm{X}^k\cB^{**})
                        -
                        {\Sigma}(\bm{X}^k\cB^{**})
                        \right)
                        \begin{pmatrix}
                            \bm{X}_{\cM}^T
                            \\
                            (\bm{X}^k_{S\cup \cK})^T
                        \end{pmatrix}^T
                \right)
                \Vert \bm{L}\bm{\delta}\Vert_2^2
                \nonumber
                \\
                \geq
                &~
                c\Vert \bm{\delta}\Vert_2^2
                -
                \lambda_{\max}(\bm{U_n})\Vert \bm{\delta}\Vert_2^2,
                \label{eq:step1-boundD}
            \end{align}
            where the last inequality is from the fact that $\cB^{**}\in\cN^*$ and \textcolor{black}{the 6th condition of (A1)}, and 
            \begin{equation*}
                \bm{U}_n:=\sum_{k=1}^{K}w_k \frac{1}{n}
                        \begin{pmatrix}
                            \bm{X}_{\cM}^T
                            \\
                            (\bm{X}^k_{S\cup \cK})^T
                        \end{pmatrix}
                        \left(
                        \hat{\Sigma}(\bm{X}^k\cB^{**})
                        -
                        {\Sigma}(\bm{X}^k\cB^{**})
                        \right)
                        \begin{pmatrix}
                            \bm{X}_{\cM}^T
                            \\
                            (\bm{X}^k_{S\cup \cK})^T
                        \end{pmatrix}^T.
            \end{equation*}
       
            Since
            \begin{align*}
                \hat{\Sigma}(\bm{X}^k\cB^{**})
                        -
                        {\Sigma}(\bm{X}^k\cB^{**})
                & =
                \left(\hat{\Sigma}(\bm{X}^k\cB^{**}) - \hat{\Sigma}(\bm{X}^k\cB^*)\right)
                +
                \left({\Sigma}(\bm{X}^k\cB^*) - {\Sigma}(\bm{X}^k\cB^{**})\right)
                \\
                &
                +
                \left(\hat{\Sigma}(\bm{X}^k\cB^*) - {\Sigma}(\bm{X}^k\cB^*)\right),
            \end{align*}
            we have $\bm{U}_n = \bm{U}_n^1+\bm{U}_n^2+\bm{U}_n^3$, where
            \begin{align*}
                \bm{U}_n^1
                &=
                \sum_{k=1}^{K}w_k \frac{1}{n}
                        \begin{pmatrix}
                            \bm{X}_{\cM}^T
                            \\
                            (\bm{X}^k_{S\cup \cK})^T
                        \end{pmatrix}
                         \left(\hat{\Sigma}(\bm{X}^k\cB^{**}) - \hat{\Sigma}(\bm{X}^k\cB^*)\right)
                        \begin{pmatrix}
                            \bm{X}_{\cM}^T
                            \\
                            (\bm{X}^k_{S\cup \cK})^T
                        \end{pmatrix}^T,
                \\
                \bm{U}_n^2
                &=
                \sum_{k=1}^{K}w_k \frac{1}{n}
                        \begin{pmatrix}
                            \bm{X}_{\cM}^T
                            \\
                            (\bm{X}^k_{S\cup \cK})^T
                        \end{pmatrix}
                         \left({\Sigma}(\bm{X}^k\cB^*) - {\Sigma}(\bm{X}^k\cB^{**})\right)
                        \begin{pmatrix}
                            \bm{X}_{\cM}^T
                            \\
                            (\bm{X}^k_{S\cup \cK})^T
                        \end{pmatrix}^T,
                \\
                \bm{U}_n^3
                &=
                \sum_{k=1}^{K}w_k \frac{1}{n}
                        \begin{pmatrix}
                            \bm{X}_{\cM}^T
                            \\
                            (\bm{X}^k_{S\cup \cK})^T
                        \end{pmatrix}
                         \left(\hat{\Sigma}(\bm{X}^k\cB^*) - {\Sigma}(\bm{X}^k\cB^*)\right)
                        \begin{pmatrix}
                            \bm{X}_{\cM}^T
                            \\
                            (\bm{X}^k_{S\cup \cK})^T
                        \end{pmatrix}^T.
            \end{align*}

            By \textcolor{black}{the 1st condition of (A1)} and the fact that $\Vert \cB^{**} - \cB^* \Vert_2\ll\sqrt{(s+m+K)\log(n)/n}$ as long as $\tau\ll \sqrt{\log(n)}$ proved in \eqref{eq:step1-B**}, we have $\max_{1\leq i\leq n}\max_{1\leq k\leq K}\vert (\bx_{i}^k)^T(\cB^{**} - \cB^*)\vert =o(1)$. This, combined with \textcolor{black}{the 4th condition of (A1)}, implies that $\lambda_{\max}(\bm{U}_n^1)<c/16$ and $\lambda_{\max}(\bm{U}_n^2)<c/16$ when $n$ is sufficiently large. Therefore, on event $E_n$, we must have
            \begin{equation*}
                \lambda_{\max}(\bm{U}_n)\leq\lambda_{\max}(\bm{U}_n^1) + \lambda_{\max}(\bm{U}_n^2) + \lambda_{\max}(\bm{U}_n^3)\leq\frac{c}{4},
            \end{equation*}
            for sufficiently large $n$. This, combined with inequality \eqref{eq:step1-boundD}, implies that
            \begin{align*}
                \bm{\delta}^T
                \sum_{k=1}^{K}w_k \frac{1}{n}
                        \begin{pmatrix}
                            \bm{Z}^T \bm{X}_{\cM}^T
                            \\
                            (\bm{X}^k_{S\cup \cK})^T
                        \end{pmatrix}
                        \hat{\Sigma}(\bm{X}^k\cB^{**})
                        \begin{pmatrix}
                            \bm{Z}^T \bm{X}_{\cM}^T
                            \\
                            (\bm{X}^k_{S\cup \cK})^T
                        \end{pmatrix}^T
                \bm{\delta}
                \geq
                c\Vert \bm{\delta}\Vert_2^2
                -
                \frac{c}{4}\Vert \bm{\delta}\Vert_2^2=\frac{3c}{4}\Vert \bm{\delta}\Vert_2^2,
            \end{align*}
            i.e., inequality \eqref{eq:step1-D1} holds on event $E_n$.
            
            
            By condition $\lambda_{n,0}\kappa_{0,0}=o(1)$ in \textcolor{black}{(A2)}, we have $\lambda_{\max}(\Lambda_0)=o(1)$. Thus, combining with inequality \eqref{eq:step1-D1}, we have $\lambda_{\min}(\bm{D})\geq \frac{1}{2}{c}$ for some sufficiently large $n$. Therefore, on the event $E_n$ defined in \eqref{eq:step1_event}, it follows from equation \eqref{eq:step1-2} that 
            \begin{align*}
                \sup_{\bm{\delta}\in\partial\cN_\tau}\overline{Q}_n(\bm{\delta}) - \overline{Q}_n(\bm{0})
                &
                \leq
                \sup_{\bm{\delta}\in\partial\cN_\tau}
                \left(
                \Vert \bm{\delta}\Vert_2\Vert\bm{\nu}\Vert_2
                -
                \frac{1}{2} \lambda_{\min}(D)\Vert \bm{\delta}\Vert_2^2
                \right)
                \\
                & \leq
                \sup_{\bm{\delta}\in\partial\cN_\tau}
                    \left(
                        \Vert \bm{\delta}\Vert_2\Vert\bm{\nu}\Vert_2 - \frac{c}{4}\Vert\bm{\delta}\Vert_2^2
                    \right)
                \\
                & = 
                \Vert\bm{\nu}\Vert_2 \tau\sqrt{\frac{m-r+s+K}{n}} - \frac{c\tau^2}{4}\frac{(m-r+s+K)}{n}.
            \end{align*}
            By Markov's inequality, 
            \begin{align*}
                \mathbb{P}(H_n) 
                &\geq
                \mathbb{P}(H_n\cap E_n) 
                \\
                &\geq 
                    \mathbb{P}
                        \left(
                        \left\{
                            \Vert \bm{\nu}\Vert_2 < \bar{c}\tau\sqrt{\frac{m-r+s+K}{n}}
                        \right\}
                        \cap
                        E_n
                        \right)
                \\
                &
                \geq
                \mathbb{P}
                        \left(
                            \Vert \bm{\nu}\Vert_2 < \bar{c}\tau\sqrt{\frac{m-r+s+K}{n}}
                        \right)
                -
                 \mathbb{P}(E_n^c)
                \\
                &
                    \geq
                    1-\frac{n\mathbb{E}\Vert\bm{\nu}\Vert_2^2}{\bar{c}^2\tau^2(m-r+s+K)}
                    -
                     \mathbb{P}(E_n^c)
                   ,
            \end{align*}
            for any fixed $\tau$ satisfying $\tau\leq\frac{1}{2}\sqrt{\log n}$ and for $n$ sufficiently large. 
            
            Next show that $\mathbb{E}\Vert\bm{\nu}\Vert_2^2 = O((m-r+s+K)/n)$. By Cauchy-Schwartz inequality,
            \begin{align*}
                \mathbb{E}\Vert\bm{\nu}\Vert_2^2 
                    & \leq
                    2\Vert \lambda_{n,0}\bar{\rho}(\cB_{0,S};\lambda_{n,0}) \Vert_2^2
                    +
                    2\mathbb{E}\left\Vert
                        \sum_{k=1}^{K}w_k \frac{1}{n}
                        \begin{pmatrix}
                            \bm{Z}^T \bm{X}_{\cM}^T
                            \\
                            (\bm{X}^k_{S\cup \cK})^T
                        \end{pmatrix}
                        \bm{H}(\bm{X}^k\cB_0)
                        \left\{
                            \bm{Y}^k - \bm{\mu}(\bm{X}^k\cB_0)
                        \right\}
                    \right\Vert_2^2
                    \\
                    & \overset{\Delta}{=}
                    2I_0+2I_1
            \end{align*}
            By the concavity of $\rho$, its derivative $\rho^{\prime}$ is monotone decreasing. Along with the definition of $d_n$, and \textcolor{black}{the 2nd condition of (A2)}, we have
            \begin{equation*}
                I_0\leq s \left\{
                    \lambda_{n,0}\rho^\prime(d_n)
                \right\}^2
                =
                o(1/n).
            \end{equation*} 
            Next show that $I_1=O(\frac{m-r+s+K}{n})$:
            \begin{align*}
                I_1
                &=
                \frac{1}{n^2}\mathbb{E}\left\Vert
                        \sum_{k=1}^{K}w_k 
                        \begin{pmatrix}
                            \bm{Z}^T \bm{X}_{\cM}^T
                            \\
                            (\bm{X}^k_{S\cup \cK})^T
                        \end{pmatrix}
                        \bm{H}(\bm{X}^k\cB_0)
                        \left\{
                            \bm{Y}^k - \bm{\mu}(\bm{X}^k\cB_0)
                        \right\}
                    \right\Vert_2^2
                \\
                &\leq 
                \frac{1}{n^2}\mathbb{E}
                \left[
                    \sum_{k=1}^{K}w_k 
                    \left\Vert
                        \begin{pmatrix}
                            \bm{Z}^T \bm{X}_{\cM}^T
                            \\
                            (\bm{X}^k_{S\cup \cK})^T
                        \end{pmatrix}
                        \bm{H}(\bm{X}^k\cB_0)
                        \left\{
                            \bm{Y}^k - \bm{\mu}(\bm{X}^k\cB_0)
                        \right\}
                    \right\Vert_2^2
                \right]
                 \\
                 & = 
                 \frac{1}{n}
                 \text{tr}\left[
                    \frac{1}{n}
                    \sum_{k=1}^{K}w_k 
                    \begin{pmatrix}
                            \bm{Z}^T \bm{X}_{\cM}^T
                            \\
                            (\bm{X}^k_{S\cup \cK})^T
                        \end{pmatrix}
                        \text{diag}\left\{(h^{\prime})^2((\bx_i^k)^T\cB_0)\cdot\Phi(1-\Phi)((\bx_i^k)^T\cB^*),i=1,\dots,n\right\}
                        \begin{pmatrix}
                            \bm{Z}^T \bm{X}_{\cM}^T
                            \\
                            (\bm{X}^k_{S\cup \cK})^T
                        \end{pmatrix}^T
                 \right]
                 \\
                 & +
                 \frac{1}{n}\cdot
                 \frac{1}{n}
                 \sum_{k=1}^{K}w_k 
                 \left\Vert
                    \begin{pmatrix}
                            \bm{Z}^T \bm{X}_{\cM}^T
                            \\
                            (\bm{X}^k_{S\cup \cK})^T
                        \end{pmatrix}
                        \bm{H}(\bm{X}^k\cB_0)
                        \left\{
                            \bm{\mu}(\bm{X}^k\cB_0) - \bm{\mu}(\bm{X}^k\cB^*) 
                        \right\}
                 \right\Vert_2^2
                 \\
                 & \overset{\Delta}{=}
                 \frac{1}{n}I_2+\frac{1}{n}I_3,
            \end{align*}
            where the first inequality is by the convexity of $\Vert\cdot\Vert_2^2$.
            
            To bound $I_2$, consider for any $i=1,\dots, n$, $k=1,\dots,K$,
            \begin{align*}
                &~
                (h^{\prime})^2((\bx_i^k)^T\cB_0)\cdot\Phi(1-\Phi)((\bx_i^k)^T\cB^*)
                \\
                = &~
                \left[
                    h^{\prime}((\bx_i^k)^T\cB^*) + h^{\prime\prime}((\bx_i^k)^T\cB^{**})\bx_{i,\cM}^T(\cB_{0,\cM} - \cB^*_{\cM})
                \right]^2
                \cdot\Phi(1-\Phi)((\bx_i^k)^T\cB^*)
                \\
                = &~
                \frac{\varphi^2}{\Phi(1-\Phi)}((\bx_i^k)^T\cB^*)
                +
                2\varphi((\bx_i^k)^T\cB)h^{\prime\prime}((\bx_i^k)^T\cB^{**}) \bx_{i,\cM}^T(\cB_{0,\cM} - \cB^*_{\cM})
                \\
                &~
                +
                \left[
                    h^{\prime\prime}((\bx_i^k)^T\cB^{**})\bx_{i,\cM}^T(\cB_{0,\cM} - \cB^*_{\cM})
                \right]^2 
                \cdot\Phi(1-\Phi)((\bx_i^k)^T\cB^*)
                \\
                \leq &~
                4{\varphi(0)^2} + O(\bx_{i,\cM}^T(\cB_{0,\cM} - \cB^*_{\cM})) + O(\left[\bx_{i,\cM}^T(\cB_{0,\cM} - \cB^*_{\cM})\right]^2)
                \\
                = &~
                4{\varphi(0)^2} + o(1)
            \end{align*}
            where the first equality is from the Taylor's expansion and $\cB_{\cM^c}^* = \cB_{0,\cM^c}$; the inequality is from the fact that $\frac{\varphi^2}{\Phi (1-\Phi)}(t)\leq \frac{\varphi^2}{\Phi (1-\Phi)}(0)$, $\varphi(t)\leq\varphi(0)$, $\vert h^{\prime\prime}(t)\vert\leq 1$ and $\Phi(1-\Phi)(t)\leq 1/4$; the last equality is because $\bx_{i,\cM}^T(\cB_{0,\cM} - \cB^*_{\cM})=o(1)$ derived from \textcolor{black}{the 1st condition of (A1)} and the fact that $\Vert \cB_{\cM}^* - \cB_{0,\cM}\Vert_2=O(\sqrt{(m-r+s+K)/n})$. Therefore, we have
            \begin{align*}
                I_2
                \leq
                &~
                 \text{tr}\left[
                    \frac{1}{n}
                    \sum_{k=1}^{K}w_k 
                    \begin{pmatrix}
                            \bm{Z}^T \bm{X}_{\cM}^T
                            \\
                            (\bm{X}^k_{S\cup \cK})^T
                        \end{pmatrix}
                        \begin{pmatrix}
                            \bm{Z}^T \bm{X}_{\cM}^T
                            \\
                            (\bm{X}^k_{S\cup \cK})^T
                        \end{pmatrix}^T
                 \right]
                 \cdot
                 O(1)
                 \\
                 =
                 &~
                 \text{tr}\left[
                 \bm{L}^T
                 \frac{1}{n}
                 \sum_{k=1}^{K}w_k
                 \begin{pmatrix}
                            \bm{X}_{\cM}^T
                            \\
                            (\bm{X}^k_{S\cup \cK})^T
                        \end{pmatrix}
                        \begin{pmatrix}
                            \bm{X}_{\cM}^T
                            \\
                            (\bm{X}^k_{S\cup \cK})^T
                        \end{pmatrix}^T
                 \bm{L}
                 \right]
                 \cdot
                 O\left(1\right)
                 \\
                 \leq
                 &~
                 (m-r+s+K)
                 \cdot
                 \lambda_{\max}
                 \left(
                 \frac{1}{n}
                 \sum_{k=1}^{K}w_k
                 \begin{pmatrix}
                            \bm{X}_{\cM}^T
                            \\
                            (\bm{X}^k_{S\cup \cK})^T
                        \end{pmatrix}
                        \begin{pmatrix}
                            \bm{X}_{\cM}^T
                            \\
                            (\bm{X}^k_{S\cup \cK})^T
                        \end{pmatrix}^T
                 \right)
                 \cdot
                 O\left(1\right).
            \end{align*}
            Along with \textcolor{black}{the 4th condition of (A1)}, we have $I_2/n=O((m-r+s+K)/n)$. 
            
            To bound $I_3$, consider the $i$-th coordinate of vector $\bm{H}(\bm{X}^k\cB_0)
                        \left(
                            \bm{\mu}(\bm{X}^k\cB_0) - \bm{\mu}(\bm{X}^k\cB^*) 
                        \right)$ for any $i=1,\dots,n$, $k=1,\dots, K$,
            \begin{align*}
                &~
                h^{\prime}((\bx_i^k)^T\cB_0)
                \cdot
                \left(
                \Phi((\bx_i^k)^T\cB_0) - \Phi((\bx_i^k)^T\cB^*)
                \right)
                \\
                =
                &~
                \left.
                \left[
                h^{\prime\prime}(\eta) \left(
                \Phi(\eta) - \Phi((\bx_i^k)^T\cB^*)
                \right)
                +
                h^{\prime}(\eta) \varphi(\eta)
                \right]
                \right\vert_{\eta=(\bx_i^k)^T\cB^{**}}
                \cdot
                \bx_{i,\cM}^T(\cB_{0,\cM} - \cB^*_{\cM})
                \\
                =
                &~
                O(\left[\bx_{i,\cM}^T(\cB_{0,\cM} - \cB^*_{\cM})\right]^2) + O(\bx_{i,\cM}^T(\cB_{0,\cM} - \cB^*_{\cM}))
            \end{align*}
            where the last equality is from the fact that $\vert h^{\prime\prime}(t)\vert \leq 1$, $\vert\Phi(t_1)-\Phi(t_2)\vert\leq\varphi(0)\vert t_1-t_2\vert$, and $h^{\prime}(t) \varphi(t) =\frac{\varphi^2}{\Phi(1-\Phi)}(t)\leq\frac{\varphi^2}{\Phi(1-\Phi)}(0)$. Along with \textcolor{black}{the 4th condition of (A1)}, we have $I_3/n=o(1/n)$. Therefore, $I_1 = O((m-r+s+K)/n)$, as desired.

            It remains to bound 
            \begin{equation*}
                \mathbb{P}(E_n^c)=\mathbb{P}\left(
                \left\Vert
                    \sum_{k=1}^{K}w_k \frac{1}{n}
                        \begin{pmatrix}
                            \bm{X}_{\cM}^T
                            \\
                            (\bm{X}^k_{S\cup \cK})^T
                        \end{pmatrix}
                        \left(
                        \hat{\Sigma}(\bm{X}^k\cB^*)
                        -
                        {\Sigma}(\bm{X}^k\cB^*)
                        \right)
                        \begin{pmatrix}
                            \bm{X}_{\cM}^T
                            \\
                            (\bm{X}^k_{S\cup \cK})^T
                        \end{pmatrix}^T
                \right\Vert_2
                \geq \frac{c}{8}
                \right).
            \end{equation*}
            Denote
            \begin{align*}
                {\mathcal{R}_n} = 1+\max_{1\leq i\leq n}\Vert \bx_{i,\cM\cup S}\Vert_2^2,
                ~
                T_n=\max_{1\leq k\leq K}
                \lambda_{\max}
                \left(
                    \begin{pmatrix}
                            \bm{X}_{\cM}^T
                            \\
                            (\bm{X}^k_{S\cup \cK})^T
                        \end{pmatrix}
                    \begin{pmatrix}
                            \bm{X}_{\cM}^T
                            \\
                            (\bm{X}^k_{S\cup \cK})^T
                        \end{pmatrix}^T
                \right)
            \end{align*}
            By the matrix Chernoff bound, we have
            \begin{align}
                &~
                \mathbb{P}\left(
                    \left\Vert
                \sum_{k=1}^{K}w_k 
                        \begin{pmatrix}
                            \bm{X}_{\cM}^T
                            \\
                            (\bm{X}^k_{S\cup \cK})^T
                        \end{pmatrix}
                        \left(
                        {\Sigma}(\bm{X}^k{\cB^*}) - \hat{\Sigma}(\bm{X}^k{\cB^*})
                        \right)
                        \begin{pmatrix}
                            \bm{X}_{\cM}^T
                            \\
                            (\bm{X}^k_{S\cup \cK})^T
                        \end{pmatrix}^T
                \right\Vert_2
                >
                t
                \right)
                \nonumber
                \\
                \leq 
                &~
                (s+m+K)
                \exp\left(
                -
                \frac{t^2}
                {{\mathcal{R}_n} T_n
                +
                {\mathcal{R}_n} t/3}
                \right)
                \nonumber
                \\
                =
                &~
                \exp\left(
                -
                \frac{t^2}
                {{\mathcal{R}_n} (T_n+t/3)}
                +\log(s+m+K)
                \right)
                \label{eq:step1-matrix_Chernoff}
            \end{align}
            From \textcolor{black}{the 1st and 4th conditions of (A1)}, we have ${\mathcal{R}_n}=O\left({\frac{n}{(s+m+K)\log n}}\right)$, $T_n=O(n)$. Then there exists constants $c_4$, $c_5$ such that ${\mathcal{R}_n}\leq  {\frac{c_4n}{(s+m+K)\log n}}$, $T_n\leq c_5 n$ for sufficiently large $n$. Thus
            \begin{align*}
                &~
                \mathbb{P}\left(
                \left\Vert
                    \sum_{k=1}^{K}w_k\frac{1}{n}
                        \begin{pmatrix}
                            \bm{X}_{\cM}^T
                            \\
                            (\bm{X}^k_{S\cup \cK})^T
                        \end{pmatrix}
                        \left(
                        \hat{\Sigma}(\bm{X}^k\cB^*)
                        -
                        {\Sigma}(\bm{X}^k\cB^*)
                        \right)
                        \begin{pmatrix}
                            \bm{X}_{\cM}^T
                            \\
                            (\bm{X}^k_{S\cup \cK})^T
                        \end{pmatrix}^T
                \right\Vert_2
                \geq \frac{c}{8}
                \right)
                \\
                \leq
                &~
                \exp\left(
                -
                \tilde{c}
                {(s+m+K)\log(n)}
                +\log(s+m+K)
                \right)
                \\
                \leq
                &~
                \exp\left(
                -
                \frac{\tilde{c}}{2}
                {(s+m+K)\log(n)}
                \right),
            \end{align*}
            for some constant $\tilde{c}$ and sufficiently large $n$.

        \paragraph{Step 2: Show that with probability tending to 1, the oracle local maximizer $\hat{\cB}$ is indeed a maximizer of $Q_n(\cB)$ with the linear constraint $\bC\cB_{\cM}=\bt$.}
            
            From step 1, we have for some constant $c$, and any $\tau\ll{\sqrt{\log n}}$,
            \begin{equation*}
                \mathbb{P}\left(
                    \Vert
                        \hat{\cB} - \cB^*
                    \Vert_2\leq\tau\sqrt{\frac{m-r+s+K}{n}}
                \right)
                \geq 
                1-c\frac{1}{\tau^2}-
                \exp\left(
                -\frac{\tilde{c}}{2}
                {(s+m+K)\log(n)}
                \right).
            \end{equation*}
            Similar to Theorem 1 in \cite{fan2011nonconcave}, it suffices to show that with probability tending to 1, $\hat{\cB}$ satisfies the following inequality:
            \begin{equation*}
                \left\Vert
                    \sum_{k=1}^{K}w_k 
                            (\bm{X}^k_{(\cM\cup S\cup\cK)^c})^T
                        \bm{H}(\bm{X}^k\hat{\cB})
                        \left\{
                            \bm{Y}^k - \bm{\mu}(\bm{X}^k\hat{\cB}) 
                        \right\}
                \right\Vert_{\infty}
                < 
                n\lambda_{n,0}{\rho}^{\prime}(0+).
            \end{equation*}
            By the Taylor's expansion, $\forall j\in(\cM\cup S\cup\cK)^c$,
            \begin{align}
                & \sum_{k=1}^{K}w_k 
                            \bm{x}_{(j)}^T
                        \bm{H}(\bm{X}^k\hat{\cB})
                        \left\{
                            \bm{Y}^k - \bm{\mu}(\bm{X}^k\hat{\cB})
                        \right\}
                                \nonumber
                \\
                = & 
                \sum_{k=1}^{K}w_k 
                            \bm{x}_{(j)}^T
                        \bm{H}(\bm{X}^k{\cB^*})
                        \left\{
                            \bm{Y}^k - \bm{\mu}(\bm{X}^k{\cB^*})
                        \right\}
                -
                \sum_{k=1}^{K}w_k 
                            \bm{x}_{(j)}^T
                        \hat{\Sigma}(\bm{X}^k{\cB^*})\bm{X}^k(\hat{\cB} - \cB^*) 
                + 
                R_j,
                \label{eq:step2-Taylor}
            \end{align}
            where 
            \begin{equation*}
                R_j = \sum_{k=1}^{K}w_k 
                    \hat{\bm{\Delta}}_{{\cM\cup S\cup\cK}}^T
                    \begin{pmatrix}
                            \bm{X}_{\cM}^T
                            \\
                            (\bm{X}^k_{S\cup \cK})^T
                        \end{pmatrix}
                    \text{diag}\left\{
                        \bx_{(j)}\circ\hat{\Sigma}^\prime(\bX^K\bar{\cB}_j)
                    \right\}
                    \begin{pmatrix}
                            \bm{X}_{\cM}^T
                            \\
                            (\bm{X}^k_{S\cup \cK})^T
                        \end{pmatrix}^T
                    \hat{\bm{\Delta}}_{{\cM\cup S\cup\cK}},
            \end{equation*}
            for some $\bar{\cB}_j$ lying on the line segment joining $\cB^*$ and $\hat{\cB}$,
            where $\hat{\bm{\Delta}}_{{\cM\cup S\cup\cK}}=\hat{\cB}_{\cM\cup S\cup\cK} - \cB^*_{\cM\cup S\cup\cK}$. Under the event 
            \begin{equation}
                \left\{\Vert
                        \hat{\cB} - \cB^*
                    \Vert_2\leq\tau\sqrt{{(m-r+s+K)}/{n}}
                \right\},
                \label{eq:step2-1}
            \end{equation}
            we have $\hat{\cB}\in\cN^*$. Along with $\cB^*\in\cN^*$, we have $\bar{\cB}_j\in\cN^*$. Then under the same event \eqref{eq:step2-1}, it follows from \textcolor{black}{the 5th condition of (A1)} that
            \begin{align*}
                \vert R_j\vert
                & \leq
                \sum_{k=1}^{K}w_k 
                    \hat{\bm{\Delta}}_{{\cM\cup S\cup\cK}}^T
                    \begin{pmatrix}
                            \bm{X}_{\cM}^T
                            \\
                            (\bm{X}^k_{S\cup \cK})^T
                        \end{pmatrix}
                    \text{diag}\left\{
                        \vert\bx_{(j)}\vert
                        \circ
                        \vert\hat{\Sigma}^\prime(\bX^K\bar{\cB}_j)\vert
                    \right\}
                    \begin{pmatrix}
                            \bm{X}_{\cM}^T
                            \\
                            (\bm{X}^k_{S\cup \cK})^T
                        \end{pmatrix}^T
                    \hat{\bm{\Delta}}_{{\cM\cup S\cup\cK}},
                \\
                & \leq
                n
                \lambda_{\max}
                \left(
                    \frac{1}{n}
                    \sum_{k=1}^{K}w_k 
                    \begin{pmatrix}
                            \bm{X}_{\cM}^T
                            \\
                            (\bm{X}^k_{S\cup \cK})^T
                        \end{pmatrix}
                    \text{diag}\left\{
                        \vert\bx_{(j)}\vert
                    \right\}
                    \begin{pmatrix}
                            \bm{X}_{\cM}^T
                            \\
                            (\bm{X}^k_{S\cup \cK})^T
                        \end{pmatrix}^T
                \right)
                \cdot
                \Vert
                    \hat{\bm{\Delta}}_{{\cM\cup S\cup\cK}}
                \Vert_2^2
                \\
                &\leq
                c_1n\cdot\tau^2{\frac{m-r+s+K}{n}}
                \\
                & =
                c_1\tau^2{(m-r+s+K)}.
            \end{align*}
            Therefore, we have
            \begin{equation}
                \sup_{j\in (\cM\cup S\cup\cK)^c}
                \vert R_j \vert
                \leq
                c_1\tau^2{(m-r+s+K)}.
                \label{eq:step2-2}
            \end{equation}

            For the second term of the RHS of inequality \eqref{eq:step2-Taylor}, 
            we can prove using \textcolor{black}{2nd and 4th conditions of (A1)} that
            \begin{equation*}
                \left\Vert
                    \frac{1}{n}
                    \sum_{k=1}^{K}w_k 
                            \bX_{(\cM\cup S\cup\cK)^c}^T
                        \hat{\Sigma}(\bm{X}^k{\cB^*})\bm{X}_{\cM\cup S\cup\cK}^k
                \right\Vert_{2,\infty}=O(1).
            \end{equation*}

            Therefore, under the same event \eqref{eq:step2-1}, we have
            \begin{align}
                &~
                \left\Vert
                    \sum_{k=1}^{K}w_k 
                            \bX_{(\cM\cup S\cup\cK)^c}^T
                        \hat{\Sigma}(\bm{X}^k{\cB^*})\bm{X}^k(\hat{\cB} - \cB^*) 
                \right\Vert_{\infty}
                \nonumber
                \\
                \leq &~
                \left\Vert
                    \sum_{k=1}^{K}w_k 
                            \bX_{(\cM\cup S\cup\cK)^c}^T
                        \hat{\Sigma}(\bm{X}^k{\cB^*})\bm{X}_{\cM\cup S\cup\cK}^k
                \right\Vert_{2,\infty}
                \cdot
                \left\Vert
                    \hat{\cB}_{\cM\cup S\cup\cK} - \cB^*_{\cM\cup S\cup\cK}
                \right\Vert_2
                \nonumber
                \\
                \leq &~
                c_2n\cdot\tau\sqrt{\frac{m-r+s+K}{n}}
                \nonumber
                \\
                = &~
                c_2\tau\sqrt{(m-r+s+K)n}.
                \label{eq:step2-3}
            \end{align}
            For the first term in the RHS of inequality \eqref{eq:step2-Taylor}, by Hoeffding's inequality, given some constant $\gamma>0$, $\forall j \in(\cM\cup S\cup\cK)^c$, 
            \begin{align*}
                \mathbb{P}\left(
                    \left\vert
                        \sum_{k=1}^{K}w_k 
                            \bm{x}_{(j)}^T
                        \bm{H}(\bm{X}^k{\cB^*})
                        \left\{
                            \bm{Y}^k - \bm{\mu}(\bm{X}^k{\cB^*})
                        \right\}
                    \right\vert
                    >
                    \gamma\sqrt{n\log p}
                \right)
                \\
                \leq 
                2\exp\left(
                -
                \frac
                    {2\gamma^2n\log p}{
                    \sum_{i=1}^{n}x_{ij}^2
                    [\sum_{k=1}^{K}w_k
                    h^{\prime}((\bx_{i}^{k})^T\cB^*)]^2
                    }
                \right).
            \end{align*}
            By \textcolor{black}{the 2nd and 3rd conditions of (A1)} and $\vert h^{\prime\prime}(\cdot)\vert\leq 1$, we have
            \begin{equation*}
                {\sum_{i=1}^{n}
                    \left[
                        \sum_{k=1}^{K}w_k
                        \vert
                            h^{\prime}((\bx_{i}^{k})^T\cB^*)\cdot x_{ij}
                        \vert
                    \right]^2}
                =O(n).
            \end{equation*}
            Therefore, combined with the previous inequality, we have
            \begin{align*}
                \mathbb{P}\left(
                    \left\vert
                        \sum_{k=1}^{K}w_k 
                            \bm{x}_{(j)}^T
                        \bm{H}(\bm{X}^k{\cB^*})
                        \left\{
                            \bm{Y}^k - \bm{\mu}(\bm{X}^k{\cB^*})
                        \right\}
                    \right\vert
                    >
                    \gamma\sqrt{n\log p}
                \right)
                \leq
                2\exp\left(
                -\frac{2\gamma^2\log p}
                {M}
                \right)
            \end{align*}
            for some constant $M$ and for $n$ sufficiently large. Thus, by Bonferroni's inequality, we have
            \begin{align*}
                &~
                \mathbb{P}\left(
                    \left\Vert
                        \sum_{k=1}^{K}w_k 
                            \bX_{(\cM\cup S\cup\cK)^c}^T
                        \bm{H}(\bm{X}^k{\cB^*})
                        \left\{
                            \bm{Y}^k - \bm{\mu}(\bm{X}^k{\cB^8})
                        \right\}
                    \right\Vert_\infty
                    >
                    \gamma\sqrt{n\log p}
                \right)
                \\
                \leq &~
                2\exp\left(
                -\frac{2\gamma^2\log p}{M}
                +
                \log(p-m-s)
                \right)
                \\
                \leq &~
                2\exp\left(
                -\frac{\gamma^2\log p}{M}
                \right),
            \end{align*}
            as long as $\gamma^2\geq M$.
            
            Combining with inequalities \eqref{eq:step2-2}, and \eqref{eq:step2-3}, we have under event \eqref{eq:step2-1} and 
            \begin{equation*}
                \left\{
                    \left\Vert
                        \sum_{k=1}^{K}w_k 
                            \bX_{(\cM\cup S\cup\cK)^c}^T
                        \bm{H}(\bm{X}^k{\cB^*})
                        \left\{
                            \bm{Y}^k - \bm{\mu}(\bm{X}^k{\cB^*})
                        \right\}
                    \right\Vert_\infty
                    \leq
                    \gamma\sqrt{n\log p}
                \right\},
            \end{equation*}
            we have
            \begin{align*}
                &~
                \left\Vert
                    \sum_{k=1}^{K}w_k 
                            (\bm{X}^k_{(\cM\cup S\cup\cK)^c})^T
                        \bm{H}(\bm{X}^k\hat{\cB})
                        \left\{
                            \bm{Y}^k - \bm{\mu}(\bm{X}^k\hat{\cB})
                        \right\}
                \right\Vert_{\infty}
                \\
                \leq
                &~
                \gamma\sqrt{n\log p} 
                +
                c_2\tau\sqrt{(m-r+s+K)n}
                +
                c_1\tau^2{(m-r+s+K)}
                \\
                < 
                &~
                n\lambda_{n,0}{\rho}^{\prime}(0+),
            \end{align*}
            as long as we have  $\gamma < \frac{{\rho}^{\prime}(0+)}{3}\lambda_{n,0}\sqrt{{n}/{\log p}}$, $\tau < \frac{{\rho}^{\prime}(0+)}{3c_2}\lambda_{n,0}\sqrt{{n}/{(m-r+s+K)}}$, and $$\tau^2 < \frac{{\rho}^{\prime}(0+)}{3c_1}{\frac{\lambda_{n,0}n}{(m-r+s+K)}}.$$ By \textcolor{black}{the 3rd condition of (A2)} and \textcolor{black}{$s+m
            +K=o(\sqrt{n})$}, we know that $\lambda_{n,0}\sqrt{{n}/{\log p}}$, $\lambda_{n,0}\sqrt{{n}/{(m-r+s+K)}}$, and $\lambda_{n,0}{{n}/{(m-r+s+K)}}$ could be arbitrarily large if $n$ is sufficiently large. Therefore, by choosing a diverging sequence of $\{\tau_n\}$ and a sufficiently large positive constant $\gamma$ such that all the previous inequality constraints on $\gamma$ and $\tau$ hold, we could have
            \begin{align*}
                & ~
                \mathbb{P}\left(
                    \left\Vert
                    \sum_{k=1}^{K}w_k 
                            (\bm{X}^k_{(\cM\cup S\cup\cK)^c})^T
                        \bm{H}(\bm{X}^k\hat{\cB})
                        \left\{
                            \bm{Y}^k - \bm{\mu}(\bm{X}^k\hat{\cB})
                        \right\}
                    \right\Vert_{\infty}
                    < 
                    n\lambda_{n,0}{\rho}^{\prime}(0+).
                \right)
                \\
                \geq 
                & ~
                1-c\frac{1}{\tau^2}
                \exp\left(
                -\frac{\tilde{c}}{2}
                {(s+m+K)\log(n)}
                \right)
                -
                2\exp\left(
                -\frac{\gamma^2\log p}{M}
                \right)
                \rightarrow 1,
            \end{align*}
            as desired.

        \paragraph{Step 3: Derive the asymptotic distribution of $\hat{\cB}_0$}
            
            We have shown in Step 1 and Step 2 that for any $\tau\ll\sqrt{\log n}$ and sufficiently large $n$,
            \begin{align}
                    \Vert
                        \hat{\cB}_{0,\cM\cup S\cup\cK} 
                        - 
                        \cB^*_{\cM\cup S\cup\cK} 
                    \Vert_2
                    =
                    O_p\left(
                        \sqrt{\frac{m-r+s+K}{n}}
                        \right),
                    \label{eq:step1-result1}
                \\
                \mathbb{P}\left(
                    \hat{\cB}_{0,(\cM\cup S\cup\cK)^c}=\bm{0}
                \right)
                \rightarrow 1,
                \label{eq:step1-result2}
                \\
                \bC\hat{\cB}_{0,\cM}=\bt.
                \label{eq:step1-result3}
            \end{align}
            Finally, we show that 
            \begin{align*}
                \sqrt{n}
                \begin{pmatrix}
                    \hat{\cB}_{0,\cM} - \cB^*_{\cM} \\
                    \hat{\cB}_{0,S\cup\cK} - \cB^*_{S\cup\cK}
                \end{pmatrix}
                &=
                \frac{1}{\sqrt{n}}
                \bm{K}_n^{-1/2}(\bm{I} - \bm{P}_n)\bm{K}_n^{-1/2}
                \sum_{k=1}^{K}w_k
                    \begin{pmatrix}
                            \bm{X}_{\cM}^T
                            \\
                            (\bm{X}^k_{S\cup \cK})^T
                        \end{pmatrix}
                        \bm{H}(\bm{X}^k\cB^*)
                        \left\{
                            \bm{Y}^k - \bm{\mu}(\bm{X}^k\cB^*)
                        \right\}
                \\
                & - 
                \sqrt{n}
                \bm{K}_n^{-1/2} \bm{P}_n\bm{K}_n^{-1/2}
                \begin{pmatrix}
                    \bC^T(\bC\Omega_{mm}\bC^T)^{-1}\bm{h}_n
                    \\
                    \bm{0}_{s+K}
                \end{pmatrix}
                +
                o_p(1).
            \end{align*}
            
            From the results of Step 1 and Step 2, we know that with probability tending to 1, $\hat{\cB}_0$ is the local maximizer of $Q_n(\cB)$ with the constraints $\bC\cB_{\cM}=\bt$, and $\cB_{(\cM\cup S\cup \cK)^c}=\bm{0}$. This implies that there exists some vector $\bm{v}\in\mathbb{R}^{r}$ such that
            \begin{equation}
                \sum_{k=1}^{K}w_k
                    \begin{pmatrix}
                            \bm{X}_{\cM}^T
                            \\
                            \bm{X}_{S}^T
                            \\
                            (\bm{X}^k_{\cK})^T
                        \end{pmatrix}
                        \bm{H}(\bm{X}^k{\hat{\cB}_0})
                        \left\{
                            \bm{Y}^k - \bm{\mu}(\bm{X}^k{\hat{\cB}_0})
                        \right\}
                =
                \begin{pmatrix}
                    \sqrt{n}\bC^T \bm{v}
                            \\
                            n\lambda_{n,0}\bar{\rho}({\hat{\cB}_{0,S}})
                            \\
                            \bm{0}_{K}
                \end{pmatrix}
                \label{eq:step3-1}
            \end{equation}
            By Taylor's expansion, the LHS of equation \eqref{eq:step3-1} is equal to
            \begin{align}
                &~
                \sum_{k=1}^{K}w_k
                    \begin{pmatrix}
                            \bm{X}_{\cM}^T
                            \\
                            (\bm{X}^k_{S\cup \cK})^T
                        \end{pmatrix}
                        \bm{H}(\bm{X}^k{{\cB^*}})
                        \left\{
                            \bm{Y}^k - \bm{\mu}(\bm{X}^k{{\cB^*}})
                        \right\}
                \nonumber
                \\
                -
                &~
                \sum_{k=1}^{K}w_k 
                        \begin{pmatrix}
                            \bm{X}_{\cM}^T
                            \\
                            (\bm{X}^k_{S\cup \cK})^T
                        \end{pmatrix}
                        {\Sigma}(\bm{X}^k{\cB^*})
                        \begin{pmatrix}
                            \bm{X}_{\cM}^T
                            \\
                            (\bm{X}^k_{S\cup \cK})^T
                        \end{pmatrix}^T
                        \begin{pmatrix}
                            \hat{\cB}_{0,\cM} - \cB^*_{\cM} \\
                            \hat{\cB}_{0,S\cup\cK} - \cB^*_{S\cup\cK}
                        \end{pmatrix}
                \nonumber
                \\
                + 
                & ~
                \sum_{k=1}^{K}w_k 
                        \begin{pmatrix}
                            \bm{X}_{\cM}^T
                            \\
                            (\bm{X}^k_{S\cup \cK})^T
                        \end{pmatrix}
                        \left(
                        {\Sigma}(\bm{X}^k{\cB^*}) - \hat{\Sigma}(\bm{X}^k{\cB^*})
                        \right)
                        \begin{pmatrix}
                            \bm{X}_{\cM}^T
                            \\
                            (\bm{X}^k_{S\cup \cK})^T
                        \end{pmatrix}^T
                        \begin{pmatrix}
                            \hat{\cB}_{0,\cM} - \cB^*_{\cM} \\
                            \hat{\cB}_{0,S\cup\cK} - \cB^*_{S\cup\cK}
                        \end{pmatrix}
                +
                \bm{R},
                \label{eq:step3-2}
            \end{align}
            where 
            \begin{equation*}
                R_j = \sum_{k=1}^{K}w_k 
                    \hat{\Delta}_{0, \cM\cup S\cup\cK}^T
                    \begin{pmatrix}
                            \bm{X}_{\cM}^T
                            \\
                            (\bm{X}^k_{S\cup \cK})^T
                        \end{pmatrix}
                    \text{diag}\left\{
                        \bx^k_{(j)}\circ\hat{\Sigma}^\prime(\bX^K\bar{\cB}_j)
                    \right\}
                    \begin{pmatrix}
                            \bm{X}_{\cM}^T
                            \\
                            (\bm{X}^k_{S\cup \cK})^T
                        \end{pmatrix}^T
                    \hat{\Delta}_{0, \cM\cup S\cup\cK},
            \end{equation*}
            for some $\bar{\cB}_j$ lying on the line segment joining $\cB^*$ and $\hat{\cB}_0$, and $\hat{\Delta}_{0, \cM\cup S\cup\cK} =\hat{\cB}_{0,\cM\cup S\cup\cK} - \cB^*_{\cM\cup S\cup\cK}$. For $j\in\cM\cup S$,
            \begin{align*}
                \vert R_j\vert
                & \leq
                \sum_{k=1}^{K}w_k 
                    \hat{\Delta}_{0, \cM\cup S\cup\cK}^T
                    \begin{pmatrix}
                            \bm{X}_{\cM}^T
                            \\
                            (\bm{X}^k_{S\cup \cK})^T
                        \end{pmatrix}
                    \text{diag}\left\{
                        \vert \bx_{(j)} \vert 
                        \circ
                        \vert \hat{\Sigma}^\prime(\bX^K\bar{\cB}_j)\vert
                    \right\}
                    \begin{pmatrix}
                            \bm{X}_{\cM}^T
                            \\
                            (\bm{X}^k_{S\cup \cK})^T
                        \end{pmatrix}^T
                    \hat{\Delta}_{0, \cM\cup S\cup\cK},
                \\
                &\leq
                n
                \lambda_{\max}
                \left(
                    \frac{1}{n}
                    \sum_{k=1}^{K}w_k 
                    \begin{pmatrix}
                            \bm{X}_{\cM}^T
                            \\
                            (\bm{X}^k_{S\cup \cK})^T
                        \end{pmatrix}
                    \text{diag}\left\{
                        \vert\bx_{(j)}\vert
                    \right\}
                    \begin{pmatrix}
                            \bm{X}_{\cM}^T
                            \\
                            (\bm{X}^k_{S\cup \cK})^T
                        \end{pmatrix}^T
                \right)
                \cdot
                \Vert
                    \hat{\cB}_{0,\cM\cup S\cup\cK} - \cB^*_{\cM\cup S\cup\cK}
                \Vert_2^2
                \\
                &=
                O(n)\cdot O_p\left(\frac{m-r+s+K}{n}\right),
            \end{align*}
            where the last equality is by \textcolor{black}{the 5th condition of (A1)} and result \eqref{eq:step1-result1} proved in step 1 and 2. Similarly, for $j\in\cK$, we have $\bx^{k}_{(j)}=-\mathbf{1}_{n}\mathbf{I}\{k=j\}$. By \textcolor{black}{the 4th condition of (A1)}, we have
            \begin{align*}
                \vert R_j\vert
                &  \leq 
                w_j
                \lambda_{\max}
                \left(
                    \begin{pmatrix}
                            \bm{X}_{\cM}^T
                            \\
                            (\bm{X}^j_{S\cup \cK})^T
                        \end{pmatrix}
                    \bm{I}_n
                    \begin{pmatrix}
                            \bm{X}_{\cM}^T
                            \\
                            (\bm{X}^j_{S\cup \cK})^T
                        \end{pmatrix}^T
                \right)
                \cdot
                \Vert
                    \hat{\cB}_{0,\cM\cup S\cup\cK} - \cB^*_{\cM\cup S\cup\cK}
                \Vert_2^2
                \\
                &=
                O(n)\cdot O_p\left(\frac{m-r+s+K}{n}\right).
            \end{align*}
            Therefore, we have
            \begin{equation}
                \Vert \bm{R}\Vert_2\leq \sqrt{s+m+K}\Vert\bm{R}\Vert_\infty=O_p\left(\sqrt{(s+m+K)^3}\right)=o_p(\sqrt{n}),
                \label{eq:step3-3}
            \end{equation}
            where the last equality is because we assume that \textcolor{black}{$s+m+K=o(n^{1/3})$}.
            
            Next we show that 
            \begin{equation}
                \left\Vert
                \sum_{k=1}^{K}w_k 
                        \begin{pmatrix}
                            \bm{X}_{\cM}^T
                            \\
                            (\bm{X}^k_{S\cup \cK})^T
                        \end{pmatrix}
                        \left(
                        {\Sigma}(\bm{X}^k{\cB^*}) - \hat{\Sigma}(\bm{X}^k{\cB^*})
                        \right)
                        \begin{pmatrix}
                            \bm{X}_{\cM}^T
                            \\
                            (\bm{X}^k_{S\cup \cK})^T
                        \end{pmatrix}^T
                        \begin{pmatrix}
                            \hat{\cB}_{0,\cM} - \cB^*_{\cM} \\
                            \hat{\cB}_{0,S\cup\cK} - \cB^*_{S\cup\cK}
                        \end{pmatrix}
                \right\Vert_2
                = o_p(\sqrt{n}).
                \label{eq:step3-4}
            \end{equation}
            Since $\Vert
                    \hat{\cB}_{0,\cM\cup S\cup\cK} - \cB^*_{\cM\cup S\cup\cK}
                \Vert_2 = O_p(\sqrt{(m-r+s+K)/n})$, it suffices to show that
            \begin{equation*}
                \left\Vert
                \sum_{k=1}^{K}w_k 
                        \begin{pmatrix}
                            \bm{X}_{\cM}^T
                            \\
                            (\bm{X}^k_{S\cup \cK})^T
                        \end{pmatrix}
                        \left(
                        {\Sigma}(\bm{X}^k{\cB^*}) - \hat{\Sigma}(\bm{X}^k{\cB^*})
                        \right)
                        \begin{pmatrix}
                            \bm{X}_{\cM}^T
                            \\
                            (\bm{X}^k_{S\cup \cK})^T
                        \end{pmatrix}^T
                \right\Vert_2
                = o_p\left(\frac{n}{\sqrt{s+m+K}}\right).
            \end{equation*}
            By the matrix Chernoff bound \eqref{eq:step1-matrix_Chernoff} in step 1, we have
            \begin{align*}
                &~
                \mathbb{P}\left(
                    \left\Vert
                \sum_{k=1}^{K}w_k 
                        \begin{pmatrix}
                            \bm{X}_{\cM}^T
                            \\
                            (\bm{X}^k_{S\cup \cK})^T
                        \end{pmatrix}
                        \left(
                        {\Sigma}(\bm{X}^k{\cB^*}) - \hat{\Sigma}(\bm{X}^k{\cB^*})
                        \right)
                        \begin{pmatrix}
                            \bm{X}_{\cM}^T
                            \\
                            (\bm{X}^k_{S\cup \cK})^T
                        \end{pmatrix}^T
                \right\Vert_2
                >
                \tau
                \frac{n}{\sqrt{s+m+K}}
                \cdot
                \sqrt{\frac{{\log(s+m+K)}}{{\log n}}}
                \right)
                \\
                \leq
                &~
                \exp\left(
                -
                \frac{
                    \tau^2
                    \frac{n^2}
                    {s+m+K}\frac{\log(s+m+K)}{{\log n}}
                     }
                {
                 {\frac{c_4n}{(s+m+K)\log n}}
                \left(
                    c_5n+ \frac{\tau}{3}\frac{n}{\sqrt{s+m+K}}
                \sqrt{\frac{{\log(s+m+K)}}{{\log n}}}
                \right)
                }
                +\log(s+m+K)
                \right)
                \\
                =
                &~
                \exp\left(
                -
                \frac{
                    \tau^2
                    \log(s+m+K)
                     }
                {
                c_4
                \left(
                    c_5+ \frac{\tau}{3}
                \sqrt{\frac{{\log(s+m+K)}}{{(s+m+K)\log n}}}
                \right)
                }
                +\log(s+m+K)
                \right)
                \\
                \leq
                &~
                \exp\left(
                -
                \frac{
                    \tau^2
                    \log(s+m+K)
                     }
                {
                c_4
                \left(
                    c_5+ \frac{\tau}{3}c_6
                \right)
                }
                +\log(s+m+K)
                \right),
            \end{align*}
            where the last inequality is because $\sqrt{\frac{{\log(s+m+K)}}{{(s+m+K)\log n}}}=o(1)$ and $c_6>0$ is some constant. The RHS of the above inequality could be arbitrarily small as long as $\tau$ is sufficiently large and $n$ is sufficiently large. Therefore, we proved that 
            \begin{align}
                \left\Vert
                \sum_{k=1}^{K}w_k 
                        \begin{pmatrix}
                            \bm{X}_{\cM}^T
                            \\
                            (\bm{X}^k_{S\cup \cK})^T
                        \end{pmatrix}
                        \left(
                        {\Sigma}(\bm{X}^k{\cB^*}) - \hat{\Sigma}(\bm{X}^k{\cB^*})
                        \right)
                        \begin{pmatrix}
                            \bm{X}_{\cM}^T
                            \\
                            (\bm{X}^k_{S\cup \cK})^T
                        \end{pmatrix}^T
                \right\Vert_2
                & =
                O_p\left(
                    \frac{n}{\sqrt{s+m+K}}
                    \cdot
                    \sqrt{\frac{{\log(s+m+K)}}{{\log n}}}
                \right)
                \nonumber
                \\
                &=
                o_p\left(
                    \frac{n}{\sqrt{s+m+K}}
                \right),
                \label{eq:step3-4.5}
            \end{align}
            as desired.
            
            Combining equations \eqref{eq:step3-2}, \eqref{eq:step3-3}, and \eqref{eq:step3-4}, we have
            \begin{align}
                & ~ 
                \sum_{k=1}^{K}w_k
                    \begin{pmatrix}
                            \bm{X}_{\cM}^T
                            \\
                            (\bm{X}^k_{S\cup \cK})^T
                        \end{pmatrix}
                        \bm{H}(\bm{X}^k{\hat{\cB}_0})
                        \left\{
                            \bm{Y}^k - \bm{\mu}(\bm{X}^k{\hat{\cB}_0})
                        \right\}
                \nonumber
                \\
                =
                &~
                \sum_{k=1}^{K}w_k
                    \begin{pmatrix}
                            \bm{X}_{\cM}^T
                            \\
                            (\bm{X}^k_{S\cup \cK})^T
                        \end{pmatrix}
                        \bm{H}(\bm{X}^k{{\cB^*}})
                        \left\{
                            \bm{Y}^k - \bm{\mu}(\bm{X}^k{{\cB^*}})
                        \right\}
                \nonumber
                \\
                -
                &~
                \sum_{k=1}^{K}w_k 
                        \begin{pmatrix}
                            \bm{X}_{\cM}^T
                            \\
                            (\bm{X}^k_{S\cup \cK})^T
                        \end{pmatrix}
                        {\Sigma}(\bm{X}^k{\cB^*})
                        \begin{pmatrix}
                            \bm{X}_{\cM}^T
                            \\
                            (\bm{X}^k_{S\cup \cK})^T
                        \end{pmatrix}^T
                        \begin{pmatrix}
                            \hat{\cB}_{0,\cM} - \cB^*_{\cM} \\
                            \hat{\cB}_{0,S\cup\cK} - \cB^*_{S\cup\cK}
                        \end{pmatrix}
                + o_p(\sqrt{n}).
                \label{eq:step3-5}
            \end{align}
            Combining equations \eqref{eq:step3-1} and \eqref{eq:step3-5}, we have
            \begin{align}
                \begin{pmatrix}
                    \sqrt{n}\bC^T \bm{v}
                            \\
                            n\lambda_{n,0}\bar{\rho}({\hat{\cB}_{0,S}})
                            \\
                            \bm{0}_{K}
                \end{pmatrix}
                =
                &~
                \sum_{k=1}^{K}w_k
                    \begin{pmatrix}
                            \bm{X}_{\cM}^T
                            \\
                            (\bm{X}^k_{S\cup \cK})^T
                        \end{pmatrix}
                        \bm{H}(\bm{X}^k{{\cB^*}})
                        \left\{
                            \bm{Y}^k - \bm{\mu}(\bm{X}^k{{\cB^*}})
                        \right\}
                \nonumber
                \\
                -
                &~
                \sum_{k=1}^{K}w_k 
                        \begin{pmatrix}
                            \bm{X}_{\cM}^T
                            \\
                            (\bm{X}^k_{S\cup \cK})^T
                        \end{pmatrix}
                        {\Sigma}(\bm{X}^k{\cB^*})
                        \begin{pmatrix}
                            \bm{X}_{\cM}^T
                            \\
                            (\bm{X}^k_{S\cup \cK})^T
                        \end{pmatrix}^T
                        \begin{pmatrix}
                            \hat{\cB}_{0,\cM} - \cB^*_{\cM} \\
                            \hat{\cB}_{0,S\cup\cK} - \cB^*_{S\cup\cK}
                        \end{pmatrix}
                + o_p(\sqrt{n}).
                \label{eq:step3-6}
            \end{align}
            Define that 
            \begin{equation*}
                \bm{K}_n = 
                \frac{1}{n}
                \sum_{k=1}^{K}w_k 
                        \begin{pmatrix}
                            \bm{X}_{\cM}^T
                            \\
                            (\bm{X}^k_{S\cup \cK})^T
                        \end{pmatrix}
                        {\Sigma}(\bm{X}^k{\cB^*})
                        \begin{pmatrix}
                            \bm{X}_{\cM}^T
                            \\
                            (\bm{X}^k_{S\cup \cK})^T
                        \end{pmatrix}^T
            \end{equation*}
            Under \textcolor{black}{the 6th condition of (A1)}, we have $\lim\inf_n\lambda_{\min}(\bm{K}_n)>0$. This implies $\Vert \bm{K}_n^{-1}\bm{R}\Vert_2=o_p(\sqrt{n})$. Combining this together with equation \eqref{eq:step3-6}, we have
            \begin{align}
                \sqrt{n}
                \begin{pmatrix}
                            \hat{\cB}_{0,\cM} - \cB^*_{\cM} \\
                            \hat{\cB}_{0,S\cup\cK} - \cB^*_{S\cup\cK}
                        \end{pmatrix}
                &=
                \frac{1}{\sqrt{n}}
                \bm{K}_n^{-1}
                \sum_{k=1}^{K}w_k
                    \begin{pmatrix}
                            \bm{X}_{\cM}^T
                            \\
                            (\bm{X}^k_{S\cup \cK})^T
                        \end{pmatrix}
                        \bm{H}(\bm{X}^k{{\cB^*}})
                        \left\{
                            \bm{Y}^k - \bm{\mu}(\bm{X}^k{{\cB^*}})
                        \right\}
                \nonumber
                \\
                &
                - \bm{K}_n^{-1}
                \begin{pmatrix}
                    \bC^T \bm{v}
                            \\
                            \sqrt{n}
                            \lambda_{n,0}\bar{\rho}({\hat{\cB}_{0,S}})
                            \\
                            \bm{0}_{K}
                \end{pmatrix}
                +
                o_p(1)
                \label{eq:step3-7}
            \end{align}
            Since $\Vert
                        \hat{\cB}_{0,\cM\cup S\cup\cK} 
                        - 
                        \cB^*_{\cM\cup S\cup\cK} 
                    \Vert_2
                    =
                    O_p\left(
                        \sqrt{\frac{m-r+s+K}{n}}
                        \right)$ 
                        and $d_n \gg \sqrt{(s+m+K)/n}$ by \textcolor{black}{the 3rd condition of (A2)}, 
                        we have $
                        \Vert
                        \hat{\cB}_{0,\cM\cup S\cup\cK} 
                        - 
                        \cB^*_{\cM\cup S\cup\cK} 
                        \Vert_{\infty}
                        \leq 
                        \Vert
                        \hat{\cB}_{0,\cM\cup S\cup\cK} 
                        - 
                        \cB^*_{\cM\cup S\cup\cK} 
                        \Vert_{2}
                        = o_p(d_n)$.
            Thus, with probability tending to $1$,
            \begin{equation*}
                \min_{j\in{S}}\vert\hat{\cB}_{0,j}\vert 
                    \geq
                    \min_{j\in{S}}\vert{\cB^*_{j}}\vert
                    - 
                    \Vert
                        \hat{\cB}_{0,\cM\cup S\cup\cK} 
                        - 
                        \cB^*_{\cM\cup S\cup\cK} 
                        \Vert_{\infty}
                    >
                    d_n.
            \end{equation*}
            By the monotonicity of $\rho^{\prime}$ and \textcolor{black}{the 2nd condition of (A2)}, we obtain
            \begin{equation*}
                \Vert 
                    \sqrt{n}
                            \lambda_{n,0}\bar{\rho}({\hat{\cB}_{0,S}})
                \Vert_2
                \leq
                \sqrt{s}
                \sqrt{n}\lambda_{n,0}
                \vert \rho^{\prime}(d_n)\vert 
                =
                o(1),
            \end{equation*}
            with probability tending to 1. This together with equation \eqref{eq:step3-7} suggests that
            \begin{align}
                \sqrt{n}
                \begin{pmatrix}
                            \hat{\cB}_{0,\cM} - \cB^*_{\cM} \\
                            \hat{\cB}_{0,S\cup\cK} - \cB^*_{S\cup\cK}
                        \end{pmatrix}
                &=
                \frac{1}{\sqrt{n}}
                \bm{K}_n^{-1}
                \sum_{k=1}^{K}w_k
                    \begin{pmatrix}
                            \bm{X}_{\cM}^T
                            \\
                            (\bm{X}^k_{S\cup \cK})^T
                        \end{pmatrix}
                        \bm{H}(\bm{X}^k{{\cB^*}})
                        \left\{
                            \bm{Y}^k - \bm{\mu}(\bm{X}^k{{\cB^*}})
                        \right\}
                \nonumber
                \\
                &-
                \bm{K}_n^{-1}
                \begin{pmatrix}
                    \bC^T 
                            \\
                            \bm{0}_{r\times (s+K)}^T
                \end{pmatrix}
                \bm{v}
                +
                o_p(1).
                \label{eq:step3-8}
            \end{align}
            
            Since $\bm{C}\cB^*_{\cM}-\bm{t}=\bm{h}_n$, we have $\bm{C}(\hat{\cB}_{0,\cM} - \cB^*_{\cM})=-\bm{h}_n$. Therefore, multiplying both sides of equation \eqref{eq:step3-8} by 
            $
            \begin{pmatrix}
                    \bC 
                            &
                            \bm{0}_{r\times (s+K)}
                \end{pmatrix}
            $,
            we have
            \begin{align}
                &~
                \begin{pmatrix}
                    \bC^T 
                            \\
                            \bm{0}_{r\times (s+K)}^T
                \end{pmatrix}^T
                \bm{K}_n^{-1}
                \begin{pmatrix}
                    \bC^T 
                            \\
                            \bm{0}_{r\times (s+K)}^T
                \end{pmatrix}
                \bm{v}
                -
                \sqrt{n}\bm{h}_n
                \nonumber
                \\
                =
                &~
                \frac{1}{\sqrt{n}}
                \begin{pmatrix}
                    \bC^T 
                            \\
                            \bm{0}_{r\times (s+K)}^T
                \end{pmatrix}^T
                \bm{K}_n^{-1}
                \sum_{k=1}^{K}w_k
                    \begin{pmatrix}
                            \bm{X}_{\cM}^T
                            \\
                            (\bm{X}^k_{S\cup \cK})^T
                        \end{pmatrix}
                        \bm{H}(\bm{X}^k{{\cB^*}})
                        \left\{
                            \bm{Y}^k - \bm{\mu}(\bm{X}^k{{\cB^*}})
                        \right\}
                +
                \bm{C}\bm{R}^*
                \label{eq:step3-9}
            \end{align}
            for some $m$-dimensional vector $\bm{R}^*$ such that $\Vert\bm{R}^*\Vert_2=o_p(1)$. Define $\Omega_{n}=\bm{K}_n^{-1}$ and $\Omega_{mm}$ is the $m\times m$ top-left block of $\Omega_{n}$. By multiplying ${(\bm{C}\Omega_{mm}\bm{C}^T)^{-1}}$ on both sides of equation \eqref{eq:step3-9}, we have
            \begin{align*}
                \bm{v}
                & =
                \frac{1}{\sqrt{n}}
                {(\bm{C}\Omega_{mm}\bm{C}^T)^{-1}}
                \begin{pmatrix}
                    \bC^T 
                            \\
                            \bm{0}_{r\times (s+K)}^T
                \end{pmatrix}^T
                \bm{K}_n^{-1}
                \sum_{k=1}^{K}w_k
                    \begin{pmatrix}
                            \bm{X}_{\cM}^T
                            \\
                            (\bm{X}^k_{S\cup \cK})^T
                        \end{pmatrix}
                        \bm{H}(\bm{X}^k{{\cB^*}})
                        \left\{
                            \bm{Y}^k - \bm{\mu}(\bm{X}^k{{\cB^*}})
                        \right\}
                \\
                & +
                {(\bm{C}\Omega_{mm}\bm{C}^T)^{-1}}
                \bm{C}\bm{R}^*
                +
                \sqrt{n}
                {(\bm{C}\Omega_{mm}\bm{C}^T)^{-1}}
                \bm{h}_n.
            \end{align*}
            This together with equation \eqref{eq:step3-8} yields
            \begin{align}
                &~
                \sqrt{n}
                \begin{pmatrix}
                            \hat{\cB}_{0,\cM} - \cB^*_{\cM} \\
                            \hat{\cB}_{0,S\cup\cK} - \cB^*_{S\cup\cK}
                        \end{pmatrix}
                =
                \frac{1}{\sqrt{n}}
                \bm{K}_n^{-1}
                \sum_{k=1}^{K}w_k
                    \begin{pmatrix}
                            \bm{X}_{\cM}^T
                            \\
                            (\bm{X}^k_{S\cup \cK})^T
                        \end{pmatrix}
                        \bm{H}(\bm{X}^k{{\cB^*}})
                        \left\{
                            \bm{Y}^k - \bm{\mu}(\bm{X}^k{{\cB^*}})
                        \right\}
                \nonumber
                \\
                -
                & ~
                \frac{1}{\sqrt{n}}
                \bm{K}_n^{-1}
                \begin{pmatrix}
                    \bC^T 
                            \\
                            \bm{0}_{r\times (s+K)}^T
                \end{pmatrix}
                {(\bm{C}\Omega_{mm}\bm{C}^T)^{-1}}
                \begin{pmatrix}
                    \bC^T 
                            \\
                            \bm{0}_{r\times (s+K)}^T
                \end{pmatrix}^T
                \bm{K}_n^{-1}
                \sum_{k=1}^{K}w_k
                    \begin{pmatrix}
                            \bm{X}_{\cM}^T
                            \\
                            (\bm{X}^k_{S\cup \cK})^T
                        \end{pmatrix}
                        \bm{H}(\bm{X}^k{{\cB^*}})
                        \left\{
                            \bm{Y}^k - \bm{\mu}(\bm{X}^k{{\cB^*}})
                        \right\}
                \nonumber
                \\
                -
                &~
                \sqrt{n}
                \bm{K}_n^{-1}
                \begin{pmatrix}
                    \bC^T 
                            \\
                            \bm{0}_{r\times (s+K)}^T
                \end{pmatrix}
                {(\bm{C}\Omega_{mm}\bm{C}^T)^{-1}}
                \bm{h}_n
                -
                \bm{K}_n^{-1}
                \begin{pmatrix}
                    \bC^T 
                            \\
                            \bm{0}_{r\times (s+K)}^T
                \end{pmatrix}
                {(\bm{C}\Omega_{mm}\bm{C}^T)^{-1}}
                \bm{C}\bm{R}^*
                +
                o_p(1)
                \nonumber
                \\
                =
                &~
                \frac{1}{\sqrt{n}}
                \bm{K}_n^{-1/2}
                (\bm{I}-\bm{P}_n)
                \bm{K}_n^{-1/2}
                \sum_{k=1}^{K}w_k
                    \begin{pmatrix}
                            \bm{X}_{\cM}^T
                            \\
                            (\bm{X}^k_{S\cup \cK})^T
                        \end{pmatrix}
                        \bm{H}(\bm{X}^k{{\cB^*}})
                        \left\{
                            \bm{Y}^k - \bm{\mu}(\bm{X}^k{{\cB^*}})
                        \right\}
                \nonumber
                \\
                -
                &~
                \sqrt{n}
                \bm{K}_n^{-1}
                \begin{pmatrix}
                    \bC^T 
                            \\
                            \bm{0}_{r\times (s+K)}^T
                \end{pmatrix}
                {(\bm{C}\Omega_{mm}\bm{C}^T)^{-1}}
                \bm{h}_n
                -
                \bm{K}_n^{-1}
                \begin{pmatrix}
                    \bC^T 
                            \\
                            \bm{0}_{r\times (s+K)}^T
                \end{pmatrix}
                {(\bm{C}\Omega_{mm}\bm{C}^T)^{-1}}
                \bm{C}\bm{R}^*
                +
                o_p(1),
                \label{eq:step3-10}
            \end{align}
            where $\bm{P}_n$ is a $(m+s+K)\times (m+s+K)$ projection matrix of rank $r$.
            \begin{equation*}
                \bm{P}_n 
                = 
                \bm{K}_n^{-1/2}
                \begin{pmatrix}
                    \bC^T 
                            \\
                            \bm{0}_{r\times (s+K)}^T
                \end{pmatrix}
                \left(
                \begin{pmatrix}
                    \bC^T 
                            \\
                            \bm{0}_{r\times (s+K)}^T
                \end{pmatrix}^T
                \bm{K}_n^{-1}
                \begin{pmatrix}
                    \bC^T 
                            \\
                            \bm{0}_{r\times (s+K)}^T
                \end{pmatrix}
                \right)^{-1}
                \begin{pmatrix}
                    \bC^T 
                            \\
                            \bm{0}_{r\times (s+K)}^T
                \end{pmatrix}^T
                \bm{K}_n^{-1/2}.
            \end{equation*}
            
            Next, we show that
            \begin{equation}
                \bm{K}_n^{-1}
                    \begin{pmatrix}
                        \bC^T 
                            \\
                            \bm{0}_{r\times (s+K)}^T
                    \end{pmatrix}
                    {(\bm{C}\Omega_{mm}\bm{C}^T)^{-1}}
                    \bm{C}\bm{R}^*
                =
                o_p(1)
                \label{eq:step3-11}
            \end{equation}
            Under \textcolor{black}{the 4th condition in (A1)}, we have $\lambda_{\max}(\bm{K}_n)=O(1)$. This implies that $\lim\inf_n\lambda_{\min}(\bm{K}_n^{-1})>0$. Thus $\lim\inf_n\lambda_{\min}(\Omega_{mm})>0$. Similarly, by \textcolor{black}{the 6th condition in (A1)}, we have $\lim\inf_n\lambda_{\min}(\bm{K}_n)>0$, which implies that $\lambda_{\max}(\Omega_{mm})=O(1)$. Therefore, using Cauchy-Schwartz inequality, we have
            \begin{align*}
                &~
                \left\Vert
                    \bm{K}_n^{-1}
                    \begin{pmatrix}
                        \bC^T 
                            \\
                            \bm{0}_{r\times (s+K)}^T
                    \end{pmatrix}
                    {(\bm{C}\Omega_{mm}\bm{C}^T)^{-1}}
                    \bm{C}\bm{R}^*
                \right\Vert_2^2
                \\
                \leq
                &~
                \Vert
                \bm{K}_n^{-1/2}
                \Vert_2^2
                \cdot
                \left\Vert
                    \bm{K}_n^{-1/2}
                    \begin{pmatrix}
                        \bC^T 
                            \\
                            \bm{0}_{r\times (s+K)}^T
                    \end{pmatrix}
                    {(\bm{C}\Omega_{mm}\bm{C}^T)^{-1}}
                    \bm{C}
                \right\Vert_2^2
                \cdot
                \Vert
                    \bm{R}^*
                \Vert_2^2
                \\
                =
                &~
                \lambda_{\max}
                (
                \bm{K}_n^{-1}
                )
                \cdot
                \left\Vert
                    \bm{C}^T
                    {(\bm{C}\Omega_{mm}\bm{C}^T)^{-1}}
                    \begin{pmatrix}
                        \bC^T 
                            \\
                            \bm{0}_{r\times (s+K)}^T
                    \end{pmatrix}^T
                    \bm{K}_n^{-1}
                    \begin{pmatrix}
                        \bC^T 
                            \\
                            \bm{0}_{r\times (s+K)}^T
                    \end{pmatrix}
                    {(\bm{C}\Omega_{mm}\bm{C}^T)^{-1}}
                    \bm{C}
                \right\Vert_2
                \cdot
                \Vert
                    \bm{R}^*
                \Vert_2^2
                \\
                =
                &~
                \lambda_{\max}
                (
                \bm{K}_n^{-1}
                )
                \cdot
                \left\Vert
                    \bm{C}^T
                    {(\bm{C}\Omega_{mm}\bm{C}^T)^{-1}}
                    \bm{C}
                \right\Vert_2
                \cdot
                \Vert
                    \bm{R}^*
                \Vert_2^2
                \\
                =
                &~
                \lambda_{\max}
                (
                \bm{K}_n^{-1}
                )
                \cdot
                \left\Vert
                    {(\bm{C}\Omega_{mm}\bm{C}^T)^{-1/2}}
                    \bm{C}
                    \Omega_{mm}^{1/2}
                    \cdot
                    \Omega_{mm}^{-1/2}
                \right\Vert_2^2
                \cdot
                \Vert
                    \bm{R}^*
                \Vert_2^2
                \\
                \leq
                &~
                \lambda_{\max}
                (
                \bm{K}_n^{-1}
                )
                \cdot
                \left\Vert
                    {(\bm{C}\Omega_{mm}\bm{C}^T)^{-1/2}}
                    \bm{C}
                    \Omega_{mm}^{1/2}
                \right\Vert_2^2
                \cdot
                \Vert
                    \Omega_{mm}^{-1/2}
                \Vert_2^2
                \cdot
                \Vert
                    \bm{R}^*
                \Vert_2^2
                \\
                =
                &~
                \lambda_{\max}
                (
                \bm{K}_n^{-1}
                )
                \cdot
                \Vert
                \bm{I}_r
                \Vert_2
                \cdot
                \lambda_{\max}
                (
                \Omega_{mm}^{-1}
                )
                \cdot
                \Vert
                    \bm{R}^*
                \Vert_2^2=o_p(1).
            \end{align*}
            This proves equation \eqref{eq:step3-11}. Hence, it follows from equation \eqref{eq:step3-10} that
            \begin{align}
                \sqrt{n}
                \begin{pmatrix}
                            \hat{\cB}_{0,\cM} - \cB^*_{\cM} \\
                            \hat{\cB}_{0,S\cup\cK} - \cB^*_{S\cup\cK}
                        \end{pmatrix}
                & = 
                \frac{1}{\sqrt{n}}
                \bm{K}_n^{-1/2}
                (\bm{I}-\bm{P}_n)
                \bm{K}_n^{-1/2}
                \sum_{k=1}^{K}w_k
                    \begin{pmatrix}
                            \bm{X}_{\cM}^T
                            \\
                            (\bm{X}^k_{S\cup \cK})^T
                        \end{pmatrix}
                        \bm{H}(\bm{X}^k{{\cB^*}})
                        \left\{
                            \bm{Y}^k - \bm{\mu}(\bm{X}^k{{\cB^*}})
                        \right\}
                \nonumber
                \\
                & -
                \sqrt{n}
                \bm{K}_n^{-1}
                \begin{pmatrix}
                    \bC^T 
                            \\
                            \bm{0}_{r\times (s+K)}^T
                \end{pmatrix}
                {(\bm{C}\Omega_{mm}\bm{C}^T)^{-1}}
                \bm{h}_n
                +
                o_p(1).
                \label{eq:step3-12}
            \end{align}
            Moreover, observe the following equality
            \begin{align*}
                \bm{h}_n
                &
                =
                {\bm{C}\Omega_{mm}\bm{C}^T}
                ({\bm{C}\Omega_{mm}\bm{C}^T})^{-1}
                \bm{h}_n
                \\
                &
                =
                \begin{pmatrix}
                    \bC^T 
                            \\
                            \bm{0}_{r\times (s+K)}^T
                \end{pmatrix}^T
                \bm{K}_n^{-1}
                \begin{pmatrix}
                    \bC^T 
                            \\
                            \bm{0}_{r\times (s+K)}^T
                \end{pmatrix}
                ({\bm{C}\Omega_{mm}\bm{C}^T})^{-1}
                \bm{h}_n
                \\
                &
                =
                \begin{pmatrix}
                    \bC^T 
                            \\
                            \bm{0}_{r\times (s+K)}^T
                \end{pmatrix}^T
                \bm{K}_n^{-1}
                \begin{pmatrix}
                    \bm{C}^T({\bm{C}\Omega_{mm}\bm{C}^T})^{-1}\bm{h}_n
                    \\
                    \bm{0}_{s+K}
                \end{pmatrix}.
            \end{align*}
            Hence, we have
            \begin{align}
                &~
                \bm{K}_n^{-1}
                \begin{pmatrix}
                    \bC^T 
                            \\
                            \bm{0}_{r\times (s+K)}^T
                \end{pmatrix}
                {(\bm{C}\Omega_{mm}\bm{C}^T)^{-1}}
                \bm{h}_n
                \nonumber
                \\
                &
                =
                \bm{K}_n^{-1}
                \begin{pmatrix}
                    \bC^T 
                            \\
                            \bm{0}_{r\times (s+K)}^T
                \end{pmatrix}
                {(\bm{C}\Omega_{mm}\bm{C}^T)^{-1}}
                \begin{pmatrix}
                    \bC^T 
                            \\
                            \bm{0}_{r\times (s+K)}^T
                \end{pmatrix}^T
                \bm{K}_n^{-1}
                \begin{pmatrix}
                    \bm{C}^T({\bm{C}\Omega_{mm}\bm{C}^T})^{-1}\bm{h}_n
                    \\
                    \bm{0}_{s+K}
                \end{pmatrix}
                \nonumber
                \\
                &
                =
                \bm{K}_n^{-1/2}
                \bm{P}_n
                \bm{K}_n^{-1/2}
                \begin{pmatrix}
                    \bm{C}^T({\bm{C}\Omega_{mm}\bm{C}^T})^{-1}\bm{h}_n
                    \\
                    \bm{0}_{s+K}
                \end{pmatrix}
                \label{eq:step3-13}
            \end{align}
            This together with equation \eqref{eq:step3-13} implies that
            \begin{align}
                \sqrt{n}
                \begin{pmatrix}
                            \hat{\cB}_{0,\cM} - \cB^*_{\cM} \\
                            \hat{\cB}_{0,S\cup\cK} - \cB^*_{S\cup\cK}
                        \end{pmatrix}
                & = 
                \frac{1}{\sqrt{n}}
                \bm{K}_n^{-1/2}
                (\bm{I}-\bm{P}_n)
                \bm{K}_n^{-1/2}
                \sum_{k=1}^{K}w_k
                    \begin{pmatrix}
                            \bm{X}_{\cM}^T
                            \\
                            (\bm{X}^k_{S\cup \cK})^T
                        \end{pmatrix}
                        \bm{H}(\bm{X}^k{{\cB^*}})
                        \left\{
                            \bm{Y}^k - \bm{\mu}(\bm{X}^k{{\cB^*}})
                        \right\}
                \nonumber
                \\
                & -
                \sqrt{n}
                \bm{K}_n^{-1/2}
                \bm{P}_n
                \bm{K}_n^{-1/2}
                \begin{pmatrix}
                    \bm{C}^T({\bm{C}\Omega_{mm}\bm{C}^T})^{-1}\bm{h}_n
                    \\
                    \bm{0}_{s+K}
                \end{pmatrix}
                +
                o_p(1).
                \label{eq:step3-14}
            \end{align}
    Similarly, one can prove that
    \begin{equation}
        \sqrt{n}
                \begin{pmatrix}
                            \hat{\cB}_{a,\cM} - \cB^*_{\cM} \\
                            \hat{\cB}_{a,S\cup\cK} - \cB^*_{S\cup\cK}
                        \end{pmatrix}
        =
        \frac{1}{\sqrt{n}}
                \bm{K}_n^{-1}
                \sum_{k=1}^{K}w_k
                    \begin{pmatrix}
                            \bm{X}_{\cM}^T
                            \\
                            (\bm{X}^k_{S\cup \cK})^T
                        \end{pmatrix}
                        \bm{H}(\bm{X}^k{{\cB^*}})
                        \left\{
                            \bm{Y}^k - \bm{\mu}(\bm{X}^k{{\cB^8}})
                        \right\}
         +
                o_p(1).
        \label{eq:step3-15}
    \end{equation}
\end{proof}

\subsection{Proof of Theorem \ref{theorem:testing_statistic_distribution}}
\begin{proof}
    Define 
    \begin{equation*}
        T_0=
        ({\bm{\omega}_n + \sqrt{n}\bm{h}_n})^T
        ({\bm{C}\Omega_{mm}\bm{C}^T})^{-1}
        ({\bm{\omega}_n + \sqrt{n}\bm{h}_n})
    \end{equation*}
    and 
    \begin{equation*}
        \bm{\omega}_n
        =
        \frac{1}{\sqrt{n}}
                \begin{pmatrix}
                    \bC^T 
                            \\
                            \bm{0}_{r\times (s+K)}^T
                \end{pmatrix}^T
                \bm{K}_n^{-1}
                \sum_{k=1}^{K}w_k
                    \begin{pmatrix}
                            \bm{X}_{\cM}^T
                            \\
                            (\bm{X}^k_{S\cup \cK})^T
                        \end{pmatrix}
                        \bm{H}(\bm{X}^k{{\cB^*}})
                        \left\{
                            \bm{Y}^k - \bm{\mu}(\bm{X}^k{{\cB^*}})
                        \right\}.
    \end{equation*}
    \paragraph{Step 1: Show that $T_W/r$ is equivalent to $T_0/r$.}
    
        By Theorem 2.1,
        \begin{equation*}
        \sqrt{n}
                \begin{pmatrix}
                            \hat{\cB}_{a,\cM} - \cB^*_{\cM} \\
                            \hat{\cB}_{a,S\cup\cK} - \cB^*_{S\cup\cK}
                        \end{pmatrix}
        =
        \frac{1}{\sqrt{n}}
                \bm{K}_n^{-1}
                \sum_{k=1}^{K}w_k
                    \begin{pmatrix}
                            \bm{X}_{\cM}^T
                            \\
                            (\bm{X}^k_{S\cup \cK})^T
                        \end{pmatrix}
                        \bm{H}(\bm{X}^k{{\cB^*}})
                        \left\{
                            \bm{Y}^k - \bm{\mu}(\bm{X}^k{{\cB^*}})
                        \right\}
         +
                \bm{R}_a,
        \end{equation*}
        where $\bm{R}_a$ is a $(m+s+K)$-dimensional vector satisfying $\Vert\bm{R}_a\Vert_2=o_p(1)$. Therefore, by multiplying both sides of the above equation by $\begin{pmatrix} \bC  & \bm{0}_{r\times (s+K)}  \end{pmatrix}$, we have
        \begin{equation}
            \sqrt{n}\bm{C}
            (\hat{\cB}_{a,\cM} - \cB^*_{\cM})
            =
            \bm{\omega}_n
            +
            \bm{C}\bm{R}_{a,J_0},
            \label{eq:*step1-1}
        \end{equation}
        where $J_0=[1,\dots,m]$. Since $\bm{C}\cB^*_{\cM}=\bm{t}+\bm{h}_n$, it follow from equation \eqref{eq:*step1-1} that
        \begin{equation*}
            \sqrt{n}
            (\bm{C}
            \hat{\cB}_{a,\cM}
            -
            \bm{t})
            =
            \bm{\omega}_n
            +
            \bm{C}\bm{R}_{a,J_0}
            +
            \sqrt{n}
            \bm{h}_n
        \end{equation*}
        and hence
        \begin{align}
            \sqrt{n}
            ({\bm{C}\Omega_{mm}\bm{C}^T})^{-1/2}
            (\bm{C}
            \hat{\cB}_{a,\cM}
            -
            \bm{t})
            &
            =
            ({\bm{C}\Omega_{mm}\bm{C}^T})^{-1/2}
            (
            \bm{\omega}_n
            +
            \bm{C}\bm{R}_{a,J_0}
            +
            \sqrt{n}
            \bm{h}_n
            )
            \nonumber
            \\
            &
            =
            ({\bm{C}\Omega_{mm}\bm{C}^T})^{-1/2}
            (
            \bm{\omega}_n
            +
            \sqrt{n}
            \bm{h}_n
            )
            +
            ({\bm{C}\Omega_{mm}\bm{C}^T})^{-1/2}
            \bm{C}\bm{R}_{a,J_0}
            \label{eq:*step1-2}
        \end{align}
        by \textcolor{black}{the 6th condition in (A1)}, we have $\lim\inf_n\lambda_{\min}(\bm{K}_n)>0$, which implies that $\lambda_{\max}(\Omega_{mm})=O(1)$.
        Then we can prove that $\Vert
            ({\bm{C}\Omega_{mm}\bm{C}^T})^{-1/2}
            \bm{C}
            \Vert_2 = O(1)$, 
        \begin{align*}
            \Vert
            ({\bm{C}\Omega_{mm}\bm{C}^T})^{-1/2}
            \bm{C}
            \Vert_2^2
            & 
            =
            \Vert
            ({\bm{C}\Omega_{mm}\bm{C}^T})^{-1/2}
            \bm{C}\Omega_{mm}^{1/2}
            \cdot
            \Omega_{mm}^{1/2}
            \Vert_2^2
            \\
            &
            \leq
            \Vert
            ({\bm{C}\Omega_{mm}\bm{C}^T})^{-1/2}
            \bm{C}\Omega_{mm}^{1/2}
            \Vert_2^2
            \Vert
            \Omega_{mm}^{1/2}
            \Vert_2^2
            \\
            &
            =
            \lambda_{\max}
            \left(
            ({\bm{C}\Omega_{mm}\bm{C}^T})^{-1/2}
            {\bm{C}\Omega_{mm}\bm{C}^T}
            ({\bm{C}\Omega_{mm}\bm{C}^T})^{-1/2}
            \right)
            \cdot
            \Vert
            \Omega_{mm}
            \Vert_2
            \\
            &
            =
            \lambda_{\max}(\bm{I}_r)
            \Vert
            \Omega_{mm}
            \Vert_2
            =
            O(1).
        \end{align*}
        This implies that 
        \begin{equation*}
            \Vert
            ({\bm{C}\Omega_{mm}\bm{C}^T})^{-1/2}
            \bm{C}\bm{R}_{a,J_0}
            \Vert_2
            \leq
            \Vert
            ({\bm{C}\Omega_{mm}\bm{C}^T})^{-1/2}
            \bm{C}
            \Vert_2
            \Vert
            \bm{R}_{a,J_0}
            \Vert_2
            =o_p(1).
        \end{equation*}
        This together with equation \eqref{eq:*step1-2} gives
        \begin{equation}
            \sqrt{n}
            ({\bm{C}\Omega_{mm}\bm{C}^T})^{-1/2}
            (\bm{C}
            \hat{\cB}_{a,\cM}
            -
            \bm{t})
            =
            ({\bm{C}\Omega_{mm}\bm{C}^T})^{-1/2}
            (
            \bm{\omega}_n
            +
            \sqrt{n}
            \bm{h}_n
            )
            +
            o_p(1).
            \label{eq:*step1-3}
        \end{equation}
        Next, show that 
        \begin{equation}
            \Vert
            ({\bm{C}\Omega_{mm}\bm{C}^T})^{-1/2}
            (
            \bm{\omega}_n
            +
            \sqrt{n}
            \bm{h}_n
            )
            \Vert_2
            =
            O_p(\sqrt{r}).
            \label{eq:*step1-4}
        \end{equation}
        To prove $\Vert
            ({\bm{C}\Omega_{mm}\bm{C}^T})^{-1/2}
            \bm{\omega}_n
            \Vert_2
            =
            O_p(\sqrt{r})$, we notice that
        \begin{align*}
            &~
            \Vert
            ({\bm{C}\Omega_{mm}\bm{C}^T})^{-1/2}
            \bm{\omega}_n
            \Vert_2^2
            \\
            =
            &~
            \frac{1}{\sqrt{n}}
            \left(
            \sum_{k=1}^{K}w_k
                    \begin{pmatrix}
                            \bm{X}_{\cM}^T
                            \\
                            (\bm{X}^k_{S\cup \cK})^T
                        \end{pmatrix}
                        \bm{H}(\bm{X}^k{{\cB^*}})
                        \left\{
                            \bm{Y}^k - \bm{\mu}(\bm{X}^k{{\cB^*}})
                        \right\}
            \right)^T
            \bm{K}_n^{-1/2}
            \cdot
            \\
            &~
            \bm{K}_n^{-1/2}
            \begin{pmatrix}
                    \bC^T 
                            \\
                            \bm{0}_{r\times (s+K)}^T
                \end{pmatrix}
            ({\bm{C}\Omega_{mm}\bm{C}^T})^{-1}
            \begin{pmatrix}
                    \bC^T 
                            \\
                            \bm{0}_{r\times (s+K)}^T
                \end{pmatrix}^T
            \bm{K}_n^{-1/2}
            \cdot
            \\
            &~
            \bm{K}_n^{-1/2}
            \frac{1}{\sqrt{n}}
            \left(
            \sum_{k=1}^{K}w_k
                    \begin{pmatrix}
                            \bm{X}_{\cM}^T
                            \\
                            (\bm{X}^k_{S\cup \cK})^T
                        \end{pmatrix}
                        \bm{H}(\bm{X}^k{{\cB^*}})
                        \left\{
                            \bm{Y}^k - \bm{\mu}(\bm{X}^k{{\cB^*}})
                        \right\}
            \right)
            \\
            =
            &~
            \left\Vert
            \bm{P}_n^{1/2}
            \bm{K}_n^{-1/2}
            \frac{1}{\sqrt{n}}
            \left(
            \sum_{k=1}^{K}w_k
                    \begin{pmatrix}
                            \bm{X}_{\cM}^T
                            \\
                            (\bm{X}^k_{S\cup \cK})^T
                        \end{pmatrix}
                        \bm{H}(\bm{X}^k{{\cB^*}})
                        \left\{
                            \bm{Y}^k - \bm{\mu}(\bm{X}^k{{\cB^*}})
                        \right\}
            \right)
            \right\Vert_2^2.
        \end{align*}
        Since
        \begin{align*}
            &~
            \mathbb{E}\left[
            \left\Vert
            \bm{P}_n^{1/2}
            \bm{K}_n^{-1/2}
            \frac{1}{\sqrt{n}}
            \left(
            \sum_{k=1}^{K}w_k
                    \begin{pmatrix}
                            \bm{X}_{\cM}^T
                            \\
                            (\bm{X}^k_{S\cup \cK})^T
                        \end{pmatrix}
                        \bm{H}(\bm{X}^k{{\cB^*}})
                        \left\{
                            \bm{Y}^k - \bm{\mu}(\bm{X}^k{{\cB^*}})
                        \right\}
            \right)
            \right\Vert_2^2
            \right]
            \\
            =
            &~
            \frac{1}{n}
            \mathbb{E}\left[
            \left\Vert
            \left.
            \sum_{k=1}^{K}w_k
            \bm{P}_n^{1/2}
            \bm{K}_n^{-1/2}
                    \begin{pmatrix}
                            \bm{X}_{\cM}^T
                            \\
                            (\bm{X}^k_{S\cup \cK})^T
                        \end{pmatrix}
                        \bm{H}(\bm{X}^k{{\cB^*}})
                        \left\{
                            \bm{Y}^k - \bm{\mu}(\bm{X}^k{{\cB^*}})
                        \right\}
            \right.
            \right\Vert_2^2
            \right]
            \\
            \leq
            &~
            \frac{1}{n}
            \sum_{k=1}^{K}w_k
            \mathbb{E}\left[
            \left\Vert
            \left.
            \bm{P}_n^{1/2}
            \bm{K}_n^{-1/2}
                    \begin{pmatrix}
                            \bm{X}_{\cM}^T
                            \\
                            (\bm{X}^k_{S\cup \cK})^T
                        \end{pmatrix}
                        \bm{H}(\bm{X}^k{{\cB^*}})
                        \left\{
                            \bm{Y}^k - \bm{\mu}(\bm{X}^k{{\cB^*}})
                        \right\}
            \right.
            \right\Vert_2^2
            \right]
            \\
            =
            &~
            \frac{1}{n}
            \sum_{k=1}^{K}w_k
            \text{tr}\left[
            \bm{P}_n^{1/2}
            \bm{K}_n^{-1/2}
            \begin{pmatrix}
                            \bm{X}_{\cM}^T
                            \\
                            (\bm{X}^k_{S\cup \cK})^T
                        \end{pmatrix}
            \bm{\Sigma}(\bm{X}^k{{\cB^*}})
            \begin{pmatrix}
                            \bm{X}_{\cM}^T
                            \\
                            (\bm{X}^k_{S\cup \cK})^T
                        \end{pmatrix}^T
            \bm{K}_n^{-1/2}
            \bm{P}_n^{1/2}
            \right]
            \\
            =
            &~
            \text{tr}\left[
            \bm{P}_n^{1/2}
            \bm{K}_n^{-1/2}
            \left(
            \frac{1}{n}
            \sum_{k=1}^{K}w_k
            \begin{pmatrix}
                            \bm{X}_{\cM}^T
                            \\
                            (\bm{X}^k_{S\cup \cK})^T
                        \end{pmatrix}
            \bm{\Sigma}(\bm{X}^k{{\cB^*}})
            \begin{pmatrix}
                            \bm{X}_{\cM}^T
                            \\
                            (\bm{X}^k_{S\cup \cK})^T
                        \end{pmatrix}^T
            \right)
            \bm{K}_n^{-1/2}
            \bm{P}_n^{1/2}
            \right]
            \\
            =
            &~
            \text{tr}\left[
            \bm{P}_n^{1/2}
            \bm{K}_n^{-1/2}
            \bm{K}_n
            \bm{K}_n^{-1/2}
            \bm{P}_n^{1/2}
            \right]
            \\
            =
            &~
            \text{tr}\left[
            \bm{P}_n
            \right]
            =
            r,
        \end{align*}
        where the first inequality is by the convexity of $\Vert\cdot\Vert_2^2$, and $\text{tr}\left[
            \bm{P}_n
            \right]
            =
            r$ is because $\bm{P}_n$ is the projection matrix of rank $r$. By Markov's inequality, we have
        \begin{equation*}
            \left\Vert
            \bm{P}_n^{1/2}
            \bm{K}_n^{-1/2}
            \frac{1}{\sqrt{n}}
            \left(
            \sum_{k=1}^{K}w_k
                    \begin{pmatrix}
                            \bm{X}_{\cM}^T
                            \\
                            (\bm{X}^k_{S\cup \cK})^T
                        \end{pmatrix}
                        \bm{H}(\bm{X}^k{{\cB^*}})
                        \left\{
                            \bm{Y}^k - \bm{\mu}(\bm{X}^k{{\cB^*}})
                        \right\}
            \right)
            \right\Vert_2^2
            =
            O_p(r).
        \end{equation*}
        Therefore, we proved that $\Vert
            ({\bm{C}\Omega_{mm}\bm{C}^T})^{-1/2}
            \bm{\omega}_n
            \Vert_2
            =O_p(\sqrt{r})$. 
        Besides, we can show that $\Vert
            ({\bm{C}\Omega_{mm}\bm{C}^T})^{-1/2}
            \Vert_2=O(1)$:
        \begin{align*}
            \lambda_{\min}({\bm{C}\Omega_{mm}\bm{C}^T})
            &
            =
            \inf_{\bm{a}\in\mathbb{R}^r,\Vert\bm{a}\Vert_2=1}
            \bm{a}^T
            {\bm{C}\Omega_{mm}\bm{C}^T}
            \bm{a}
            \\
            &
            \geq
            \inf_{\bm{a}\in\mathbb{R}^r,\Vert\bm{a}\Vert_2=1}
            \lambda_{\min}
            ({\Omega_{mm}})
            \cdot
            \bm{a}^T
            {\bm{C}\bm{C}^T}
            \bm{a}
            \\
            &
            \geq
            \lambda_{\min}
            ({\Omega_{n}})
            \cdot 
            \lambda_{\min } 
            ({\bm{C}\bm{C}^T}).
        \end{align*}
        Under \textcolor{black}{the 4th condition in (A1)}, we have $\lambda_{\max}(\bm{K}_n)=O(1)$. This implies that $\liminf_{n}\lambda_{\min}(\bm{K}_n^{-1})>0$. And \textcolor{black}{condition (A3)}: $\lambda_{\max}(({\bm{C}\bm{C}^T})^{-1})=O(1)$ implies that $\liminf_{n}\lambda_{\min}({\bm{C}\bm{C}^T})>0$. Thus, $\lim\inf_{n}\lambda_{\min}({\bm{C}\Omega_{mm}\bm{C}^T})>0$, or equivalently, $\Vert
            ({\bm{C}\Omega_{mm}\bm{C}^T})^{-1/2}
            \Vert_2=O(1)$. This together with \textcolor{black}{condition (A3)} $\Vert\bm{h}_n\Vert_2=O(\sqrt{r/n})$ implies that
        \begin{equation*}
            \Vert
            \sqrt{n}
            ({\bm{C}\Omega_{mm}\bm{C}^T})^{-1/2}
            \bm{h}_n
            \Vert_2=O(\sqrt{r}).
        \end{equation*}
        Thus equation \eqref{eq:*step1-4} is proved as desired. Along with \eqref{eq:*step1-3} implies that
        \begin{equation}
            \Vert
                \sqrt{n}
                ({\bm{C}\Omega_{mm}\bm{C}^T})^{-1/2}
                (\bm{C}
                \hat{\cB}_{a,\cM}
                -
                \bm{t})
            \Vert_2
            =
            O_p(\sqrt{r}).
            \label{eq:*step1-5}
        \end{equation}
        By Lemma \ref{lemma1}, we have
        \begin{equation}
            \left\Vert
                \bm{I} - 
                    \Psi_n^{1/2}
                    ({\bm{C}\widehat{\Omega}_{a,mm}\bm{C}^T})^{-1}
                    \Psi_n^{1/2}
            \right\Vert_2
            =
            O_p
            \left(
            \frac{s+m+K}{\sqrt{n}}
            \right),
            \label{eq:*step1-lemma1}
        \end{equation}
        where $\Psi_n={\bm{C}\Omega_{mm}\bm{C}^T}$, 
        then we will have
        \begin{align*}
            &~
            \left\vert
                n
                (\bm{C}
                \hat{\cB}_{a,\cM}
                -
                \bm{t})^T
                ({\bm{C}\Omega_{mm}\bm{C}^T})^{-1}
                (\bm{C}
                \hat{\cB}_{a,\cM}
                -
                \bm{t})
                -
                n
                (\bm{C}
                \hat{\cB}_{a,\cM}
                -
                \bm{t})^T
                ({\bm{C}\widehat{\Omega}_{a,mm}\bm{C}^T})^{-1}
                (\bm{C}
                \hat{\cB}_{a,\cM}
                -
                \bm{t})
            \right\vert
            \\
            = 
            &~
            \left\vert
                n
                (\bm{C}
                \hat{\cB}_{a,\cM}
                -
                \bm{t})^T
                \Psi_n^{-1/2}
                \left(
                    \bm{I} - 
                    \Psi_n^{1/2}
                    ({\bm{C}\widehat{\Omega}_{a,mm}\bm{C}^T})^{-1}
                    \Psi_n^{1/2}
                \right)
                \Psi_n^{-1/2}
                (\bm{C}
                \hat{\cB}_{a,\cM}
                -
                \bm{t})
            \right\vert
            \\
            \leq
            &~
            \left\Vert
                \bm{I} - 
                    \Psi_n^{1/2}
                    ({\bm{C}\widehat{\Omega}_{a,mm}\bm{C}^T})^{-1}
                    \Psi_n^{1/2}
            \right\Vert_2
            \left\Vert
                \sqrt{n}
                \Psi_n^{-1/2}
                (\bm{C}
                \hat{\cB}_{a,\cM}
                -
                \bm{t})
            \right\Vert_2^2
            \\
            =
            &~
            O_p
            \left(
            \frac{s+m+K}{\sqrt{n}}
            \right)
            \cdot
            O_p(\sqrt{r})
            =
            o_p(r),
        \end{align*}
        where the last equality is because the assumption that $s+m+K=o(n^{1/3})$. Therefore, by the definition of $T_W$, we have
        \begin{align*}
            T_W
            &=
            n
                (\bm{C}
                \hat{\cB}_{a,\cM}
                -
                \bm{t})^T
                ({\bm{C}\Omega_{mm}\bm{C}^T})^{-1}
                (\bm{C}
                \hat{\cB}_{a,\cM}
                -
                \bm{t})
            +
            o_p(r)
            \\
            &=
            T_0
            +
            o_p({r}).
        \end{align*}

    \paragraph{Step 2: Show that $T_S/r$ is equivalent to $T_0/r$.}

        Based on the proof of Theorem 2.1, we have:
        \begin{align}
                & ~ 
                \frac{1}{\sqrt{n}}
                \sum_{k=1}^{K}w_k
                    \begin{pmatrix}
                            \bm{X}_{\cM}^T
                            \\
                            (\bm{X}^k_{S\cup \cK})^T
                        \end{pmatrix}
                        \bm{H}(\bm{X}^k{\hat{\cB}_0})
                        \left\{
                            \bm{Y}^k - \bm{\mu}(\bm{X}^k{\hat{\cB}_0})
                        \right\}
                \nonumber
                \\
                =
                &~
                \frac{1}{\sqrt{n}}
                \sum_{k=1}^{K}w_k
                    \begin{pmatrix}
                            \bm{X}_{\cM}^T
                            \\
                            (\bm{X}^k_{S\cup \cK})^T
                        \end{pmatrix}
                        \bm{H}(\bm{X}^k{{\cB^*}})
                        \left\{
                            \bm{Y}^k - \bm{\mu}(\bm{X}^k{{\cB^*}})
                        \right\}
                \nonumber
                \\
                -
                &~
                \frac{1}{\sqrt{n}}
                \sum_{k=1}^{K}w_k 
                        \begin{pmatrix}
                            \bm{X}_{\cM}^T
                            \\
                            (\bm{X}^k_{S\cup \cK})^T
                        \end{pmatrix}
                        {\Sigma}(\bm{X}^k{\cB^*})
                        \begin{pmatrix}
                            \bm{X}_{\cM}^T
                            \\
                            (\bm{X}^k_{S\cup \cK})^T
                        \end{pmatrix}^T
                        \begin{pmatrix}
                            \hat{\cB}_{0,\cM} - \cB^*_{\cM} \\
                            \hat{\cB}_{0,S\cup\cK} - \cB^*_{S\cup\cK}
                        \end{pmatrix}
                + o_p(1).
                \label{eq:*step2-1}
            \end{align}
            and 
            \begin{align}
                \sqrt{n}
                \begin{pmatrix}
                            \hat{\cB}_{0,\cM} - \cB^*_{\cM} \\
                            \hat{\cB}_{0,S\cup\cK} - \cB^*_{S\cup\cK}
                        \end{pmatrix}
                & = 
                \frac{1}{\sqrt{n}}
                \bm{K}_n^{-1/2}
                (\bm{I}-\bm{P}_n)
                \bm{K}_n^{-1/2}
                \sum_{k=1}^{K}w_k
                    \begin{pmatrix}
                            \bm{X}_{\cM}^T
                            \\
                            (\bm{X}^k_{S\cup \cK})^T
                        \end{pmatrix}
                        \bm{H}(\bm{X}^k{{\cB^*}})
                        \left\{
                            \bm{Y}^k - \bm{\mu}(\bm{X}^k{{\cB^*}})
                        \right\}
                \nonumber
                \\
                & -
                \sqrt{n}
                \bm{K}_n^{-1}
                \begin{pmatrix}
                    \bC^T 
                            \\
                            \bm{0}_{r\times (s+K)}^T
                \end{pmatrix}
                {(\bm{C}\Omega_{mm}\bm{C}^T)^{-1}}
                \bm{h}_n
                +
                o_p(1).
            \end{align}
            Since $\lambda_{\max}(\bm{K}_n)=O(1)$, we have
            \begin{align}
                \sqrt{n}
                \bm{K}_n
                \begin{pmatrix}
                            \hat{\cB}_{0,\cM} - \cB^*_{\cM} \\
                            \hat{\cB}_{0,S\cup\cK} - \cB^*_{S\cup\cK}
                        \end{pmatrix}
                & = 
                \frac{1}{\sqrt{n}}
                \bm{K}_n^{1/2}
                (\bm{I}-\bm{P}_n)
                \bm{K}_n^{-1/2}
                \sum_{k=1}^{K}w_k
                    \begin{pmatrix}
                            \bm{X}_{\cM}^T
                            \\
                            (\bm{X}^k_{S\cup \cK})^T
                        \end{pmatrix}
                        \bm{H}(\bm{X}^k{{\cB^*}})
                        \left\{
                            \bm{Y}^k - \bm{\mu}(\bm{X}^k{{\cB^*}})
                        \right\}
                \nonumber
                \\
                & -
                \sqrt{n}
                \begin{pmatrix}
                    \bC^T 
                            \\
                            \bm{0}_{r\times (s+K)}^T
                \end{pmatrix}
                {(\bm{C}\Omega_{mm}\bm{C}^T)^{-1}}
                \bm{h}_n
                +
                o_p(1),
            \end{align}
            The LHS of above equation is the same as the second term in the RHS of equation \eqref{eq:*step2-1}. Thus, this together with \eqref{eq:*step2-1} implies that
            \begin{align}
                & ~ 
                \frac{1}{\sqrt{n}}
                \sum_{k=1}^{K}w_k
                    \begin{pmatrix}
                            \bm{X}_{\cM}^T
                            \\
                            (\bm{X}^k_{S\cup \cK})^T
                        \end{pmatrix}
                        \bm{H}(\bm{X}^k{\hat{\cB}_0})
                        \left\{
                            \bm{Y}^k - \bm{\mu}(\bm{X}^k{\hat{\cB}_0})
                        \right\}
                \nonumber
                \\
                =
                &~
                \frac{1}{\sqrt{n}}
                \sum_{k=1}^{K}w_k
                    \begin{pmatrix}
                            \bm{X}_{\cM}^T
                            \\
                            (\bm{X}^k_{S\cup \cK})^T
                        \end{pmatrix}
                        \bm{H}(\bm{X}^k{{\cB^*}})
                        \left\{
                            \bm{Y}^k - \bm{\mu}(\bm{X}^k{{\cB^*}})
                        \right\}
                \nonumber
                \\
                -
                &~
                \frac{1}{\sqrt{n}}
                \bm{K}_n^{1/2}
                (\bm{I}-\bm{P}_n)
                \bm{K}_n^{-1/2}
                \sum_{k=1}^{K}w_k
                    \begin{pmatrix}
                            \bm{X}_{\cM}^T
                            \\
                            (\bm{X}^k_{S\cup \cK})^T
                        \end{pmatrix}
                        \bm{H}(\bm{X}^k{{\cB^*}})
                        \left\{
                            \bm{Y}^k - \bm{\mu}(\bm{X}^k{{\cB^*}})
                        \right\}
                \nonumber
                \\
                +
                &~
                \sqrt{n}
                \begin{pmatrix}
                    \bC^T 
                            \\
                            \bm{0}_{r\times (s+K)}^T
                \end{pmatrix}
                {(\bm{C}\Omega_{mm}\bm{C}^T)^{-1}}
                \bm{h}_n
                + o_p(1)
                \nonumber
                \\
                =
                &~
                \frac{1}{\sqrt{n}}
                \bm{K}_n^{1/2}
                \bm{P}_n
                \bm{K}_n^{-1/2}
                \sum_{k=1}^{K}w_k
                    \begin{pmatrix}
                            \bm{X}_{\cM}^T
                            \\
                            (\bm{X}^k_{S\cup \cK})^T
                        \end{pmatrix}
                        \bm{H}(\bm{X}^k{{\cB^*}})
                        \left\{
                            \bm{Y}^k - \bm{\mu}(\bm{X}^k{{\cB^*}})
                        \right\}
                \nonumber
                \\
                +
                &~
                \sqrt{n}
                \begin{pmatrix}
                    \bC^T 
                            \\
                            \bm{0}_{r\times (s+K)}^T
                \end{pmatrix}
                {(\bm{C}\Omega_{mm}\bm{C}^T)^{-1}}
                \bm{h}_n
                + o_p(1)
                \nonumber
                \\
                =
                &~
                \begin{pmatrix}
                    \bC^T 
                            \\
                            \bm{0}_{r\times (s+K)}^T
                \end{pmatrix}
                {(\bm{C}\Omega_{mm}\bm{C}^T)^{-1}}
                (\bm{\omega_n}
                +
                \sqrt{n}
                \bm{h}_n
                )
                +
                o_p(1),
            \end{align}
            where the last equality is by the definition of $\bm{P}_n$ and $\bm{\omega}_n$.
            Multiplying both sides by $\bm{K}_n^{-1/2}$ and from the fact that $\lim\inf_n\lambda_{\min}(\bm{K}_n)>0$, we have
            \begin{align}
                & ~ 
                \frac{1}{\sqrt{n}}
                \bm{K}_n^{-1/2}
                \sum_{k=1}^{K}w_k
                    \begin{pmatrix}
                            \bm{X}_{\cM}^T
                            \\
                            (\bm{X}^k_{S\cup \cK})^T
                        \end{pmatrix}
                        \bm{H}(\bm{X}^k{\hat{\cB}_0})
                        \left\{
                            \bm{Y}^k - \bm{\mu}(\bm{X}^k{\hat{\cB}_0})
                        \right\}
                \nonumber
                \\
                =
                &~
                \bm{K}_n^{-1/2}
                \begin{pmatrix}
                    \bC^T 
                            \\
                            \bm{0}_{r\times (s+K)}^T
                \end{pmatrix}
                {(\bm{C}\Omega_{mm}\bm{C}^T)^{-1}}
                (\bm{\omega_n}
                +
                \sqrt{n}
                \bm{h}_n
                )
                +
                o_p(1),
            \end{align}
            We proved in Step 1 that $\Vert
            ({\bm{C}\Omega_{mm}\bm{C}^T})^{-1/2}
            (\bm{\omega_n}
                +
                \sqrt{n}
                \bm{h}_n
                )
            \Vert_2
            =O_p(\sqrt{r})$, thus
            \begin{equation*}
                \left\Vert
                    \bm{K}_n^{-1/2}
                \begin{pmatrix}
                    \bC^T 
                            \\
                            \bm{0}_{r\times (s+K)}^T
                \end{pmatrix}
                {(\bm{C}\Omega_{mm}\bm{C}^T)^{-1}}
                (\bm{\omega_n}
                +
                \sqrt{n}
                \bm{h}_n
                )
                \right\Vert_2^2
                =
                \Vert
            ({\bm{C}\Omega_{mm}\bm{C}^T})^{-1/2}
            (\bm{\omega_n}
                +
                \sqrt{n}
                \bm{h}_n
                )
            \Vert_2^2
            =O_p({r}).
            \end{equation*}
        This implies that 
        \begin{equation*}
            \left\Vert
                    \frac{1}{\sqrt{n}}
                \bm{K}_n^{-1/2}
                \sum_{k=1}^{K}w_k
                    \begin{pmatrix}
                            \bm{X}_{\cM}^T
                            \\
                            (\bm{X}^k_{S\cup \cK})^T
                        \end{pmatrix}
                        \bm{H}(\bm{X}^k{\hat{\cB}_0})
                        \left\{
                            \bm{Y}^k - \bm{\mu}(\bm{X}^k{\hat{\cB}_0})
                        \right\}
                \right\Vert_2
            =O_p({\sqrt{r}}).
        \end{equation*}
        By \textcolor{black}{Lemma \ref{lemma1}}, we have
        \begin{equation*}
            \left\Vert
                \bm{I} 
                - 
                \Omega_n^{-1/2}
                \widehat{\Omega}_0
                \Omega_n^{-1/2}
            \right\Vert_2
            =
            O_p
            \left(
            \frac{s+m+K}{\sqrt{n}}
            \right).
        \end{equation*}
        Then under event $\{\widehat{S}_0=S\}$, we have
        \begin{align*}
            &~
        \left\vert
            \frac{1}{n}
            \mathcal{S}_n(\hat{\cB}_0)^T
            (\Omega_n - \widehat{\Omega}_0)
            \mathcal{S}_n(\hat{\cB}_0)
            \right\vert
            \\
            \leq
            &~
            \left\Vert
                \frac{1}{\sqrt{n}}
                \bm{K}_n^{-1/2}
                \sum_{k=1}^{K}w_k
                    \begin{pmatrix}
                            \bm{X}_{\cM}^T
                            \\
                            (\bm{X}^k_{\widehat{S}_0\cup \cK})^T
                        \end{pmatrix}
                        \bm{H}(\bm{X}^k{\hat{\cB}_0})
                        \left\{
                            \bm{Y}^k - \bm{\mu}(\bm{X}^k{\hat{\cB}_0})
                        \right\}
            \right\Vert_2^2
            \left\Vert
                \bm{I} 
                - 
                \Omega_n^{-1/2}
                \widehat{\Omega}_0
                \Omega_n^{-1/2}
            \right\Vert_2
            \\
            =
            &~
            \left\Vert
                \frac{1}{\sqrt{n}}
                \bm{K}_n^{-1/2}
                \sum_{k=1}^{K}w_k
                    \begin{pmatrix}
                            \bm{X}_{\cM}^T
                            \\
                            (\bm{X}^k_{S\cup \cK})^T
                        \end{pmatrix}
                        \bm{H}(\bm{X}^k{\hat{\cB}_0})
                        \left\{
                            \bm{Y}^k - \bm{\mu}(\bm{X}^k{\hat{\cB}_0})
                        \right\}
            \right\Vert_2^2
            \left\Vert
                \bm{I} 
                - 
                \Omega_n^{-1/2}
                \widehat{\Omega}_0
                \Omega_n^{-1/2}
            \right\Vert_2
            \\
            =
            &~
            O_p
            \left(
            \frac{r(s+m+K)}{\sqrt{n}}
            \right)
            =
            o_p(r).
        \end{align*}
        Since $\mathbb{P}(\widehat{S}_0=S)\rightarrow 1$ and $s+m+K=o(\sqrt{n})$, we have
        \begin{equation*}
            T_S = T_0+o_p(r).
        \end{equation*}

    \paragraph{Step 3: Show that $T_L/r$ is equivalent to $T_0/r$.}

        By Theorem 2.1, we have
        \begin{align}
                \sqrt{n}
                \begin{pmatrix}
                            \hat{\cB}_{a,\cM} - \hat{\cB}_{0,\cM} \\
                            \hat{\cB}_{a,S\cup\cK} - \hat{\cB}_{0,S\cup\cK}
                        \end{pmatrix}
                & = 
                \frac{1}{\sqrt{n}}
                \bm{K}_n^{-1/2}
                {\bm{P}_n}
                \bm{K}_n^{-1/2}
                \sum_{k=1}^{K}w_k
                    \begin{pmatrix}
                            \bm{X}_{\cM}^T
                            \\
                            (\bm{X}^k_{S\cup \cK})^T
                        \end{pmatrix}
                        \bm{H}(\bm{X}^k{{\cB^*}})
                        \left\{
                            \bm{Y}^k - \bm{\mu}(\bm{X}^k{{\cB^*}})
                        \right\}
                \nonumber
                \\
                & +
                \sqrt{n}
                \bm{K}_n^{-1/2}
                \bm{P}_n
                \bm{K}_n^{-1/2}
                \begin{pmatrix}
                    \bm{C}^T({\bm{C}\Omega_{mm}\bm{C}^T})^{-1}\bm{h}_n
                    \\
                    \bm{0}_{s+K}
                \end{pmatrix}
                +
                o_p(1).
                \nonumber
        \end{align}
        Notice that
        \begin{align*}
            & ~
            \bm{K}_n^{-1/2}
                \bm{P}_n
                \bm{K}_n^{-1/2}
                \begin{pmatrix}
                    \bm{C}^T({\bm{C}\Omega_{mm}\bm{C}^T})^{-1}\bm{h}_n
                    \\
                    \bm{0}_{s+K}
                \end{pmatrix}
            \\
            =
            &~
            \bm{K}_n^{-1}
            \begin{pmatrix}
                    \bC^T 
                            \\
                            \bm{0}_{r\times (s+K)}^T
                \end{pmatrix}
            ({\bm{C}\Omega_{mm}\bm{C}^T})^{-1}
            \begin{pmatrix}
                    \bC^T 
                            \\
                            \bm{0}_{r\times (s+K)}^T
                \end{pmatrix}^T
            \bm{K}_n^{-1}
            \begin{pmatrix}
                    \bC^T 
                            \\
                            \bm{0}_{r\times (s+K)}^T
                \end{pmatrix}
            ({\bm{C}\Omega_{mm}\bm{C}^T})^{-1}
            \bm{h}_n    
            \\
            =
            & ~
            \bm{K}_n^{-1}
            \begin{pmatrix}
                    \bC^T 
                            \\
                            \bm{0}_{r\times (s+K)}^T
                \end{pmatrix}
            ({\bm{C}\Omega_{mm}\bm{C}^T})^{-1}
            \bm{h}_n.   
        \end{align*}
        It follows that
        \begin{align}
                \sqrt{n}
                \begin{pmatrix}
                            \hat{\cB}_{a,\cM} - \hat{\cB}_{0,\cM} \\
                            \hat{\cB}_{a,S\cup\cK} - \hat{\cB}_{0,S\cup\cK}
                        \end{pmatrix}
            & = 
                \frac{1}{\sqrt{n}}
                \bm{K}_n^{-1/2}
                {\bm{P}_n}
                \bm{K}_n^{-1/2}
                \sum_{k=1}^{K}w_k
                    \begin{pmatrix}
                            \bm{X}_{\cM}^T
                            \\
                            (\bm{X}^k_{S\cup \cK})^T
                        \end{pmatrix}
                        \bm{H}(\bm{X}^k{{\cB^*}})
                        \left\{
                            \bm{Y}^k - \bm{\mu}(\bm{X}^k{{\cB^*}})
                        \right\}
                \nonumber
            \\
            & +
                \sqrt{n}
                \bm{K}_n^{-1}
            \begin{pmatrix}
                    \bC^T 
                            \\
                            \bm{0}_{r\times (s+K)}^T
                \end{pmatrix}
            ({\bm{C}\Omega_{mm}\bm{C}^T})^{-1}
            \bm{h}_n
                +
                o_p(1)
            \nonumber
            \\
            & =
                \bm{K}_n^{-1}
            \begin{pmatrix}
                    \bC^T 
                            \\
                            \bm{0}_{r\times (s+K)}^T
                \end{pmatrix}
            ({\bm{C}\Omega_{mm}\bm{C}^T})^{-1}
            (\bm{\omega}_n+\sqrt{n}\bm{h}_n)
            + o_p(1).
                \label{eq:*step3-1}
        \end{align}
        We proved in Step 1 that $\Vert
            ({\bm{C}\Omega_{mm}\bm{C}^T})^{-1/2}
            (\bm{\omega_n}
                +
                \sqrt{n}
                \bm{h}_n
                )
            \Vert_2
            =O_p(\sqrt{r})$, thus
            \begin{align}
                &~
                \left\Vert
                    \bm{K}_n^{-1}
                \begin{pmatrix}
                    \bC^T 
                            \\
                            \bm{0}_{r\times (s+K)}^T
                \end{pmatrix}
                {(\bm{C}\Omega_{mm}\bm{C}^T)^{-1}}
                (\bm{\omega_n}
                +
                \sqrt{n}
                \bm{h}_n
                )
                \right\Vert_2^2
                \nonumber
                \\
                \leq&~
                \left\Vert
                \bm{K}_n^{-1/2}
                \right\Vert_2^2
                \left\Vert
                    \bm{K}_n^{-1/2}
                    \begin{pmatrix}
                        \bC^T 
                                \\
                                \bm{0}_{r\times (s+K)}^T
                    \end{pmatrix}
                    {(\bm{C}\Omega_{mm}\bm{C}^T)^{-1}}
                    (\bm{\omega_n}
                    +
                    \sqrt{n}
                    \bm{h}_n
                    )
                \right\Vert_2^2
                \nonumber
            \\
            =&~
            \lambda_{\max}(\bm{K}_n^{-1})
            \left\Vert
                    {(\bm{C}\Omega_{mm}\bm{C}^T)^{-1/2}}
                    (\bm{\omega_n}
                    +
                    \sqrt{n}
                    \bm{h}_n
                    )
                \right\Vert_2^2
                \nonumber
            \\
            =&~
            O_p({r}).
            \label{eq:*step3-1.5}
        \end{align}
        Thus, we proved that
        \begin{equation}
            \left\Vert
            \hat{\cB}_{a,\cM\cup S\cup\cK} - \hat{\cB}_{0,\cM\cup S\cup\cK}
            \right\Vert_2
            =
            O_p({\sqrt{r/n}})
            \label{eq:*step3-2}
        \end{equation}
        Under the event $\left\{\hat{\cB}_{a,(\cM\cup S\cup\cK)^c} = \hat{\cB}_{0,(\cM\cup S\cup\cK)^c}=\bm{0}\right\}$, denote $\hat{\Delta}_{a,0}=\hat{\cB}_{a,\cM\cup S\cup\cK} - \hat{\cB}_{0,\cM\cup S\cup\cK}$. By the second-order Taylor expansion, we obtain that
        \begin{align}
            M_n(\hat{\cB}_{a})
            -
            M_n(\hat{\cB}_{0})
            &
            =
            \frac{1}{n}
                \hat{\Delta}_{a,0}^T
            \sum_{k=1}^{K}w_k
                    \begin{pmatrix}
                            \bm{X}_{\cM}^T
                            \\
                            (\bm{X}^k_{S\cup \cK})^T
                        \end{pmatrix}
                        \bm{H}(\bm{X}^k{\hat{\cB}_a})
                        \left\{
                            \bm{Y}^k - \bm{\mu}(\bm{X}^k{\hat{\cB}_a})
                        \right\}
            \nonumber
            \\
            & 
            +
            \frac{1}{2n}
            \hat{\Delta}_{a,0}^T
            \sum_{k=1}^{K}
            w_k 
                        \begin{pmatrix}
                            \bm{X}_{\cM}^T
                            \\
                            (\bm{X}^k_{S\cup \cK})^T
                        \end{pmatrix}
                        \hat{\Sigma}(\bm{X}^k{\hat{\cB}_a})
                        \begin{pmatrix}
                            \bm{X}_{\cM}^T
                            \\
                            (\bm{X}^k_{S\cup \cK})^T
                        \end{pmatrix}^T
                        \hat{\Delta}_{a,0}
                +
            \frac{1}{n}
            \hat{\Delta}_{a,0}^T
                \bm{R},
                \label{eq:*step3-2.5}
        \end{align}
        where $\bm{R}$ satisfies $\Vert\bm{R}\Vert_{2}=O_p(r\sqrt{s+m+K})$. Therefore, we have
        \begin{align}
            \left\Vert
                \frac{1}{n}
                \hat{\Delta}_{a,0}^T
                \bm{R}
            \right\Vert_2
            & 
            \leq
            \frac{1}{n}
            \left\Vert
                \hat{\Delta}_{a,0}
            \right\Vert_2
            \left\Vert
                \bm{R}
            \right\Vert_2
            =
            O_p\left(
            \frac{\sqrt{r}}{n}
            \cdot
            \frac{r\sqrt{s+m+K}}{\sqrt{n}}
            \right)
            =
            o_p
            \left(
            \frac{{\sqrt{r}}}{n}
            \right),
            \label{eq:*step3-2.75}
        \end{align}
        where the last equality is because $(s+m+K)=o(n^{1/3})$. Notice that 
        \begin{equation*}
            \hat{\Sigma}(\bm{X}^k{\hat{\cB}_a}) = 
            (\hat{\Sigma}(\bm{X}^k{\hat{\cB}_a})
                        -
                        \hat{\Sigma}(\bm{X}^k{{\cB^*}}))
            +
            (\hat{\Sigma}(\bm{X}^k{{\cB^*}})
                        -
                        {\Sigma}(\bm{X}^k{{\cB^*}}))
            +
            {\Sigma}(\bm{X}^k{{\cB^*}}).
        \end{equation*}
        By Cauchy Schwartz inequality,
        \begin{align}
            &~
            \left\vert
                \frac{1}{n}
                \hat{\Delta}_{a,0}^T
                \sum_{k=1}^{K}w_k 
                        \begin{pmatrix}
                            \bm{X}_{\cM}^T
                            \\
                            (\bm{X}^k_{S\cup \cK})^T
                        \end{pmatrix}
                        (\hat{\Sigma}(\bm{X}^k{\hat{\cB}_a})
                        -
                        \hat{\Sigma}(\bm{X}^k{{\cB^*}})
                        \begin{pmatrix}
                            \bm{X}_{\cM}^T
                            \\
                            (\bm{X}^k_{S\cup \cK})^T
                        \end{pmatrix}^T
                        \hat{\Delta}_{a,0}
            \right\vert
            \nonumber
            \\
            \leq&~
            \left\Vert
            \frac{1}{n}
            \sum_{k=1}^{K}w_k 
                        \begin{pmatrix}
                            \bm{X}_{\cM}^T
                            \\
                            (\bm{X}^k_{S\cup \cK})^T
                        \end{pmatrix}
                        (\hat{\Sigma}(\bm{X}^k{\hat{\cB}_a})
                        -
                        \hat{\Sigma}(\bm{X}^k{{\cB^*}})
                        \begin{pmatrix}
                            \bm{X}_{\cM}^T
                            \\
                            (\bm{X}^k_{S\cup \cK})^T
                        \end{pmatrix}^T
            \right\Vert_2
            \left\Vert
            \hat{\Delta}_{a,0}
            \right\Vert_2^2.
            \nonumber
        \end{align}
        By the \textcolor{black}{the 4th condition of (A1)} and $\sup_{t}\vert\hat{\Sigma}^\prime(t)\vert=O_p(1)$, we have
        \begin{align*}
            &~
            \left\Vert
            \frac{1}{n}
            \sum_{k=1}^{K}w_k 
                        \begin{pmatrix}
                            \bm{X}_{\cM}^T
                            \\
                            (\bm{X}^k_{S\cup \cK})^T
                        \end{pmatrix}
                        (\hat{\Sigma}(\bm{X}^k{\hat{\cB}_a})
                        -
                        \hat{\Sigma}(\bm{X}^k{{\cB^*}})
                        \begin{pmatrix}
                            \bm{X}_{\cM}^T
                            \\
                            (\bm{X}^k_{S\cup \cK})^T
                        \end{pmatrix}^T
            \right\Vert_2
            \\
            =&~
            \left\Vert
            \frac{1}{n}
            \sum_{k=1}^{K}w_k 
                        \begin{pmatrix}
                            \bm{X}_{\cM}^T
                            \\
                            (\bm{X}^k_{S\cup \cK})^T
                        \end{pmatrix}
                        \hat{\Sigma}^\prime(\bm{X}^k{\bar{\cB}})
                        \begin{pmatrix}
                            \bm{X}_{\cM}^T
                            \\
                            (\bm{X}^k_{S\cup \cK})^T
                        \end{pmatrix}^T
                 \hat{\Delta}_{a,0}
            \right\Vert_2
            \\
            \leq &~
            \left\Vert
            \frac{1}{n}
            \sum_{k=1}^{K}w_k 
                        \begin{pmatrix}
                            \bm{X}_{\cM}^T
                            \\
                            (\bm{X}^k_{S\cup \cK})^T
                        \end{pmatrix}
                        \sup_{t}\vert\hat{\Sigma}^\prime(t)\vert
                        \begin{pmatrix}
                            \bm{X}_{\cM}^T
                            \\
                            (\bm{X}^k_{S\cup \cK})^T
                        \end{pmatrix}^T
            \right\Vert_2
            \left\Vert
                 \hat{\Delta}_{a,0}
            \right\Vert_2
            \\
            =&~
            O_p(1)\cdot
            O_p\left(
            \sqrt{{r}/{n}}
            \right).
        \end{align*}
        Therefore, we have
        \begin{align}
            &~
            \left\vert
                \frac{1}{n}
                 \hat{\Delta}_{a,0}^T
                \sum_{k=1}^{K}w_k 
                        \begin{pmatrix}
                            \bm{X}_{\cM}^T
                            \\
                            (\bm{X}^k_{S\cup \cK})^T
                        \end{pmatrix}
                        (\hat{\Sigma}(\bm{X}^k{\hat{\cB}_a})
                        -
                        \hat{\Sigma}(\bm{X}^k{{\cB^*}})
                        \begin{pmatrix}
                            \bm{X}_{\cM}^T
                            \\
                            (\bm{X}^k_{S\cup \cK})^T
                        \end{pmatrix}^T
                         \hat{\Delta}_{a,0}
            \right\vert
            \nonumber
            \\
            \leq&~
            O_p(\sqrt{r/n})\cdot O_p(r/n)=o_p(\sqrt{r}/n).
            \label{eq:*step3-3}
        \end{align}
        Next, show that
        \begin{align}
            \left\vert
            \frac{1}{n}
                 \hat{\Delta}_{a,0}^T
                \sum_{k=1}^{K}w_k 
                        \begin{pmatrix}
                            \bm{X}_{\cM}^T
                            \\
                            (\bm{X}^k_{S\cup \cK})^T
                        \end{pmatrix}
                        (
                        \hat{\Sigma}(\bm{X}^k{{\cB^*}})
                        -
                        {\Sigma}(\bm{X}^k{{\cB^*}})
                        )
                        \begin{pmatrix}
                            \bm{X}_{\cM}^T
                            \\
                            (\bm{X}^k_{S\cup \cK})^T
                        \end{pmatrix}^T
                         \hat{\Delta}_{a,0}
            \right\vert 
            =
            o_p
            \left(
            \frac{{\sqrt{r}}}{n}
            \right).
            \label{eq:*step3-4}
        \end{align}
        When proving the asymptotic distribution in Theorem 2.1 \eqref{eq:step3-4.5}, we proved that
        \begin{align*}
                \left\Vert
                \sum_{k=1}^{K}w_k 
                        \begin{pmatrix}
                            \bm{X}_{\cM}^T
                            \\
                            (\bm{X}^k_{S\cup \cK})^T
                        \end{pmatrix}
                        \left(
                        {\Sigma}(\bm{X}^k{\cB^*}) - \hat{\Sigma}(\bm{X}^k{\cB^*})
                        \right)
                        \begin{pmatrix}
                            \bm{X}_{\cM}^T
                            \\
                            (\bm{X}^k_{S\cup \cK})^T
                        \end{pmatrix}^T
                \right\Vert_2
                & =
                O_p\left(
                    \frac{n}{\sqrt{s+m+K}}
                    \cdot
                    \sqrt{\frac{{\log(s+m+K)}}{{\log n}}}
                \right).
            \end{align*}
        This implies that
        \begin{align*}
            & ~
            \left\vert
            \frac{1}{n}
                 \hat{\Delta}_{a,0}^T
                \sum_{k=1}^{K}w_k 
                        \begin{pmatrix}
                            \bm{X}_{\cM}^T
                            \\
                            (\bm{X}^k_{S\cup \cK})^T
                        \end{pmatrix}
                        (
                        \hat{\Sigma}(\bm{X}^k{{\cB^*}})
                        -
                        {\Sigma}(\bm{X}^k{{\cB^*}})
                        )
                        \begin{pmatrix}
                            \bm{X}_{\cM}^T
                            \\
                            (\bm{X}^k_{S\cup \cK})^T
                        \end{pmatrix}^T
                         \hat{\Delta}_{a,0}
            \right\vert 
            \\
            \leq
            & ~
            \frac{1}{n}
            \left\Vert
                \sum_{k=1}^{K}w_k 
                        \begin{pmatrix}
                            \bm{X}_{\cM}^T
                            \\
                            (\bm{X}^k_{S\cup \cK})^T
                        \end{pmatrix}
                        \left(
                        {\Sigma}(\bm{X}^k{\cB^*}) - \hat{\Sigma}(\bm{X}^k{\cB^*})
                        \right)
                        \begin{pmatrix}
                            \bm{X}_{\cM}^T
                            \\
                            (\bm{X}^k_{S\cup \cK})^T
                        \end{pmatrix}^T
                \right\Vert_2
            \cdot
            \left\Vert
                 \hat{\Delta}_{a,0}
            \right\Vert_2^2
            \\
            =
            &~
            O_p\left(
                    \frac{1}{\sqrt{s+m+K}}
                    \cdot
                    \sqrt{\frac{{\log(s+m+K)}}{{\log n}}}
                \right)
            \cdot
            O_p\left(
            {\frac{r}{n}}
            \right)
            \\
            =
            & ~
            O_p\left(
            \frac{\sqrt{r}}{n}
            \cdot 
            \sqrt{\frac{r}{s+m+K}}
            \sqrt{\frac{{\log(s+m+K)}}{{\log n}}}
            \right)
            =
            o_p\left(
            \frac{\sqrt{r}}{n}
            \right).
        \end{align*}
        where the last equality is because $r\leq s+m+K$ and $s+m+K=o(\sqrt{n})$. Combining \eqref{eq:*step3-2.5}, \eqref{eq:*step3-2.75}, \eqref{eq:*step3-3}, and \eqref{eq:*step3-4}, we have under the event $\left\{\hat{\cB}_{a,(\cM\cup S\cup\cK)^c} = \hat{\cB}_{0,(\cM\cup S\cup\cK)^c}=\bm{0}\right\}$,
        \begin{align}
            n
            \left[ 
            M_n(\hat{\cB}_{a})
            -
            M_n(\hat{\cB}_{0})
            \right]
            &
            =
            \begin{pmatrix}
                            \hat{\cB}_{a,\cM} - \hat{\cB}_{0,\cM} \\
                            \hat{\cB}_{a,S\cup\cK} - \hat{\cB}_{0,S\cup\cK}
            \end{pmatrix}^T
            \sum_{k=1}^{K}w_k
                    \begin{pmatrix}
                            \bm{X}_{\cM}^T
                            \\
                            (\bm{X}^k_{S\cup \cK})^T
                        \end{pmatrix}
                        \bm{H}(\bm{X}^k{\hat{\cB}_a})
                        \left\{
                            \bm{Y}^k - \bm{\mu}(\bm{X}^k{\hat{\cB}_a})
                        \right\}
                        \nonumber
            \\
            & 
            +
            \frac{n}{2}
            \begin{pmatrix}
                            \hat{\cB}_{a,\cM} - \hat{\cB}_{0,\cM} \\
                            \hat{\cB}_{a,S\cup\cK} - \hat{\cB}_{0,S\cup\cK}
                \end{pmatrix}^T
            \bm{K}_n
                        \begin{pmatrix}
                            \hat{\cB}_{a,\cM} - \hat{\cB}_{0,\cM} \\
                            \hat{\cB}_{a,S\cup\cK} - \hat{\cB}_{0,S\cup\cK}
                        \end{pmatrix}
                \nonumber
            \\
            &
            +
            o_p\left(
            \sqrt{r}
            \right).
            \label{eq:*step3-4.5}
        \end{align}
        Next show that 
        \begin{equation}
            \left\Vert
                \sum_{k=1}^{K}w_k
                    \begin{pmatrix}
                            \bm{X}_{\cM}^T
                            \\
                            (\bm{X}^k_{S\cup \cK})^T
                        \end{pmatrix}
                        \bm{H}(\bm{X}^k{\hat{\cB}_a})
                        \left\{
                            \bm{Y}^k - \bm{\mu}(\bm{X}^k{\hat{\cB}_a})
                        \right\}
            \right\Vert_2
            =o_p(\sqrt{n})
            \label{eq:*step3-5}
        \end{equation}
        By Theorem 2.1, we have with probability tending to 1 that $\min_{j\in S}\vert\hat{\cB}_{a,j}\vert\geq d_n$. By the first-order optimality condition of $\hat{\cB}_a$, we have
        \begin{equation*}
            \sum_{k=1}^{K}w_k
                    \begin{pmatrix}
                            \bm{X}_{\cM}^T
                            \\
                            (\bm{X}^k_{S\cup \cK})^T
                        \end{pmatrix}
                        \bm{H}(\bm{X}^k{\hat{\cB}_a})
                        \left\{
                            \bm{Y}^k - \bm{\mu}(\bm{X}^k{\hat{\cB}_a})
                        \right\}
            =
            \begin{pmatrix}
                    \bm{0}_{m}
                    \\
                    n\lambda_{n,a}\bar{\rho}(\hat{\cB}_{a,S)},\lambda_{n,a})
                    \\
                    \bm{0}_{K}
            \end{pmatrix}.
        \end{equation*}
        Under the \textcolor{black}{2nd condition of (A2)} that $p^{\prime}_{\lambda_{n,a}}(d_n)=o(1/\sqrt{(s+m)n})$, we have
        \begin{equation*}
            \Vert
                n\lambda_{n,a}\bar{\rho}(\hat{\cB}_{a,S\cup\cK)},\lambda_{n,a})
            \Vert_2
            \leq 
            n\sqrt{s}
            \vert 
            p^{\prime}_{\lambda_{n,a}}(d_n)
            \vert 
            =
            o(\sqrt{n}).
        \end{equation*}
        This proves equation \eqref{eq:*step3-5} as desired. Along with \eqref{eq:*step3-2}, we have
        \begin{equation}
            \left\vert
            \begin{pmatrix}
                            \hat{\cB}_{a,\cM} - \hat{\cB}_{0,\cM} \\
                            \hat{\cB}_{a,S\cup\cK} - \hat{\cB}_{0,S\cup\cK}
            \end{pmatrix}^T
            \sum_{k=1}^{K}w_k
                    \begin{pmatrix}
                            \bm{X}_{\cM}^T
                            \\
                            (\bm{X}^k_{S\cup \cK})^T
                        \end{pmatrix}
                        \bm{H}(\bm{X}^k{\hat{\cB}_a})
                        \left\{
                            \bm{Y}^k - \bm{\mu}(\bm{X}^k{\hat{\cB}_a})
                        \right\}
            \right\vert 
            =
            o_p(\sqrt{r}).
        \end{equation}
        As a result, we have under the event $\left\{\hat{\cB}_{a,(\cM\cup S\cup\cK)^c} = \hat{\cB}_{0,(\cM\cup S\cup\cK)^c}=\bm{0}\right\}$,
        \begin{align}
            2n
            \left[ 
            M_n(\hat{\cB}_{a})
            -
            M_n(\hat{\cB}_{0})
            \right]
            &
            =
            {n}
            \begin{pmatrix}
                            \hat{\cB}_{a,\cM} - \hat{\cB}_{0,\cM} \\
                            \hat{\cB}_{a,S\cup\cK} - \hat{\cB}_{0,S\cup\cK}
                \end{pmatrix}^T
            \bm{K}_n
                        \begin{pmatrix}
                            \hat{\cB}_{a,\cM} - \hat{\cB}_{0,\cM} \\
                            \hat{\cB}_{a,S\cup\cK} - \hat{\cB}_{0,S\cup\cK}
                        \end{pmatrix}
            +
            o_p\left(
            \sqrt{r}
            \right).
            \label{eq:*step3-6}
        \end{align}
        From the fact that $\left\Vert
            {(\bm{C}\Omega_{mm}\bm{C}^T)^{-1/2}}
                (\bm{\omega_n}
                +
                \sqrt{n}
                \bm{h}_n
                )
            \right\Vert_2=O_p(\sqrt{r})$ and $\lambda_{\max}(\bm{K}_n)=O(1)$, we have
        \begin{align*}
            &~
            \left\Vert
            \begin{pmatrix}
                    \bC^T 
                            \\
                            \bm{0}_{r\times (s+K)}^T
                \end{pmatrix}^T
            {(\bm{C}\Omega_{mm}\bm{C}^T)^{-1}}
                (\bm{\omega_n}
                +
                \sqrt{n}
                \bm{h}_n
                )
            \right\Vert_2^2
            \\
            =
            &~
            \left\Vert
            \bm{K}_n^{1/2}
            \bm{K}_n^{-1/2}
            \begin{pmatrix}
                    \bC^T 
                            \\
                            \bm{0}_{r\times (s+K)}^T
                \end{pmatrix}^T
            {(\bm{C}\Omega_{mm}\bm{C}^T)^{-1}}
                (\bm{\omega_n}
                +
                \sqrt{n}
                \bm{h}_n
                )
            \right\Vert_2^2
            \\
            \leq&~
            \Vert
            \bm{K}_n^{1/2}
            \Vert_2^2
            \left\Vert
            \bm{K}_n^{-1/2}
            \begin{pmatrix}
                    \bC^T 
                            \\
                            \bm{0}_{r\times (s+K)}^T
                \end{pmatrix}^T
            {(\bm{C}\Omega_{mm}\bm{C}^T)^{-1}}
                (\bm{\omega_n}
                +
                \sqrt{n}
                \bm{h}_n
                )
            \right\Vert_2^2
            \\
            =&~
            \Vert
            \bm{K}_n
            \Vert_2
            \left\Vert
            {(\bm{C}\Omega_{mm}\bm{C}^T)^{-1/2}}
                (\bm{\omega_n}
                +
                \sqrt{n}
                \bm{h}_n
                )
            \right\Vert_2^2=O_p(r)
        \end{align*}
        In view of equation \eqref{eq:*step3-1}, we have
        \begin{align*}
            &~
            \left\vert
            {n}
            \begin{pmatrix}
                            \hat{\cB}_{a,\cM} - \hat{\cB}_{0,\cM} \\
                            \hat{\cB}_{a,S\cup\cK} - \hat{\cB}_{0,S\cup\cK}
                \end{pmatrix}^T
            \bm{K}_n
                        \begin{pmatrix}
                            \hat{\cB}_{a,\cM} - \hat{\cB}_{0,\cM} \\
                            \hat{\cB}_{a,S\cup\cK} - \hat{\cB}_{0,S\cup\cK}
                        \end{pmatrix}
            -
            (\bm{\omega_n}
                +
                \sqrt{n}
                \bm{h}_n
                )^T
                {(\bm{C}\Omega_{mm}\bm{C}^T)^{-1}}
                (\bm{\omega_n}
                +
                \sqrt{n}
                \bm{h}_n
                )
            \right\vert
            \\
            =&~
            o_p\left(
            (\bm{\omega_n}
                +
                \sqrt{n}
                \bm{h}_n
                )^T
            {(\bm{C}\Omega_{mm}\bm{C}^T)^{-1}}
            \begin{pmatrix}
                    \bC^T 
                            \\
                            \bm{0}_{r\times (s+K)}^T
                \end{pmatrix}
                \right)
            +
            o_p(\bm{K}_n)
            \\
            =&~
            o_p(\sqrt{r}).
        \end{align*}
        Along with equation \eqref{eq:*step3-6} and the fact that $\mathbb{P}(\hat{\cB}_{a,(\cM\cup S\cup\cK)^c} = \hat{\cB}_{0,(\cM\cup S\cup\cK)^c}=\bm{0})\rightarrow 1$ proved in Theorem \ref{theorem:estimator_statistical_properties}, we can prove that
        \begin{equation*}
            T_L=T_0+o_p(\sqrt{r})=T_0+o_p(r).
        \end{equation*}

    \paragraph{Step 4: Find the asymptotic distribution of $T_0$.}

        Define $\mathcal{T}_n = \text{Cov}(\bm{\omega}_n)$, we then prove that
        \begin{equation}
            \sup_{\mathcal{A}}
            \left\vert
            \mathbb{P}
            \left(
                \mathcal{T}_n^{-1/2}\bm{\omega}_n\in\mathcal{A}
            \right)
            -
            \mathbb{P}
            \left(
                \bm{Z}\in\mathcal{A}
            \right)
            \right\vert
            \rightarrow 0,
        \end{equation}
        where $\bm{Z}\in\mathbb{R}^r$ stands for a mean zero Gaussian random vector with identity covariance matrix, and the supremum is taken over all convex sets $\mathcal{A}\subseteq\mathbb{R}^r$. 
        By the definition of $\bm{\omega}_n$, we have
        \begin{align*}
            \mathcal{T}_n^{-1/2}\bm{\omega}_n
            &=
            \frac{1}{\sqrt{n}}
                \mathcal{T}_n^{-1/2}
                \begin{pmatrix}
                    \bC^T 
                            \\
                            \bm{0}_{r\times (s+K)}^T
                \end{pmatrix}^T
                \bm{K}_n^{-1}
                \sum_{k=1}^{K}w_k
                    \begin{pmatrix}
                            \bm{X}_{\cM}^T
                            \\
                            (\bm{X}^k_{S\cup \cK})^T
                        \end{pmatrix}
                        \bm{H}(\bm{X}^k{{\cB^*}})
                        \left\{
                            \bm{Y}^k - \bm{\mu}(\bm{X}^k{{\cB^*}})
                        \right\}
            \\
            &=\sum_{i=1}^{n}
            \frac{1}{\sqrt{n}}
            \mathcal{T}_n^{-1/2}
            \begin{pmatrix}
                    \bC^T 
                            \\
                            \bm{0}_{r\times (s+K)}^T
                \end{pmatrix}^T
            \bm{K}_n^{-1}
            \sum_{k=1}^{K}w_k
            h^{\prime}((\bx_{i}^{k})^T{\cB^*})
            \left\{
                \tilde{y}_{ik} - \mu((\bx_{i}^{k})^T{\cB^*})
            \right\}
            \begin{pmatrix}
                            \bx_{i,\cM}
                            \\
                            \bx^k_{i,S\cup \cK}
                        \end{pmatrix}
            \\
            &=:\sum_{i=1}^{n}\xi_i,
        \end{align*}
        where $\xi_i$'s are independent. Since $\mathcal{T}_n = \text{Cov}(\bm{\omega}_n)$, we have $\sum_{i=1}^{n}\text{Cov}(\xi_i)=I_r$. By Theorem 1 in \cite{bentkus2005lyapunov}, we have
        \begin{equation*}
            \sup_{\mathcal{A}}
            \left\vert
            \mathbb{P}
            \left(
                \mathcal{T}_n^{-1/2}\bm{\omega}_n\in\mathcal{A}
            \right)
            -
            \mathbb{P}
            \left(
                \bm{Z}\in\mathcal{A}
            \right)
            \right\vert
            \leq
            c_0
            r^{1/4}
            \sum_{i=1}^{n}\mathbb{E}\Vert\xi_i\Vert_2^3,
        \end{equation*}
        for some constant $c_0$, where the supremum is taken for all convex subsets in $\mathbb{R}^{r}$. To bound $\sum_{i=1}^{n}\mathbb{E}\Vert\xi_i\Vert_2^3$, notice that by convexity of $\Vert\cdot\Vert_2^3$,
        \begin{align*}
            \sum_{i=1}^{n}\mathbb{E}\Vert\xi_i\Vert_2^3
            &=
            \frac{1}{n^{3/2}}
            \sum_{i=1}^{n}
            \mathbb{E}
            \left\Vert
            \sum_{k=1}^{K}w_k
            \mathcal{T}_n^{-1/2}
            \begin{pmatrix}
                    \bC^T 
                            \\
                            \bm{0}_{r\times (s+K)}^T
                \end{pmatrix}^T
            \bm{K}_n^{-1}
            h^{\prime}((\bx_{i}^{k})^T{\cB^*})
            \left\{
                \tilde{y}_{ik} - \mu((\bx_{i}^{k})^T{\cB^*})
            \right\}
            \begin{pmatrix}
                            \bx_{i,\cM}
                            \\
                            \bx^k_{i,S\cup \cK}
                        \end{pmatrix}
            \right\Vert_2^3
            \\
            &\leq
            \frac{1}{n^{3/2}}
            \sum_{i=1}^{n}
            \sum_{k=1}^{K}w_k
            \mathbb{E}
            \left\Vert
            \mathcal{T}_n^{-1/2}
            \begin{pmatrix}
                    \bC^T 
                            \\
                            \bm{0}_{r\times (s+K)}^T
                \end{pmatrix}^T
            \bm{K}_n^{-1}
            h^{\prime}((\bx_{i}^{k})^T{\cB^*})
            \left\{
                \tilde{y}_{ik} - \mu((\bx_{i}^{k})^T{\cB^*})
            \right\}
            \begin{pmatrix}
                            \bx_{i,\cM}
                            \\
                            \bx^k_{i,S\cup \cK}
                        \end{pmatrix}
            \right\Vert_2^3
            \\
            &=
            \frac{1}{n^{3/2}}
            \sum_{i=1}^{n}
            \sum_{k=1}^{K}w_k
            \left\Vert
            \mathcal{T}_n^{-1/2}
            \begin{pmatrix}
                    \bC^T 
                            \\
                            \bm{0}_{r\times (s+K)}^T
                \end{pmatrix}^T
            \bm{K}_n^{-1}
            \begin{pmatrix}
                            \bx_{i,\cM}
                            \\
                            \bx^k_{i,S\cup \cK}
                        \end{pmatrix}
            \right\Vert_2^3
            \left\vert
            h^{\prime}((\bx_{i}^{k})^T{\cB^*})
            \right\vert^3
            \mathbb{E}
            \left\vert
                \tilde{y}_{ik} - \mu((\bx_{i}^{k})^T{\cB^*})
            \right\vert^3.
        \end{align*}
        Since $\left\vert \tilde{y}_{ik} - \mu((\bx_{i}^{k})^T{\cB^*}) \right\vert\leq 1$, we have $\mathbb{E}
            \left\vert
                \tilde{y}_{ik} - \mu((\bx_{i}^{k})^T{\cB^*})
            \right\vert^3=O(1)$. 
        By the \textcolor{black}{3rd condition of (A1)} and $\vert h^{\prime\prime}\vert\leq1$, we have $\vert h^{\prime}((\bx_{i}^{k})^T{\cB^*})\vert=O(1)$. Therefore,
        \begin{align*}
            \sum_{i=1}^{n}\mathbb{E}\Vert\xi_i\Vert_2^3
            &\leq
            \frac{1}{n^{3/2}}
            \sum_{i=1}^{n}
            \sum_{k=1}^{K}w_k
            \left\Vert
            \mathcal{T}_n^{-1/2}
            \begin{pmatrix}
                    \bC^T 
                            \\
                            \bm{0}_{r\times (s+K)}^T
                \end{pmatrix}^T
            \bm{K}_n^{-1}
            \begin{pmatrix}
                            \bx_{i,\cM}
                            \\
                            \bx^k_{i,S\cup \cK}
                        \end{pmatrix}
            \right\Vert_2^3
            O(1)
            \\
            &\leq
            \frac{1}{n^{3/2}}
            \sum_{i=1}^{n}
            \sum_{k=1}^{K}w_k
            \left\Vert
            \mathcal{T}_n^{-1/2}
            \begin{pmatrix}
                    \bC^T 
                            \\
                            \bm{0}_{r\times (s+K)}^T
                \end{pmatrix}^T
            \bm{K}_n^{-1}
            \bm{V}_n^{1/2}
            \right\Vert_2^3
            \left\Vert
            \bm{V}_n^{-1/2}
            \begin{pmatrix}
                            \bx_{i,\cM}
                            \\
                            \bx^k_{i,S\cup \cK}
                        \end{pmatrix}
            \right\Vert_2^3
            O(1).
        \end{align*}
        Notice that 
        \begin{align*}
            \left\Vert
            \mathcal{T}_n^{-1/2}
            \begin{pmatrix}
                    \bC^T 
                            \\
                            \bm{0}_{r\times (s+K)}^T
                \end{pmatrix}^T
            \bm{K}_n^{-1}
            \bm{V}_n^{1/2}
            \right\Vert_2
            &=
            \lambda_{\max}
            \left(
            \mathcal{T}_n^{-1/2}
            \begin{pmatrix}
                    \bC^T 
                            \\
                            \bm{0}_{r\times (s+K)}^T
                \end{pmatrix}^T
            \bm{K}_n^{-1}
            \bm{V}_n
            \bm{K}_n^{-1}
            \begin{pmatrix}
                    \bC^T 
                            \\
                            \bm{0}_{r\times (s+K)}^T
                \end{pmatrix}
            \mathcal{T}_n^{-1/2}
            \right)^{1/2}
            =1,
        \end{align*}
        this along with the additional condition \eqref{eq:theorem(limiting_distribution)condition} in Theorem \ref{theorem:testing_statistic_distribution}:
        \begin{equation*}
            \frac{r^{1/4}}{n^{3/2}}\sum_{i=1}^{n}\sum_{k=1}^{K}
            w_k
            \left\{
                (\bm{x}_{i,\cM\cup{S}\cup\cK}^k)^T
                \bm{V}_n^{-1}
                (\bm{x}_{i,\cM\cup{S}\cup\cK}^k)
            \right\}^{3/2}
            \rightarrow 0,
        \end{equation*}
        implies that
        \begin{equation*}
            \sup_{\mathcal{A}}
            \left\vert
            \mathbb{P}
            \left(
                \mathcal{T}_n^{-1/2}\bm{\omega}_n\in\mathcal{A}
            \right)
            -
            \mathbb{P}
            \left(
                \bm{Z}\in\mathcal{A}
            \right)
            \right\vert
            \rightarrow 0.
        \end{equation*}
        
        Since for any $x\in\mathbb{R}^+$, the event $\{T_0\leq x\}$ is equivalent to the event $\{\mathcal{T}_n^{-1/2}\bm{\omega}_n\in\mathcal{A}_x\}$, where 
        \begin{equation*}
            \mathcal{A}_x 
            =
            \left\{
            \bm{z}\in\mathbb{R}^r:
            \left\Vert
                \Psi_n^{-1/2}
                \left(
                {\mathcal{T}_n}^{1/2} \cdot \bm{z}
                +
                \sqrt{n}\bm{h}_n
                \right)
            \right\Vert_2^2
            \leq x
            \right\}.
        \end{equation*}
        It follows that
        \begin{equation}
            \sup_{x}
            \left\vert
            \mathbb{P}
            \left(
                T_0\leq x
            \right)
            -
            \mathbb{P}
            \left(
                \bm{Z}\in\mathcal{A}_x
            \right)
            \right\vert
            \rightarrow 0,
            \label{eq:*step4-T0_result}
        \end{equation}
        where 
        $\left\Vert\Psi_n^{-1/2}
                \left(
                {\mathcal{T}_n}^{1/2} \cdot \bm{Z}
                +
                \sqrt{n}\bm{h}_n
                \right)
            \right\Vert_2^2$ follows a generalized Chi-squared distribution.

    Consider any statistic $T=T_0+o_p(r)$. For any $x$ and any sufficiently small $\varepsilon>0$, we have $\mathbb{P}(\vert T-T_0\vert > \varepsilon r)=o(1)$. It follows from \eqref{eq:*step4-T0_result} that
    \begin{equation}
        \begin{aligned}
        \mathbb{P}(T\leq x)
        &
        \leq
        \mathbb{P}(T_0\leq x+\varepsilon r) + \mathbb{P}(\vert T-T_0\vert > \varepsilon r)
        \leq
        \mathbb{P}
            \left(
                \bm{Z}\in\mathcal{A}_{x+\varepsilon r}
            \right)
        +o(1)
        \\
        \mathbb{P}(T\leq x)
        &
        \geq
        \mathbb{P}(T_0\leq x-\varepsilon r) - \mathbb{P}(\vert T-T_0\vert > \varepsilon r)
        \geq
        \mathbb{P}
            \left(
                \bm{Z}\in\mathcal{A}_{x-\varepsilon r}
            \right)-o(1).
        \end{aligned}
        \label{eq:*step4-T_bound_by_T0}
    \end{equation}
    It remains to prove that
    \begin{equation}
        \lim_{\varepsilon\rightarrow 0^+}{\lim\sup}_{n}
        \left\vert
        \mathbb{P}
            \left(
                \bm{Z}\in\mathcal{A}_{x+\varepsilon r}
            \right)
        -
        \mathbb{P}
            \left(
                \bm{Z}\in\mathcal{A}_{x-\varepsilon r}
            \right)
        \right\vert
        = 0
        \label{eq:*step4-gen_chisq_result}
    \end{equation}
    or equivalently,
    \begin{equation*}
        \lim_{\varepsilon\rightarrow 0^+}{\lim\sup}_{n}
        \mathbb{P}
            \left(
            x-\varepsilon r
            \leq
                \left(
                \bm{Z}
                +
                \sqrt{n}{\mathcal{T}_n}^{-1/2} \bm{h}_n
                \right)^T
                {\mathcal{T}_n}^{1/2}
                \Psi_n^{-1}
                {\mathcal{T}_n}^{1/2}
                \left(
                \bm{Z}
                +
                \sqrt{n}{\mathcal{T}_n}^{-1/2} \bm{h}_n
                \right)
            \leq 
            x+\varepsilon r
            \right)
        = 0
    \end{equation*}
    This result can be proved using the fact that
    $
    \Vert
        \sqrt{n}{\cT}_{n}^{-1/2}\bh_n
    \Vert_2
    =O(\sqrt{r})
    $
    ,
    $\lambda_{\max}(
                {\mathcal{T}_n}^{1/2}
                \Psi_n^{-1}
                {\mathcal{T}_n}^{1/2})=O(1)$ 
    as well as 
    $\lim\sup_{n}\lambda_{\min}(
                {\mathcal{T}_n}^{1/2}
                \Psi_n^{-1}
                {\mathcal{T}_n}^{1/2})=O(1)$
    given in results \eqref{eq:lemma_result6+}, \eqref{eq:lemma_result6++} and \eqref{eq:lemma_result7} of Lemma \ref{lemma1}. Combining \eqref{eq:*step4-T_bound_by_T0} with \eqref{eq:*step4-gen_chisq_result}, we obtain that
    \begin{equation}
            \sup_{x}
            \left\vert
            \mathbb{P}
            \left(
                T\leq x
            \right)
            -
            \mathbb{P}
            \left(
                \bm{Z}\in\mathcal{A}_x
            \right)
            \right\vert
            \rightarrow 0.
            \label{eq:*step4-T_result}
    \end{equation}
    In the first three steps, we have shown $T_0 = T_W + o_p(1) = T_S + o_p(1) = T_L + o_p(1)$. This together with \eqref{eq:*step4-T_result} implies that the generalized $\chi^2$ approximation holds for our partial penalized statistics. Hence, the proof is completed.

\end{proof}

\subsection{Proof of Theorem \ref{theorem:type_I_error}}
From the results
\eqref{eq:*step4-gen_chisq_result} proved in Theorem \ref{theorem:testing_statistic_distribution}, it suffices to prove that
\begin{equation}
        \mathbb{P}
            \left(
                \bm{Z}^T
                {\mathcal{T}_n}^{1/2}
                \Psi_n^{-1}
                {\mathcal{T}_n}^{1/2}
                \bm{Z}
            \leq
            \hat{\chi}_{n,(1-\alpha)}
            |\bm{X},\bm{Y}\right)
    = 1-\alpha + o_p(1),
    \label{eq:thm3-result1}
\end{equation}
For any fixed $\varepsilon>0$, we have the following inequality
\begin{align*}
    &~
    \mathbb{P}
            \left(
                \bm{Z}^T
                {\widehat{\mathcal{T}}_{n,a}}^{1/2}
                {\widehat{\Psi}_{n,a}}^{-1}
                {\widehat{\mathcal{T}}_{n,a}}^{1/2}
                \bm{Z}
            \leq
            \hat{\chi}_{n,(1-\alpha)}-\varepsilon
            |\bm{X},\bm{Y}\right)
    -
    \mathbb{P}
            \left(
                \left\vert
                \bm{Z}^T
                (
                {\widehat{\mathcal{T}}_{n,a}}^{1/2}
                {\widehat{\Psi}_{n,a}}^{-1}
                {\widehat{\mathcal{T}}_{n,a}}^{1/2}
                -
                {\mathcal{T}_n}^{1/2}
                \Psi_n^{-1}
                {\mathcal{T}_n}^{1/2}
                )
                \bm{Z}
                \right\vert
            >
            \varepsilon
            |\bm{X},\bm{Y}\right)
    \\
    \leq&~
    \mathbb{P}
            \left(
                \bm{Z}^T
                {\mathcal{T}_n}^{1/2}
                \Psi_n^{-1}
                {\mathcal{T}_n}^{1/2}
                \bm{Z}
            \leq
            \hat{\chi}_{n,(1-\alpha)}
            |\bm{X},\bm{Y}\right)
    \\
    \leq&~
    \mathbb{P}
            \left(
                \bm{Z}^T
                {\widehat{\mathcal{T}}_{n,a}}^{1/2}
                {\widehat{\Psi}_{n,a}}^{-1}
                {\widehat{\mathcal{T}}_{n,a}}^{1/2}
                \bm{Z}
            \leq
            \hat{\chi}_{n,(1-\alpha)}+\varepsilon
            |\bm{X},\bm{Y}\right)
    +
    \mathbb{P}
            \left(
                \left\vert
                \bm{Z}^T
                (
                {\widehat{\mathcal{T}}_{n,a}}^{1/2}
                {\widehat{\Psi}_{n,a}}^{-1}
                {\widehat{\mathcal{T}}_{n,a}}^{1/2}
                -
                {\mathcal{T}_n}^{1/2}
                \Psi_n^{-1}
                {\mathcal{T}_n}^{1/2}
                )
                \bm{Z}
                \right\vert
            >
            \varepsilon
            |\bm{X},\bm{Y}\right).
\end{align*}
On the event
\begin{equation}
    \mathcal{E}_{n}
    :=
    \left\{
        \left\Vert                        {\widehat{\mathcal{T}}_{n,a}}^{1/2}
                {\widehat{\Psi}_{n,a}}^{-1}
                {\widehat{\mathcal{T}}_{n,a}}^{1/2} - 
                {{\mathcal{T}}_{n}}^{1/2}
                {\Psi}_{n}^{-1}
                {{\mathcal{T}}_{n}}^{1/2}
                \right\Vert_2
        \leq
        \frac{\varepsilon}{
        \frac{\varepsilon}{2}
        \frac{\sqrt{n}}{s+m+K}
        +r
        }
    \right\},
    \label{eq:thm3-event1}
\end{equation}
by the sub-exponential tail bound for $\chi^2(0,r)$, we have
\begin{align*}
    &~
    \mathbb{P}
            \left(
                \left\vert
                \bm{Z}^T
                (
                {\widehat{\mathcal{T}}_{n,a}}^{1/2}
                {\widehat{\Psi}_{n,a}}^{-1}
                {\widehat{\mathcal{T}}_{n,a}}^{1/2}
                -
                {\mathcal{T}_n}^{1/2}
                \Psi_n^{-1}
                {\mathcal{T}_n}^{1/2}
                )
                \bm{Z}
                \right\vert
            >
            \varepsilon
            |\bm{X},\bm{Y}\right)
    \\
    \leq&~
    \mathbb{P}
            \left(
                \left\Vert                        {\widehat{\mathcal{T}}_{n,a}}^{1/2}
                {\widehat{\Psi}_{n,a}}^{-1}
                {\widehat{\mathcal{T}}_{n,a}}^{1/2} - 
                {{\mathcal{T}}_{n}}^{1/2}
                {\Psi}_{n}^{-1}
                {{\mathcal{T}}_{n}}^{1/2}
                \right\Vert_2
                \bZ^T\bZ
            >
            \varepsilon
            |\bm{X},\bm{Y}\right)
    \\
    \leq&~
    \exp\left(
    -\frac{1}{8}
    \left(
    \frac{\varepsilon}{
    \Vert                        {\widehat{\mathcal{T}}_{n,a}}^{1/2}
                {\widehat{\Psi}_{n,a}}^{-1}
                {\widehat{\mathcal{T}}_{n,a}}^{1/2} - 
                {{\mathcal{T}}_{n}}^{1/2}
                {\Psi}_{n}^{-1}
                {{\mathcal{T}}_{n}}^{1/2}
    \Vert_2
    }
    -
    r
    \right)
    \right)
    \\
    \leq&~
    \exp\left(
    -\frac{\varepsilon}{16}\frac{\sqrt{n}}{s+m+K}
    \right),
\end{align*}
which could be arbitrarily small as long as $n$ is large enough. Furthermore, combining the fact that
$\Vert                        {\widehat{\mathcal{T}}_{n,a}}^{1/2}
                {\widehat{\Psi}_{n,a}}^{-1}
                {\widehat{\mathcal{T}}_{n,a}}^{1/2} - 
                {{\mathcal{T}}_{n}}^{1/2}
                {\Psi}_{n}^{-1}
                {{\mathcal{T}}_{n}}^{1/2}
                \Vert_2
        =O_p
            \left(
            {s+m+K}/{\sqrt{n}}
            \right)$
proved in Lemma \ref{lemma1} as well as the additional condition $(s+m+K)r=o(\sqrt{n})$, we know that with probability tending to $1$, the event $\mathcal{E}_{n}$ \eqref{eq:thm3-event1} holds. Therefore, the above ineuqality holds with probability tending to $1$. Then we have
\begin{align*}
    &~
    \mathbb{P}
            \left(
                \left\vert
                \bm{Z}^T
                (
                {\widehat{\mathcal{T}}_{n,a}}^{1/2}
                {\widehat{\Psi}_{n,a}}^{-1}
                {\widehat{\mathcal{T}}_{n,a}}^{1/2}
                -
                {\mathcal{T}_n}^{1/2}
                \Psi_n^{-1}
                {\mathcal{T}_n}^{1/2}
                )
                \bm{Z}
                \right\vert
            >
            \varepsilon
            \right)
    \\
    \leq&~
    \mathbb{E}
    \left[
        \bm{1}_{\mathcal{E}_{n}}
        \cdot
        \mathbb{P}
            \left(
                \left\vert
                \bm{Z}^T
                (
                {\widehat{\mathcal{T}}_{n,a}}^{1/2}
                {\widehat{\Psi}_{n,a}}^{-1}
                {\widehat{\mathcal{T}}_{n,a}}^{1/2}
                -
                {\mathcal{T}_n}^{1/2}
                \Psi_n^{-1}
                {\mathcal{T}_n}^{1/2}
                )
                \bm{Z}
                \right\vert
            >
            \varepsilon
            |\bm{X},\bm{Y}\right)
    \right]
    +
    \mathbb{P}(\mathcal{E}_{n}^{c})
    \\
    \leq&~
    \exp\left(
    -\frac{\varepsilon}{16}\frac{\sqrt{n}}{s+m+K}
    \right)
    +
    \mathbb{P}(\mathcal{E}_{n}^{c}),
\end{align*}
which implies that
\begin{equation}
    {\lim\sup}_{n}
    \mathbb{P}
            \left(
                \left\vert
                \bm{Z}^T
                (
                {\widehat{\mathcal{T}}_{n,a}}^{1/2}
                {\widehat{\Psi}_{n,a}}^{-1}
                {\widehat{\mathcal{T}}_{n,a}}^{1/2}
                -
                {\mathcal{T}_n}^{1/2}
                \Psi_n^{-1}
                {\mathcal{T}_n}^{1/2}
                )
                \bm{Z}
                \right\vert
            >
            \varepsilon
            \right)
    =0.
    \label{eq:thm3-result_part_I}
\end{equation}
Similar to the proof of \eqref{eq:*step4-gen_chisq_result} with $\bh_n=\bm{0}$ under the null hypothesis, we can prove that 
\begin{equation*}
    \lim_{\varepsilon\rightarrow 0^+}{\lim\sup}_{n}
        \mathbb{P}
            \left(
            x-\varepsilon r
            \leq
                \bm{Z}^T
                {\widehat{\mathcal{T}}_{n,a}}^{1/2}
                {\widehat{\Psi}_{n,a}}^{-1}
                {\widehat{\mathcal{T}}_{n,a}}^{1/2}
                \bm{Z}
            \leq 
            x+\varepsilon r
            |\bX,\bY \right)
        \overset{a.s.}{=} 0.
\end{equation*}
Combining the previous results with the fact that 
\begin{equation*}
    \mathbb{P}\left(\left.
                \bm{Z}^T
                {\widehat{\mathcal{T}}_{n,a}}^{1/2}
                {\widehat{\Psi}_{n,a}}^{-1}
                {\widehat{\mathcal{T}}_{n,a}}^{1/2} 
                \bm{Z}
            \leq 
            \hat{\chi}_{n,(1-\alpha)}
            \right\vert
            \bm{X},\bm{Y}
            \right)
            =
            1-\alpha,
\end{equation*}
given by the definition of $\hat{\chi}_{n,(1-\alpha)}$, we proved the desired result \eqref{eq:thm3-result1}.

\subsection{Analysis of the Supermarket Data}\label{sec:realdata}
In our empirical study, we use the supermarket data in \cite{lan2016testing}. This data set contains a total of $n=464$ daily records. For each record, the response is the number of customers and the predictors are the sales volumes of $p=6398$ products. Both the response and predictors are standardized so that they have zero mean and unit variance. This study identifies the products that attract the most customers and further performs hypothesis testings on those products to evaluate the impact of those products on the customers.

We apply the testing procedures of \cite{shi2019linear} under the standard normal linear model assumptions, and our testing procedures under non-parametric Box-Cox model to this dataset.
To formulate the testing hypotheses, we randomly split the data set into two subsets: one for preliminary analysis ($232$ observations) and one for hypothesis testings ($232$ observations).

\subsection{Computing the partial penalized CPR estimators $\hat{\bbt}_a$}\label{appendix_sec:compute_beta_a}
From \eqref{eq:beta_a}, $(\hat{\bbt}_a, \hat{\bb}_a)$ is the minimizer of the following objective function
\begin{equation}
    {\sum_{k=1}^{K}w_k\left\{
    \frac{1}{n}\sum_{i=1}^{n}L(\check{y}_{ki}(\bx_i^{T}\bbt-b_{k}))
    \right\}}
    + 
    \sum_{j \notin \cM}p_\lambda(\vert\beta_j\vert).
    \label{eq:partial_penalized_objective}
\end{equation}
It can be solved by combining the local linear approximation (LLA) algorithm and the coordinate-majorization-descent (CMD) algorithm. Details of the LLA-CMD algorithm are discussed below.

\paragraph{Outer loop: local linear approximation}
The local linear approximation (LLA) algorithm \citep{zou2008one} takes advantage of the special folded concave structure and utilizes the majorization-minimization (MM) principle to turn a concave regularization problem into a sequence of weighted $\ell_1$ penalized problems. 
Let $(\hat{\bbt}^\textit{curr},\hat{\bm{b}}^\textit{curr})$ be the current estimate. The folded concave penalty could be majorized by a local linear approximation function:
\begin{equation}
	\sum_{j\notin \cM}p_\lambda(\vert\beta_j\vert)\leq
	\sum_{j\notin \cM}p_\lambda(\vert\hat{\beta}_j^\textit{curr}\vert) + p_\lambda^{\prime}(\vert{\hat{\beta}_j^\textit{curr}}\vert)(\vert\beta_j\vert-\vert{\hat{\beta}_j^\textit{curr}}\vert),
\end{equation}
which is the best convex majorization of the concave penalty function (Theorem 2 of \citealt{zou2008one}). Then the objective function of \eqref{eq:partial_penalized_objective} could be majorized by a weighted $\ell_1$ penalized problem:
\begin{equation}
	M_n(\bbt,\bm{b})
        +
        \sum_{j \notin \cM}p_\lambda(\vert\hat{\beta}_j^\textit{curr}\vert) + p_\lambda^{\prime}(\vert{\hat{\beta}_j^\textit{curr}}\vert)(\vert\beta_j\vert-\vert{\hat{\beta}_j^\textit{curr}}\vert).
\end{equation}
The details of the LLA algorithm are summarized  in Algorithm \ref{alg:partial_penalized_LLA}.


\paragraph{Inner loop: coordinate-majorization-descent} 
For our weighted $\ell_1$ penalized composite probit regression problem \eqref{eq:cpr_weightedL1} within each LLA iteration in Algorithm \ref{alg:partial_penalized_LLA}, we may also apply the coordinate descent algorithm \citep{friedman2010regularization} which has been successfully used in solving some high-dimensional models. In the case of probit regression, we need to pay attention to computer overflow errors that may occur during the computation of the CDF $\Phi(\cdot)$. Standard algorithms like Newton-Raphson are very sensitive to large values of the linear predictor \citep{demidenko2001computational}. 

We prefer to use a numerically stable and efficient algorithm to solve \eqref{eq:cpr_weightedL1}. Due to the good property of probit regression loss given in Lemma 2 (\textcolor{black}{refer to the first paper}), that is, the second derivative of the probit regression loss function can be bounded by 1, we can fix the computer overflow error issue by using the coordinate-majorization-descent algorithm  \citep{yang2013efficient} which only uses the gradient information of the composite probit loss function.

Let $(\hat{\bbt}^\textit{curr},\hat{\bm{b}}^\textit{curr})$ be the current estimate. Define the current margin $\hat{r}_{ki}^\textit{curr}=\check{y}_{ki}(\bx_i^{T}{\hat{\bbt}^\textit{curr}}-\hat{{b}}_k^\textit{curr})$ for $k=1,\cdots,K$, $i=1,\cdots,n$  and current $\ell_1$ weights
$\hat{\omega}_j^\textit{curr}=p^\prime_{\lambda}(\vert\hat{\beta}^\textit{curr}_j\vert)$ for $j\in \{1,\cdots,p\}\backslash\cM$. 

To update the $\zeta$-th coordinate of $\bbt$, where $\zeta\in\{1,\cdots,p\}\backslash\cM$, define the $F$ function:
\begin{equation}
    F(\beta_{\zeta}\vert \hat{\bbt}^\textit{curr},\hat{\bm{b}}^\textit{curr})
    :=
    \sum_{k=1}^{K}w_k\left\{\frac{1}{n}\sum_{i=1}^{n}L(\check{y}_{ki}x_{i\zeta}(\beta_{\zeta}-{\hat{\beta}^\textit{curr}}_{\zeta}) + \hat{r}_{ki}^\textit{curr})\right\} + \hat{\omega}^\textit{curr}_{\zeta}\vert\beta_{\zeta}\vert.
\end{equation}
By Lemma \ref{lemma:2deriv_bounded} and $\check{y}_{ki}^2=1$, this $F$ function can be majorized by a penalized quadratic function defined as
\begin{equation*}
    \mathcal{F}(\beta_{\zeta}\vert\hat{\bbt}^\textit{curr},\hat{\bm{b}}^\textit{curr})
    :=
    \sum_{k=1}^{K}w_k\cdot\frac{1}{n}\sum_{i=1}^{n}\left[
	L(\hat{r}_{ki}^\textit{curr}) + L^\prime(\hat{r}_{ki}^\textit{curr})\check{y}_{ki}x_{i\zeta}(\beta_{\zeta}-{\hat{\beta}^\textit{curr}}_{\zeta}) + \frac{1}{2}x_{is}^2(\beta_{\zeta}-{\hat{\beta}^\textit{curr}}_{\zeta})^2
    \right] 
    + 
    \hat{\omega}^\textit{curr}_{\zeta}\vert\beta_{\zeta}\vert.
\end{equation*}
We can easily solve the minimizer of the majorization function by a simple soft thresholding rule:
\begin{equation}
    {\hat{\beta}^\textit{new}}_{\zeta}=\mathcal{S}\left(
	{\hat{\beta}^\textit{curr}}_{\zeta}+\frac{-\sum_{k=1}^{K}w_k\sum_{i=1}^{n}L^\prime(\hat{r}_{ki}^\textit{curr})\check{y}_{ki}x_{i\zeta}}{\sum_{i=1}^{n}x_{i\zeta}^2},\frac{\hat{\omega}^\textit{curr}_{\zeta}}{\frac{1}{n}\sum_{i=1}^{n}x_{i\zeta}^2}
	\right),
\end{equation}
where $\mathcal{S}(z,t)  = (\vert z\vert -t)_+\cdot\text{sgn}(z)$. The updating of the margins is given by
\begin{equation*}
    \hat{r}_{ki}^\textit{new} = \hat{r}_{ki}^\textit{curr} + \cky_{ki}x_{i\zeta}\cdot({\hat{\beta}^\textit{new}}_{\zeta} - {\hat{\beta}^\textit{curr}}_{\zeta}),~
    \text{for}~
    k=1,\dots,K,~
    i=1,\dots,n.
\end{equation*}

To update $\zeta$-th coordinate of $\bbt$, where $\zeta\in\cM$, we follow the same trick and the updating rule is as follows:
\begin{equation}
    \hat{\beta}_{\zeta}^\textit{new}=	{\hat{\beta}^\textit{curr}}_{\zeta} 
    +
    \frac{-\sum_{k=1}^{K}w_k\sum_{i=1}^{n}L^\prime(\hat{r}_{ki}^\textit{curr})\check{y}_{ki}x_{i\zeta}}{\sum_{i=1}^{n}x_{i\zeta}^2}.
    \label{eq:beta_update_wo_penalty}
\end{equation}
The updating of margins is also given by equation \eqref{eq:margin_update_from_beta}.

Simiarly, for the $\zeta$-th intercept $b_{\zeta}$, we consider minimizing the quadratic majorization:
\begin{equation}
	\mathcal{F}(b_{\zeta}\vert\hat{\bbt}^\textit{curr},\hat{\bm{b}}^\textit{curr})
        :=
        \frac{1}{n}\sum_{i=1}^{n}\left\{
	   L(\hat{r}_{{\zeta}i}^\textit{curr}) + L^\prime(\hat{r}_{{\zeta}i}^\textit{curr})(-\check{y}_{{\zeta}i})(b_{\zeta}-\hat{b}_{\zeta}^\textit{curr}) + \frac{1}{2}(b_{\zeta}-\hat{b}_{\zeta}^\textit{curr})^2
	\right\},
\end{equation}
which has a minimizer 
\begin{equation}
    \hat{b}_{\zeta}^\textit{new}=\hat{b}_{\zeta}^\textit{curr}+\frac{1}{n}\sum_{i=1}^{n}L^\prime(r_{{\zeta}i})\check{y}_{{\zeta}i}.
\end{equation}
The updating of margins is given by
\begin{equation}
    \hat{r}_{ki}^\textit{new} =
    \begin{cases}
        \hat{r}_{\zeta i}^\textit{curr} - \cky_{{\zeta}i}(\hat{b}_{\zeta}^\textit{new} - \hat{b}_{\zeta}^\textit{curr}),
        &
        \text{if } k=\zeta;
        \\
        \hat{r}_{ki}^\textit{curr},
        &
        \text{if } k\neq \zeta.
    \end{cases}
    \label{eq:margin_update_from_b}
\end{equation}

To sum up, the CMD algorithm for solving the weighted $\ell_1$-penalized CPR is given in Algorithm \ref{alg:GCD}.

\begin{algorithm}[ht!]
	\caption{The LLA Algorithm for solving unconstrained estimator}\label{alg:partial_penalized_LLA}
	\begin{algorithmic}[1]
		\State Initialize $\hat{\bbt}^{(0)}=\hat{\bbt}^{\text{initial}}$ and compute the adaptive weight
		\begin{equation*}
			\hat{\bm{\omega}}^{(0)}
                =
                \left(
                    \hat{\omega}_j^{(0)},~j\in \{1,\dots,p\}\backslash\cM 
                \right)
                =
                \left(
                    p_{\lambda}^\prime(|\hat{\beta}_j^{(0)}|),~j\in \{1,\cdots,p\}\backslash\cM
			\right)
		\end{equation*}
		\State For $m=1,2,\cdots,$ repeat the LLA iteration till convergence
		\begin{itemize}
			\item[(2.a)] Obtain $(\hat{\beta}^{(m)},\hat{\bm{b}}^{(m)})$ by solving the following optimization problem
			\begin{equation}
				(\hat{\bbt}^{(m)},\hat{\bm{b}}^{(m)})=\arg\min_{\bbt,\bm{b}}M_n(\bbt,\bm{b})+\sum_{j\notin\cM}\hat{\omega}_{j}^{(m-1)}\cdot |\beta_j|,
				\label{eq:cpr_weightedL1}
			\end{equation}
			\item[(2.b)] Update the adaptive weight vector $\hat{\bm{\omega}}^{(m)}$ with $\hat{\omega}_j^{(m)}=p_{\lambda}^\prime(|\hat{\beta}_j^{(m)}|)$.
		\end{itemize}
	\end{algorithmic}
\end{algorithm}

\begin{algorithm}[ht!]
	\caption{The CMD algorithm for solving \eqref{eq:cpr_weightedL1}.}\label{alg:GCD}
	\begin{algorithmic}[1]
		\State Input the weight vector $\hat{\bm{\omega}}$.
		\State Initialize $(\hat{\bbt}, \hat{\bm{b}})$ and its corresponding margin $\hat{r}_{ki}=\check{y}_{ki}(\bx_i^{T}\hat{\bbt}-\hat{b}_{k})$, for $k=1,\dots,K$, and $i=1,\dots,n$.
		\State Iterate (3.a)-(3.b) until convergence:
		\begin{itemize}
		    \item[(3.a)] Cyclic coordinate descent for coefficients: for $j = 1, 2, \cdots, p$,
		    \begin{itemize}
		        \item[(3.a.1)] Compute
			        \begin{equation*}
				        \hat{\beta}_j^{\textit{new}} = 
                                \begin{cases}
                                    \mathcal{S}\left(
	                        \hat{\beta}_{j}+\frac{-\sum_{k=1}^{K}w_k\sum_{i=1}^{n}L^\prime(\hat{r}_{ki})\check{y}_{ki}x_{ij}}{\sum_{i=1}^{n}x_{ij}^2},\frac{\tilde{\omega}_{j}}{\frac{1}{n}\sum_{i=1}^{n}x_{ij}^2}
	                        \right),
                                    &\text{if }j\in\{1,\dots,p\}\backslash\cM;
                                    \\
                                    \hat{\beta}_{j}+\frac{-\sum_{k=1}^{K}w_k\sum_{i=1}^{n}L^\prime(\hat{r}_{ki})\check{y}_{ki}x_{ij}}{\sum_{i=1}^{n}x_{ij}^2},
                                    &\text{if }j\in\cM.
                                \end{cases}
			        \end{equation*}
                    \item[(3.a.2)] Compute the new margins: $\hat{r}_{ki}^\textit{new} = \hat{r}_{ki} + \cky_{ki}x_{ij}\cdot({\hat{\beta}^\textit{new}}_{j} - \hat{\beta}_{j})$, 
                    for $k=1,\dots,K$, $i=1,\dots,n$.
			    \item[(3.a.3)] Update $\hat{\beta}_{j} = {\hat{\beta}_{j}}^\textit{new}$, and $\hat{r}_{ki} = \hat{r}_{ki}^\textit{new}$, for $k=1,\dots,K$, $i=1,\dots,n$.
		    \end{itemize}
		    \item[(3.b)] Cyclic coordinate descent for intercepts: for $k = 1, 2, \cdots, K$,
		    \begin{itemize}
		        \item[(3.b.1)] Compute
			        \begin{equation*}
				        \hat{b}_{k}^\textit{new}=\hat{b}_{k}+\frac{1}{n}\sum_{i=1}^{n}L^\prime(r_{ki})\check{y}_{ki}.
			        \end{equation*}
                    \item[(3.b.2)] Compute the new margins: $\hat{r}_{ki}^\textit{new} = \hat{r}_{ki} - \cky_{ki}(\hat{b}_{k}^\textit{new} - \hat{b}_k),$ for $i=1,\dots,n$.
			    \item[(3.b.3)] Update $\hat{b}_{k} = \hat{b}_{k}^\textit{new}$, and $\hat{r}_{ki}= \hat{r}_{ki}^\textit{new}$, for $i=1,\dots,n$.
		    \end{itemize}
		\end{itemize} 
	\end{algorithmic}
\end{algorithm}

\clearpage

\nocite*{}  

\end{document}